\documentclass[11pt,a4paper]{article}
\pdfoutput=1
\usepackage{graphicx}
\usepackage{natbib}
\usepackage{amsmath,amsfonts,amssymb,amsbsy,amsthm,mathtools}
\usepackage{bbm,bm}
\usepackage{color}
\usepackage{subfig,float}
\usepackage{booktabs}
\usepackage{url}
\usepackage[a4paper]{geometry}
\usepackage{a4wide} 
\usepackage[colorlinks=true,linkcolor=blue,citecolor=blue,pdfborder={0 0 0}]{hyperref}
\allowdisplaybreaks
\numberwithin{equation}{section}

\newtheorem{thm}{Theorem}[section]
\newtheorem{cor}[thm]{Corollary}

\theoremstyle{definition}
\newtheorem{defn}{Definition}[section]
\newtheorem{rem}[defn]{Remark}

\DeclareMathOperator*{\argmin}{arg\,min}
\newcommand{\diag}{\operatorname{diag}}
\newcommand{\RR}{\mathbb{R}}
\newcommand{\diff}{\, \mathrm{d}}
\newcommand{\PP}{\mathbb{P}}
\newcommand{\EE}{\mathbb{E}}

\newcommand{\QQ}{\mathcal{Q}}
\newcommand{\eqd}{\overset{\mathrm{d}}{=}}
\newcommand{\dto}{\xrightarrow{d}} 
\newcommand{\pto}{\xrightarrow{p}} 
\newcommand{\Hessian}{\mathcal{H}}
\newcommand{\htheta}{\widehat{\theta}_{n,k}}
\newcommand{\Htheta}{\widehat{\Theta}_{n,k}}
\newcommand{\Normal}{\mathcal{N}}

\newcommand{\abs}[1]{\lvert{#1}\rvert}
\newcommand{\norm}[1]{\lVert {#1} \rVert}
\newcommand{\point}{\,\cdot\,}
\newcommand{\Mopt}{M_{\mathrm{opt}}}

\begin{document}


\vspace{-2cm}
\newcommand{\Title}{\vspace{-1.5cm}
A continuous updating weighted least squares estimator of tail dependence in high dimensions}

\title{\Title}
  \author{John H.J. Einmahl\\
    \small Tilburg University, Department of Econometrics \& OR and CentER\\
    \small P.O. Box 90153, 5000~LE Tilburg, The Netherlands.\\
    \small E-mail: j.h.j.einmahl@uvt.nl\\
    \and
    Anna Kiriliouk\thanks{The research by A. Kiriliouk was funded by a FRIA grant of the ``Fonds de la Recherche Scientifique -- FNRS'' (Belgium).}\hspace{.2cm}
    \qquad
    Johan Segers\thanks{J. Segers gratefully acknowledges funding by contract ``Projet d'Act\-ions de Re\-cher\-che Concert\'ees'' No.\ 12/17-045 of the ``Communaut\'e fran\c{c}aise de Belgique'' and by IAP research network Grant P7/06 of the Belgian government (Belgian Science Policy).}\hspace{.2cm}\\
    \small Universit\'{e} catholique de Louvain\\
    \small Institut de Statistique, Biostatistique et Sciences Actuarielles\\
    \small Voie du Roman Pays~20, B-1348 Louvain-la-Neuve, Belgium.\\
    \small E-mail: anna.kiriliouk@uclouvain.be, johan.segers@uclouvain.be}
\date{}
\maketitle

\vspace{-0.5cm}
\begin{abstract}
Likelihood-based procedures are a common way to estimate tail dependence parameters. They are not applicable, however, in non-differentiable models such as those arising from recent max-linear structural equation models. Moreover, they can be hard to compute in higher dimensions. An adaptive weighted least-squares procedure matching nonparametric estimates of the stable tail dependence function with the corresponding values of a parametrically specified proposal yields a novel minimum-distance estimator. The estimator is easy to calculate and applies to a wide range of sampling schemes and tail dependence models. In large samples, it is asymptotically normal with an explicit and estimable covariance matrix. The minimum distance obtained forms the basis of a goodness-of-fit statistic whose asymptotic distribution is chi-square. Extensive Monte Carlo simulations confirm the excellent finite-sample performance of the estimator and demonstrate that it is a strong competitor to currently available methods. The estimator is then applied to disentangle sources of tail dependence in European stock markets. 
\end{abstract}

\noindent%
\vfill


\section{Introduction}

Extreme value analysis has been applied to measure and manage financial and actuarial risks, assess natural hazards stemming from heavy rainfall, wind storms, and earthquakes, and control processes in the food industry, internet traffic, aviation, and other branches of human activity. The extension from univariate to multivariate data gives rise to the concept of tail dependence. The latter can and will be represented here by the stable tail dependence function, denoted by $\ell$ \citep{huang1992, dreeshuang1998}, or tail dependence function for short. Estimating this tail dependence function is the subject of this paper. Fitting tail dependence models for spatial phenomena observed at finitely many sites constitutes an interesting special case.

In high(er) dimensions, the class of tail dependence functions becomes rather unwieldy, and therefore we follow the common route of modelling it parametrically. Note that this is far from imposing a fully parametric model on the data generating process. In particular, we only assume a domain-of-attraction condition at the copula level. Parametric models for tail dependence have their origins in \citet{gumbel1960}, and many models have since then been proposed, see, e.g., \citet{coles1991}, and more recently, \citet{kabluchko2009}.

Likelihood-based procedures are perhaps the most common way to estimate tail dependence parameters \citep{davison2012, wadsworth2014, huser2015}. Likelihood methods, however, are not applicable to models involving non-differentiable tail dependence functions. Such functions arise in max-linear models \citep{wang+s:2011}, in particular factor models \citep{einmahl2012} or structural equation models based on directed acyclic graphs \citep{gissibl2015}. Moreover, likelihoods can be hard to compute, especially in higher dimensions. This is why current likelihood methods are usually based on composite likelihoods, relying on pairs or triples of variables only, not exploiting information from higher-dimensional tuples.

It is the goal of this paper to estimate the true parameter vector $\theta_0$ of the tail dependence function $\ell$ and to assess the goodness-of-fit of the parametric model. The parameter estimator is obtained by comparing, at finitely many points in the domain of $\ell$, some initial, typically nonparametric, estimator of the latter with the corresponding values of the parametrically specified proposals, and retaining the parameter value yielding the best match. The method is generic in the sense that it applies to many parametric models, differentiable or not, and to many initial estimators, not only the usual empirical tail dependence function but also, for instance, bias-corrected versions thereof \citep{fougeres2015nr2, beirlant2015}. Further, the method avoids integration or differentiation of functions of many variables and can therefore handle joint dependence between many variables simultaneously, more easily than the likelihood methods mentioned earlier and the M-estimator approach in \citet{einmahl2016}. This feature is particularly interesting for inferring on higher-order interactions, going beyond mere distance-based dependence models such as those frequently employed for spatial extremes. Finally, in those situations where likelihood methods are applicable, the new estimator is a strong competitor.

The distance between the initial estimator and the parametric candidates is measured through weighted least squares. The weight matrix may depend on the unknown parameter $\theta$ and is hence estimated simultaneously. The construction of the estimator bears some similarity with the continuous updating generalized method of moments \citep{hansen1996}; the present estimator, however, is substantially different and does not use moments. Our flexible estimation procedure is related to that in \citet{einmahl2016}, but the continuous updating procedure is new in multivariate extreme value statistics. 

We show that the weighted least squares estimator for the tail dependence function is consistent and asymptotically normal, provided that the initial estimator enjoys these properties too, as is the case for the empirical tail dependence function and its recently proposed bias-corrected variations. The asymptotic covariance matrix is a function of the unknown parameter and can thus be estimated by a plug-in technique. We also provide novel goodness-of-fit tests for the parametric tail dependence model based on a comparison between the nonparametric and the parametric estimators. Under the null hypothesis that the tail dependence model is correctly specified, the test statistics are asymptotically chi-square distributed.

The paper is organized as follows. In Section~\ref{sec:estimator} we present the estimator, the goodness-of-fit statistic, and their asymptotic distributions. 
Section~\ref{sec:simulation} reports on a Monte Carlo simulation study involving a variety of models, as well as a finite-sample comparison of our estimator with estimators based on composite likelihoods. An application to European stock market data is presented in Section~\ref{sec:application}, where we try to disentangle sources of tail dependence stemming from the country of origin (Germany versus France) and the economic sector (chemicals versus insurance), fitting a structural equation model. All proofs are deferred to the appendix.

\section{Inference on tail dependence parameters}
\label{sec:estimator}

\subsection{Setup}

Let $X_i = (X_{i1},\ldots,X_{id})$, $i \in \{1,\ldots,n\}$, be random vectors in $\RR^d$ with a common cumulative distribution function $F$ and marginal cumulative distribution functions $F_1,\ldots,F_d$. The (stable) tail dependence function $\ell : [0, \infty)^d \to [0, \infty)$ is defined as
\begin{equation}\label{eq:ell2}
  \ell(x)
  := 
  \lim_{t \downarrow 0} t^{-1} \, 
  \PP[  1 - F_1(X_{11}) \leq t x_1 \text{ or } \ldots \text{ or } 1 - F_d (X_{1d}) \leq t x_d],
\end{equation}
for $x \in [0, \infty)^d$, provided the limit exists, as we will assume throughout. Existence of the limit is a necessary, but not sufficient, condition for $F$ to be in the max-domain of attraction of a $d$-variate Generalized Extreme Value distribution. Closely related to $\ell$ is the exponent measure function $V(z) = \ell(1/z_1, \ldots, 1/z_d)$, for $z \in (0, \infty]^d$. For more background on multivariate extreme value theory, see for instance \citet{beirlant2004} or \citet{dehaan2006}.

The function $\ell$ is convex, homogeneous of order one, and satisfies $\max(x_1, \ldots, x_d) \le \ell(x) \le x_1 + \cdots + x_d$ for all $x \in [0, \infty)^d$. If $d = 2$, these properties characterize the class of all $d$-variate tail dependence functions, but not if $d \ge 3$ \citep{ressel:2013}. For any dimension $d \ge 2$, the collection of $d$-variate tail dependence functions is infinite-dimensional. This poses challenges to inference on tail dependence, especially in higher dimensions.

The usual way of dealing with this problem consists of considering parametric models for $\ell$, a number of which are presented in Section~\ref{sec:simulation}. Henceforth we assume that $\ell$ belongs to a parametric family $\{\ell(\cdot \, ; \theta) : \theta \in \Theta \}$ with $\Theta \subset \RR^p$. Let $\theta_0$ denote the true parameter vector, that is, let $\theta_0$ denote the unique point in $\Theta$ such that $\ell(x) = \ell(x; \theta_0)$ for all $x \in [0,\infty)^d$. Our aim is to estimate the parameter $\theta_0$ and to test the goodness-of-fit of the model.

Extremal coefficients are popular summary measures of tail dependence \citep{dehaan1984, smith1990, schlather2003}. For non-empty $J \subset \{1,\ldots,d\}$, let $e_J \in \RR^d$  be defined by
\begin{equation}\label{ej}
  (e_J)_j 
  := 
  \begin{cases}
    1 & \text{ if $j \in J$,} \\
    0 & \text{ if $j \in \{ 1,\ldots,d \} \setminus J$.}
  \end{cases}
\end{equation}
The extremal coefficients are defined by 
\begin{equation}
\label{eq:extrCoeff}
  \ell_J 
 := 
  \ell(e_J)
  =
  \lim_{t \downarrow 0} t^{-1} \PP[ \max_{j \in J} F_j(X_{1j}) \ge 1 - t ]. 
\end{equation}
The extremal coefficients $\ell_J \in [1,\abs{J}]$ can be interpreted as assigning to each subset $J$ the effective number of tail independent variables among $(X_{1j})_{j \in J}$.

Comparing initial and parametric estimators of the extremal coefficients is a special case of the inference method that we propose. In fact, \citet{smith1990} already proposes an estimator based on pairwise ($\abs{J}=2$) extremal coefficients; see also \citet{dehaanpereira2006} and \citet{oesting2015}.

\subsection{Continuous updating weighted least squares estimator}
\label{sec:estimation}

Let $\widehat{\ell}_{n,k}$ denote an initial estimator of $\ell$ based on $X_1, \ldots, X_n$; some possibilities will be described in Subsection~\ref{sec:nonparametric}. The estimators $\widehat{\ell}_{n,k}$ that we will consider depend on an intermediate sequence $k = k_n \in (0,n]$, that is, 
\begin{equation}
\label{eq:intermed} 
  k \rightarrow \infty \quad \text{and } k/n \rightarrow 0, 
  \qquad \text{as $n \rightarrow \infty$.}
\end{equation}
The sequence $k$ will determine the tail fraction of the data that we will use for inference, see for instance Subsection~\ref{sec:nonparametric}.

Let $c_1,\ldots,c_q \in [0,\infty)^d$, with $c_m = (c_{m1},\ldots,c_{md})$ for $m = 1,\ldots,q$,  be $q$ points in which we will evaluate $\ell$ and $\widehat{\ell}_{n,k}$. Consider the $q \times 1$ column vectors
\begin{align}
\nonumber
  \widehat{L}_{n,k} 
  &:= 
  \bigl( \widehat{\ell}_{n,k} (c_m) \bigr)_{m = 1}^q, \\
\label{eq:L}
  L( \theta ) 
  &:= 
  \bigl( \ell (c_m ; \theta) \bigr)_{m=1}^q, \\
\label{eq:Dnk}
  D_{n,k} (\theta) 
  &:= 
  \widehat{L}_{n,k} - L(\theta),
\end{align}
where $\theta \in \Theta$. The points $c_1,\ldots,c_q$ need to be chosen in such a way that the map $L : \Theta \to \RR^q$ is one-to-one, i.e., $\theta$ is identifiable from the values of $\ell(c_1;\theta), \ldots, \ell(c_q;\theta)$. In particular, we will assume that $q \ge p$, where $p$ is the dimension of the parameter space $\Theta$. Since $\ell(c e_{\{j\}}) = c$ for any tail dependence function $\ell$, any $c \in [0, \infty)$ and any $j \in \{1, \ldots, d\}$, we will choose the points $c_m$ in such a way that each point has at least two positive coordinates.

For $\theta \in \Theta$, let $\Omega(\theta)$ be a symmetric, positive definite $q\times q$ matrix with ordered eigenvalues $0<\lambda_1(\theta)\leq \ldots \leq \lambda_q(\theta)$ and define
\begin{equation}
\label{eq:fnk}
  f_{n,k} (\theta) 
  := 
  \norm{ D_{n,k} (\theta) }^2_{\Omega(\theta)}
  := 
  D_{n,k}^T (\theta) \, \Omega (\theta) \, D_{n,k} (\theta).
\end{equation}
Our continuous updating weighted least squares estimator for $\theta_0$ is defined as
\begin{equation}\label{eq:estimator}
  \htheta 
  := 
  \argmin_{\theta \in \Theta} f_{n,k} (\theta) 
  = 
  \argmin_{\theta \in \Theta} \left\{ D_{n,k} (\theta)^T \, \Omega (\theta) \, D_{n,k} (\theta) \right\}.
\end{equation} 
The set of minimizers could be empty or could have more than one element. The present notation, suggesting that there exists a unique minimizer, will be justified in Theorem~\ref{thm:consistency}. If all points $c_m$ are chosen as $e_{J_m}$ in \eqref{ej} for some collection $J_1, \ldots, J_q$ of $q$ different subsets of $\{1, \ldots, d\}$, each subset having at least two elements, then we will refer to our estimator as an extremal coefficients estimator.

We will address the optimal choice of $\Omega (\theta)$ below. The simplest choice for $\Omega (\theta) $ is  the identity matrix $I_q$, yielding an ordinary least-squares estimator
\begin{equation}\label{armi}
\htheta := \argmin_{\theta \in \Theta} \sum_{m=1}^q
  \bigl(   \widehat{\ell}_{n,k} (c_m) - \ell (c_m; \theta) \bigr)^2.
\end{equation}
This special case of our estimator is similar to the estimator proposed in \citet{fougeres2015} in the more specific context of fitting max-stable distributions to a random sample from such a distribution. 


\subsection{Consistency and asymptotic normality}

If $L$ is differentiable at an interior point $\theta \in \Theta$, its total derivative will be denoted by $\dot{L} (\theta) \in \RR^{q \times p}$. Differentiability of the map $\theta \mapsto L(\theta)$ is a basic smoothness condition on the model; we do not assume differentiability of the map $x \mapsto \ell(x; \theta)$.

\begin{thm}[Existence, uniqueness and consistency]
\label{thm:consistency}
Let $\{ \ell(\point;\theta): \theta \in \Theta \}$, with $\Theta \subset \RR^p$, be a parametric family of $d$-variate stable tail dependence functions. Let $c_1,\ldots,c_q \in [0,\infty)^d$ be $q \ge p$ points
such that the map $L: \theta \mapsto (\ell (c_m ; \theta))_{m=1}^q$ is a homeomorphism from $\Theta$ to $L(\Theta)$. Let the true $d$-variate distribution function $F$ have  stable tail dependence function $\ell(\point;\theta_0)$ for some interior point $\theta_0 \in \Theta$.  
Assume that $L$ is twice continuously differentiable on a neighbourhood of $\theta_0$ and that $\dot{L} (\theta_0)$ is of full rank; also assume that $\Omega:\Theta \to \RR^{q \times q}$ is twice continuously differentiable on a neighbourhood of $\theta_0$. Assume $\lambda_1 := \inf_{\theta \in \Theta} \lambda_1(\theta)>0$. Finally assume, for $m=1, \ldots, q$, and for a positive sequence $k = k_n$ satisfying \eqref{eq:intermed},
\begin{equation} 
\label{ltol}
  \widehat{\ell}_{n,k}(c_m) \pto \ell (c_m ; \theta_0), 
  \qquad \text{as $n \to \infty$}.
\end{equation}
Then with probability tending to one, the minimizer $\htheta$ in \eqref{eq:estimator} exists and is unique. Moreover,
\begin{equation*}
  \htheta \pto \theta_0, 
  \qquad \text{as $n \rightarrow \infty$}.
\end{equation*}
\end{thm}

\begin{thm}[Asymptotic normality]
\label{thm:an}
If in addition to the assumptions of Theorem~\ref{thm:consistency}, the  estimator $\widehat{\ell}_{n,k}$ satisfies
\begin{equation}
\label{eq:nonparcond}
  \sqrt{k} \, D_{n,k}(\theta_0)
  =
  \left(
    \sqrt{k} \left\{ \widehat{\ell}_{n,k} (c_m) - \ell (c_m; \theta_0) \right\}
  \right)_{m=1}^q 
  \dto 
  \Normal_q \bigl( 0, \Sigma (\theta_0) \bigr), 
  \qquad \text{as $n \to \infty$},
\end{equation}
for some $q \times q$ covariance matrix $\Sigma (\theta_0)$, then, as $n \rightarrow \infty$,
\begin{equation}\label{eq:theta:an}
  \sqrt{k} \, (\htheta - \theta_0) 
  = ( \dot{L}^T \Omega \dot{L} \bigr)^{-1} 
  \dot{L}^T \Omega \,
  \sqrt{k} \, D_{n,k}(\theta_0) + o_p(1) 
  \dto 
  \Normal_p \bigl( 0, M(\theta_0) \bigr),
\end{equation}
where the $p \times p$ covariance matrix $M(\theta_0)$ is defined by
\begin{equation*}
  M(\theta_0)
  :=
  (\dot{L}^T \Omega \dot{L})^{-1} \, 
  \dot{L}^T \Omega \,
  \Sigma(\theta_0) \,
  \Omega \dot{L} \,
  (\dot{L}^T \Omega \dot{L})^{-1},
\end{equation*}
and the matrices $\dot{L}$ and $\Omega$ are evaluated at $\theta_0$. 
\end{thm}

Provided $\Sigma (\theta_0)$ is invertible, we can choose $\Omega$ in such a way that the asymptotic covariance matrix $M(\theta_0)$ is minimal, say $M_{\textnormal{opt}} (\theta_0)$, i.e., the difference $M(\theta_0) - M_{\textnormal{opt}}(\theta_0)$ is positive semi-definite. The minimum is attained at $\Omega(\theta_0) = \Sigma(\theta_0)^{-1}$ and the matrix $M(\theta_0)$ becomes simply
\begin{equation}
\label{eq:Mopt}
  \Mopt(\theta_0) =
  \bigl( \dot L(\theta_0)^T \, \Sigma(\theta_0)^{-1} \, \dot L(\theta_0) \bigr)^{-1},
\end{equation}
see for instance \citet[page 339]{abadir2005}. Now extend the covariance matrix $\Sigma (\theta_0)$ to the whole parameter space $\Theta$ by letting the map $\theta\mapsto \Sigma(\theta)$ be such that $\Sigma(\theta)$ is an invertible covariance matrix and $\Sigma^{-1} : \Theta \to \mathbb{R}^{q \times q}$ satisfies the assumptions on $\Omega$. 

\begin{cor}[Optimal weight matrix]
\label{cor:an:Mopt}
If the assumptions of Theorem~\ref{thm:an} are satisfied and $\htheta$ is the estimator based on the  weight matrix $\Omega (\theta) = \Sigma (\theta)^{-1}$, then, with $\Mopt$ as in \eqref{eq:Mopt}, we have
\begin{equation}
\label{eq:theta:an:opt} 
  \sqrt{k} ( \htheta - \theta_0 )
  \dto
  \Normal_p \bigl( 0, \Mopt(\theta_0) \bigr),
  \qquad \text{as $n \to \infty$.}
\end{equation}
\end{cor}

The asymptotic covariance matrices $M$  and $\Mopt$ in \eqref{eq:theta:an} and \eqref{eq:theta:an:opt}, respectively, depend on the unknown parameter vector $\theta_0$ through the matrices $\dot{L}(\theta)$, $\Omega(\theta)$ and $\Sigma(\theta)$ evaluated at $\theta = \theta_0$. If these matrices vary continuously with $\theta$, then it is a standard procedure to construct confidence regions and hypothesis tests, cf. \citet[Corollaries 4.3 and 4.4]{einmahl2012}.

\subsection{Goodness-of-fit testing}

It is of obvious importance to be able to test the goodness-of-fit of the parametric family of tail dependence functions that we intend to use. The basis for such a test is $D_{n,k} (\widehat{\theta}_{n,k})$, the difference vector between the initial and parametric estimators of $\ell( c_m)$ at the estimated value of the parameter. 

\begin{cor}
\label{cor:GoF}
Under the assumptions of Theorem~\ref{thm:an}, we have
\begin{align}
\nonumber
  \sqrt{k} \, D_{n,k}( \htheta )
  &=
  (I_q - P(\theta_0)) \, \sqrt{k} \, D_{n,k}( \theta_0 ) + o_p(1) \\
  &\dto
  \Normal_q 
  \bigl( 
    0, 
    (I_q - P(\theta_0)) \, \Sigma(\theta_0) \, (I_q - P(\theta_0))^T 
  \bigr),
  \qquad \text{as $n \to \infty$},
\label{eq:Dnk:GoF}
\end{align}
where $P := \dot{L} (\dot{L}^T \Omega \dot{L})^{-1} \, \dot{L}^T \Omega$ has rank $p$ and $I_q - P$ has rank $q-p$.
\end{cor}

The easiest case in which \eqref{eq:Dnk:GoF} can be exploited is when $\Sigma(\theta)$ is invertible and $\Omega (\theta) = \Sigma(\theta)^{-1}$. Then it suffices to consider the minimum attained by the criterion function $f_{n,k}$ in \eqref{eq:fnk}, i.e., the test statistic is just $f_{n,k} (\htheta) = \min_{\theta \in \Theta} f_{n,k} (\theta)$. Observe that it is important here that we allow $\Omega$ to depend on $\theta$.

\begin{cor}
\label{cor:GoF:Mopt}
Let $q > p$. If the assumptions of Corollary~\ref{cor:an:Mopt} are satisfied, in particular if $\Omega(\theta) = \Sigma(\theta)^{-1}$,
then
\begin{equation*}
  k \, f_{n,k} (\htheta) \dto \chi_{q-p}^2, 
  \qquad \text{as $n \to \infty$.}
\end{equation*}
\end{cor}

If $\Omega(\theta)$ is different from $\Sigma(\theta)^{-1}$, for instance when $\Sigma(\theta)$ is not invertible, a goodness-of-fit test can still be based upon \eqref{eq:Dnk:GoF} by considering the spectral decomposition of the limiting covariance matrix. For convenience, we suppress the dependence on $\theta$. Let
\begin{equation*}
  (I_q - P) \, \Sigma \, (I_q - P)^T = V D V^T
\end{equation*}
where $V = (v_1, \ldots, v_q)$ is an orthogonal $q \times q$ matrix, $V^T V = I_q$, the columns of which are orthonormal eigenvectors, and $D$ is diagonal, $D = \diag(\nu_1, \ldots, \nu_q)$, with $\nu_1 \ge \ldots \ge \nu_q = 0$ the corresponding eigenvalues, at least $p$ of which are zero, the rank of $I_q - P$ being $q-p$. Let $s \in \{1, \ldots, q-p\}$ be such that $\nu_s > 0$ and consider the $q \times q$ matrix
\[
  A := V_s D_s^{-1} V_s^T
\]
where $D_s = \diag( \nu_1, \ldots, \nu_s )$ is an $s \times s$ diagonal matrix and where $V_s = (v_1, \ldots, v_s)$ is a $q \times s$ matrix having the first $s$ eigenvectors as its columns. 

\begin{cor}
\label{cor:GoF:eigen}
If the assumptions of Theorem~\ref{thm:an} hold and if $s \in \{1, \ldots, q-p\}$ is such that, in a neighbourhood of $\theta_0$, $\nu_s(\theta) > 0$ and the matrix $A(\theta)$ depends continuously on $\theta$, then
\[
  k \, D_{n,k}( \htheta )^T \, A( \htheta ) \, D_{n,k}( \htheta )
  \dto \chi_s^2,
  \qquad \text{as $n \to \infty$}.
\]
\end{cor}

\begin{rem}\label{rem:spq}
If $\Sigma(\theta)$ is invertible for all $\theta$, then we can set $s = q-p$ and $\Omega(\theta) = \Sigma(\theta)^{-1}$. The difference between the two test statistics in Corollaries~\ref{cor:GoF:Mopt} and~\ref{cor:GoF:eigen} then converges to zero in probability, i.e., the two tests are asymptotically equivalent under the null hypothesis.
\end{rem}

\subsection{Choice of the initial estimator}
\label{sec:nonparametric}

Our estimator in \eqref{eq:estimator} is flexible enough to allow for various initial estimators, perhaps based on exceedances over high thresholds or rather on vectors of componentwise block maxima extracted from a multivariate time series \citep{bucher2014}. Here we will focus on the former case, and more specifically on the empirical tail dependence function and a variant thereof.

For simplicity, we assume that the random vectors $X_i$, $i \in \{1,\ldots,n\},$ are not only identically distributed but also independent, so that they are a random sample from $F$. Let $R_{ij}^n$ denote the rank of $X_{ij}$ among $X_{1j}, \ldots , X_{nj}$ for $j=1,\ldots,d$. For convenience, assume that $F$ is continuous.

\subsubsection*{Empirical stable tail dependence function}

A natural estimator of $\ell(x)$ is obtained by replacing $F$ and $F_1, \ldots, F_d$ in~\eqref{eq:ell2} by their empirical counterparts and replacing $t$ by $k/n$, yielding
\begin{equation}
\label{eq:ellclassic:raw}
  \widetilde{\ell}'_{n,k} (x) 
  := 
  \frac{1}{k} \sum_{i=1}^n 
  \mathbbm{1} 
  \left\{
    R_{i1}^n > n + 1 - kx_1  \text{ or } \ldots \text{ or }  R_{id}^n > n + 1 - kx_d
  \right\}.
\end{equation}
This estimator, the empirical stable tail dependence function,  was introduced for $d=2$  in \citet{huang1992} and studied further in \citet{dreeshuang1998}. A slight modification of it allows for better finite-sample properties,
\begin{equation}
\label{eq:ellclassic}
  \widetilde{\ell}_{n,k} (x) 
  := 
  \frac{1}{k} \sum_{i=1}^n 
  \mathbbm{1} 
  \left\{
    R_{i1}^n > n + 1/2 - kx_1  \text{ or } \ldots \text{ or }  R_{id}^n > n + 1/2 - kx_d
  \right\}.
\end{equation}

By \citet[Theorem~4.6]{einmahl2012}, this estimator satisfies \eqref{eq:nonparcond} under conditions controlling the rate of convergence in \eqref{eq:ell2} and the growth rate of the intermediate sequence $k = k_n$. The first-order partial derivatives $\dot{\ell}_j(x; \theta_0)$ of $x \mapsto \ell(x; \theta_0)$ are assumed to exist and to be continuous in neighbourhoods of the points $c_m$ for which $c_{mj} > 0$.

In this case, the entries of the matrix $\Sigma(\theta)$ in \eqref{eq:nonparcond}, for $\theta$ in the interior of $\Theta$, are, for $i,j \in \{1,\ldots,q\}$, given by
\begin{equation}
\label{eq:Sigma:classic}
  \Sigma_{i,j} (\theta)  = \EE [ B (c_i) \, B(c_j) ],
\end{equation}
with
$
  B (c_i)
  := 
  W_\ell (c_i) 
  - 
  \sum_{j=1}^d \dot{\ell}_j (c_i) W_{\ell} (c_{ij} \, e_j)
$
and with $(W_\ell(x) : x \in [0, \infty)^d)$ a zero-mean Gaussian process with covariance function
$
  \EE [W_{\ell} (x) \, W_{\ell} (y) ] 
  =
  \ell(x) + \ell(y) - \ell(x \vee y)
$,
the maximum being taken componentwise. For points $c_i$ of the form $e_J$ in \eqref{ej}, the expectation in \eqref{eq:Sigma:classic} can be calculated as follows: for non-empty subsets $J$ and $K$ of $\{1, \ldots, d\}$,
\begin{align*}
  \EE [ B (e_J) \, B(e_K) ] 
  =  \ell_J + \ell_K - \ell_{J \cup K} & - \sum_{j \in J} \dot{\ell}_{j,J} \, (1 + \ell_K - \ell_{\{j\} \cup K}) \\
  & - \sum_{k \in K} \dot{\ell}_{k,K} \, (\ell_J + 1 - \ell_{J \cup \{k\}})
  + \sum_{j \in J} \sum_{k \in K} \dot{\ell}_{j,J} \dot{\ell}_{k,K} \, (2 - \ell_{\{j,k\}}),
\end{align*}
where $\ell_J := \ell( e_J; \theta_0 )$ and $\dot{\ell}_{j,J} := \dot{\ell}_j( e_J;\theta_0)$.

\subsubsection*{Bias-corrected  estimator}

A drawback of $\widetilde{\ell}_{n,k}$ in \eqref{eq:ellclassic} is its possibly quickly growing bias as $k$ increases.
Recently, two bias-corrected estimators  have been proposed. We consider here the  kernel-type estimator of  \citet{beirlant2015}, which is  partly based on (the one in)  \citet{fougeres2015nr2}.

Consider first a rescaled version of $\widetilde{\ell}'_{n,k}$ in \eqref{eq:ellclassic:raw}, defined as $\widetilde{\ell}_{n,k,a} (x) := a^{-1} \widetilde{\ell}'_{n,k} (a x)$ for $a > 0$. Then define the weighted average
\begin{equation}
\label{eq:ell:beirlant}
  \breve{\ell}_{n,k} (x) 
  := \frac{1}{k} \sum_{j=1}^k K(a_j)\, \widetilde{\ell}_{n,k,a_j} (x), 
  \qquad a_j := \frac{j}{k+1}, \,\, j \in \{1,\ldots,k\},
\end{equation}
where $K$ is a kernel function, i.e., a positive function on $(0,1)$ such that $\int_0^1 K(u) \diff u = 1$.

In addition to \eqref{eq:ell2}, we assume there exist a positive function $\alpha$ on $(0, \infty)$ tending to $0$ as $t \downarrow 0$ and a non-zero function $M$ on $[0, \infty)^d$ such that for all $x \in [0,\infty)^d$,
\begin{equation}
\label{eq:alpha}
\lim_{t \downarrow 0} \frac{1}{\alpha(t)} [t^{-1}\PP \left\{ 1 - F_1 (X_{11}) \leq tx_1 \text{ or } \ldots \text{ or } 1 - F_d (X_{1d})  \leq t x_d \right\}-\ell(x)] = M(x).
\end{equation}
Moreover, we assume a third-order condition on $\ell$ \citep[equation (3)]{beirlant2015}. In \citet[Theorem 1]{beirlant2015} the asymptotic distribution of $\breve{\ell}_{n,k}$ in \eqref{eq:ell:beirlant} is derived under these three assumptions and for intermediate sequences $k = k_n$ growing faster than the ones considered above. A non-zero asymptotic bias term arises and the idea is to estimate and remove it, thereby obtaining a possibly more accurate estimator.

In order to achieve this bias reduction, the rate function, $\alpha$, and its index of regular variation, $\beta$, need to be estimated. Consider another intermediate sequence $k_1 = k_{1,n}$ such that $k /k_1\to 0$. The bias-corrected estimator  is then defined as
\begin{equation*}
\overline{\ell}_{n,k,k_1} (x) := \frac{\breve{\ell}_{n,k} (x) - (k_1/k)^{\widehat{\beta}_{k_1} (x)} \widehat{\alpha}_{k_1} (x) \frac{1}{k} \sum_{j=1}^k K(a_j) a_j^{-\widehat{\beta}_{k_1} (X)}}{\frac{1}{k} \sum_{j=1}^k K(a_j)},
\end{equation*}
where $\widehat{\alpha}_{k_1}$ and $\widehat{\beta}_{k_1}$ are the estimators of $\alpha$ and $\beta$ defined in \citet{beirlant2015}.
Under the mentioned conditions, asymptotic normality as in \eqref{eq:nonparcond} holds, where the limiting random vector is equal in distribution to $\int_0^1 K(u) u^{-1/2}\diff u$ times the one corresponding to  $\widetilde{\ell}_{n,k}$.
Here, the growth rate of $k$ here can be taken faster than  when using  $\widetilde{\ell}_{n,k}$.

A simple choice for $K$ is a power kernel, i.e, $K(t) = (\tau + 1) t^\tau$ for $t \in (0,1)$ and $\tau > -1/2$. Then $\int_0^1 K(u) \, u^{-1/2} \, \diff u = (2+\tau)/(1 + 2 \tau)$. Note that this factor tends to 1 if $\tau \to \infty$. In practice, we take $\tau = 5$ as recommended in \citet{beirlant2015}.

\section{Simulation studies}
\label{sec:simulation}
We conduct simulation studies for data in the max-domain of attraction of the logistic model, the Brown--Resnick process and the max-linear model. For each model, we report the empirical bias, standard deviation, and root mean squared error (RMSE) of our estimators. We also study the finite-sample performance of the goodness-of-fit statistic of Corollary~\ref{cor:GoF:Mopt}. All simulations were done in the \textsf{R} statistical software environment \citep{R}.

\subsection{Logistic model: comparison with likelihood methods}
The $d$-dimensional logistic model has stable tail dependence function
\begin{equation*}
  \ell(x_1,\ldots,x_d ; \theta) 
  = \bigl( x_1^{1/\theta} + \cdots + x_d^{1/\theta} \bigr)^\theta, 
  \qquad \theta \in [0,1].
\end{equation*}
The domain-of-attraction condition \eqref{eq:ell2} holds for instance if $F$ has continuous margins and its copula is Archimedean with generator $\phi(t) = 1/ (t^\theta +  1)$, also known as the outer power Clayton copula \citep{hofert2015}.

In \citet{huser2015}, a comprehensive comparison of likelihood estimators for $\theta$ has been performed based on random samples from this copula. We compare those results to our extremal coefficients estimator, i.e., the weighted least squares estimator based on points $c_m$ of the form $e_J$, with $J$ ranging in the collection
\begin{equation}
\label{eq:QQ:a}
  \QQ_{a} := \bigl\{J \subset \{1,\ldots,d\} : \abs{J} = a \bigr\}
\end{equation}
for $a \in \{2, 3\}$. Moreover, we let $\Omega (\theta)$ be the identity matrix, since by exchangeability of the model, a weighting procedure can bring no improvements. 

Following \citet[Section 4.2]{huser2015}, we simulated $10 \, 000$ random samples of size $n = 10 \, 000$ from the outer power Clayton copula. For the likelihood-based estimators, the margins are standardized to the unit Pareto scale via the rank transformation
\begin{equation*}
  X_{ij}^* := \frac{n}{n + 1/2 - R_{ij}^n}, 
  \qquad i \in \{ 1,\ldots,n \},\ j \in \{ 1,\ldots,d \}.
\end{equation*}
Again as in \citet[Section 4.2]{huser2015}, we take dimension $d \in \{2,5,10,15,20,25,30\}$ and parameter $\theta \in \{0.3,0.6,0.9,0.95\}$. Note that in the likelihood setting, this is a very demanding experiment, and three of the ten likelihood-based estimators considered in \citet{huser2015} are only computed for $d \in \{2,5,10\}$. In \citet{huser2015}, threshold probabilities are set to $0.98$, corresponding to $k = 200$ in our setup. 

\begin{figure}[p]
\centering
\subfloat{\includegraphics[width=0.3\textwidth]{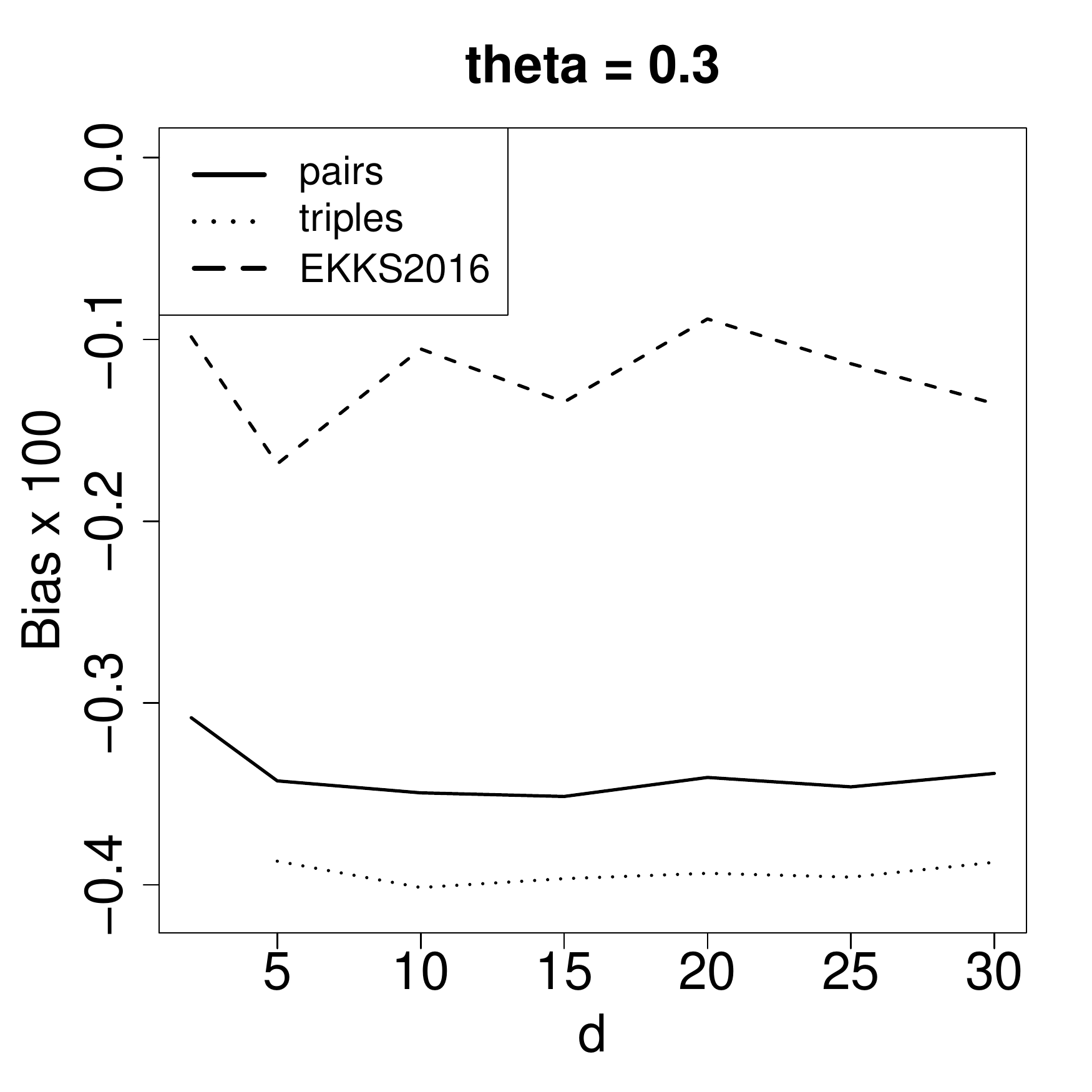}}
\subfloat{\includegraphics[width=0.3\textwidth]{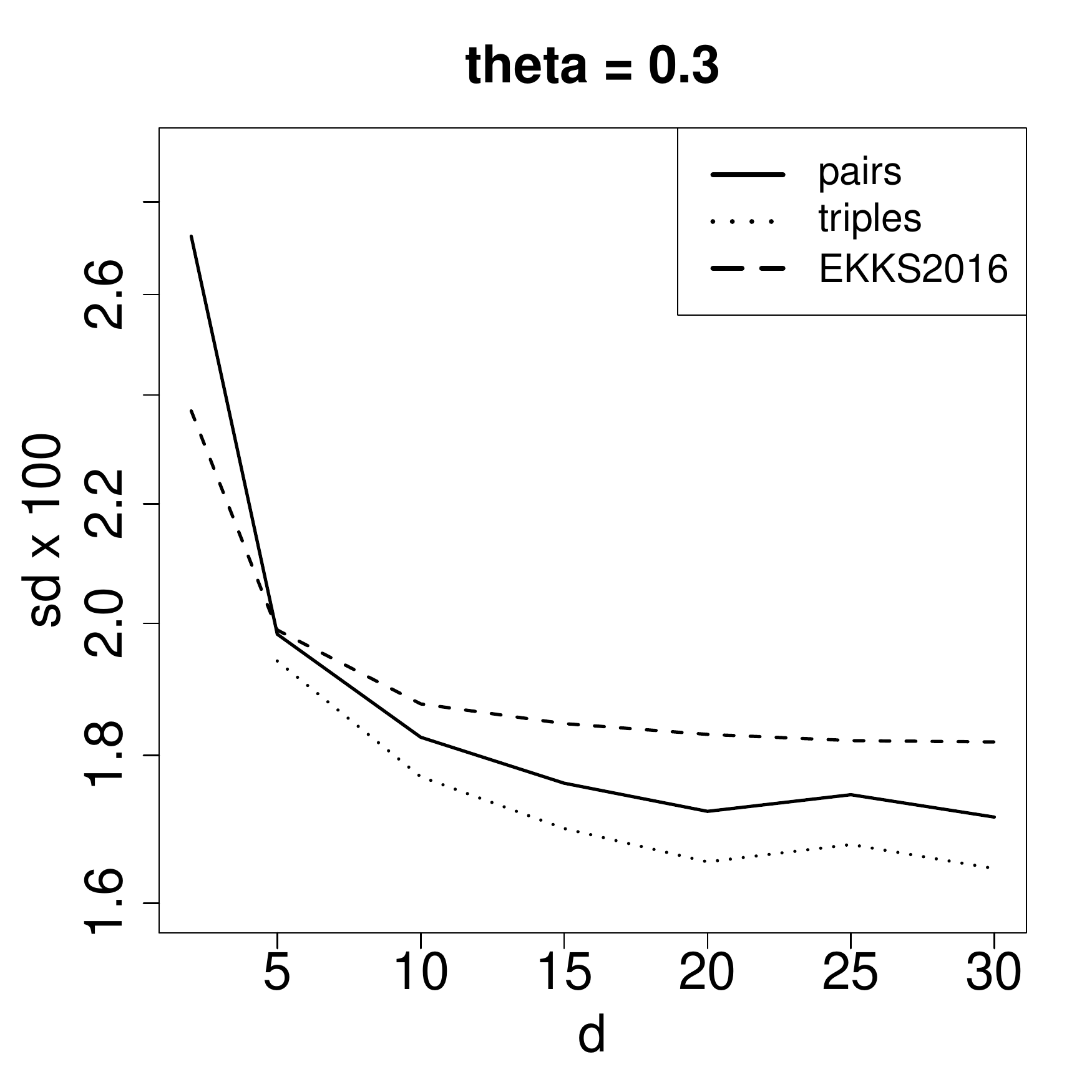}}
\subfloat{\includegraphics[width=0.3\textwidth]{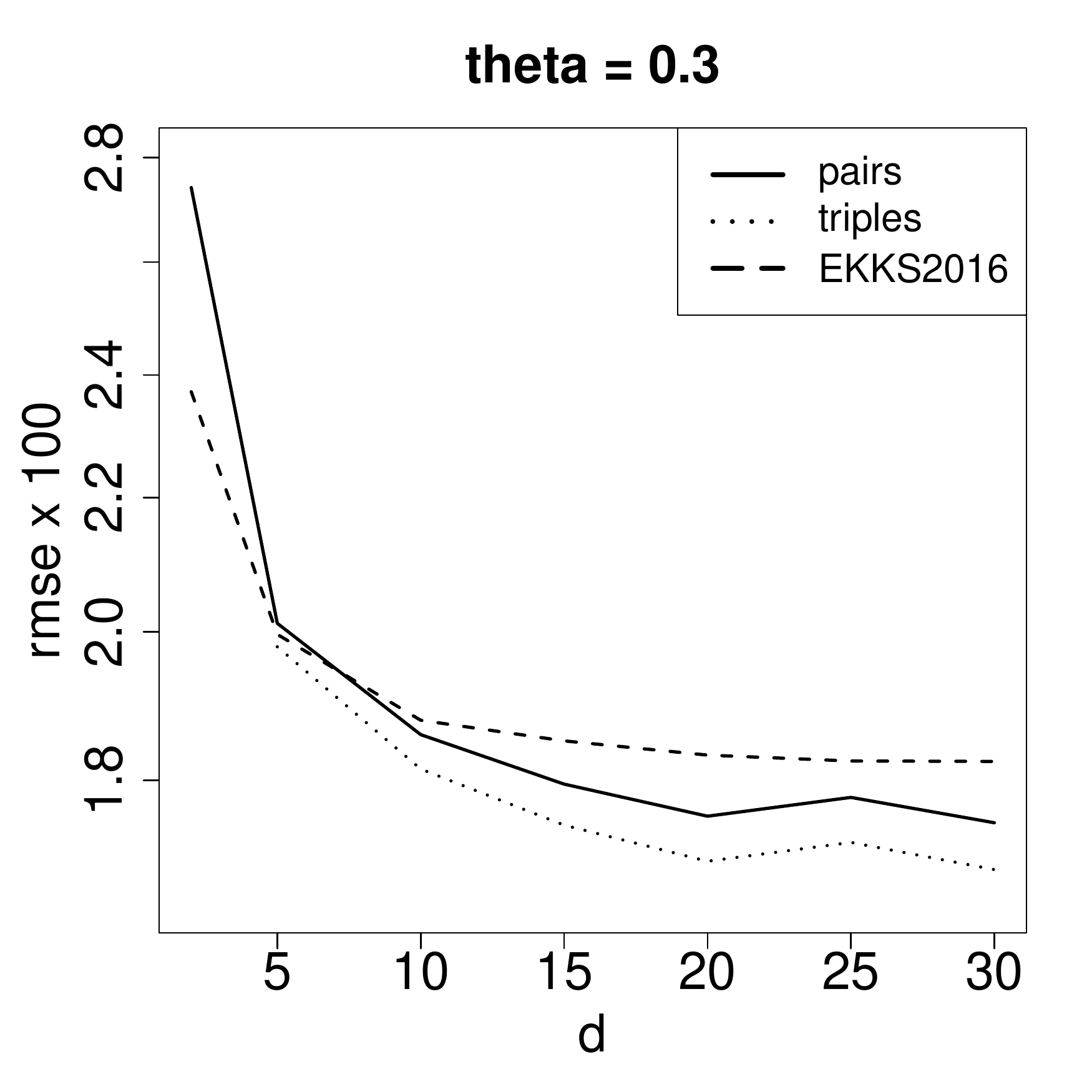}} \\
\subfloat{\includegraphics[width=0.3\textwidth]{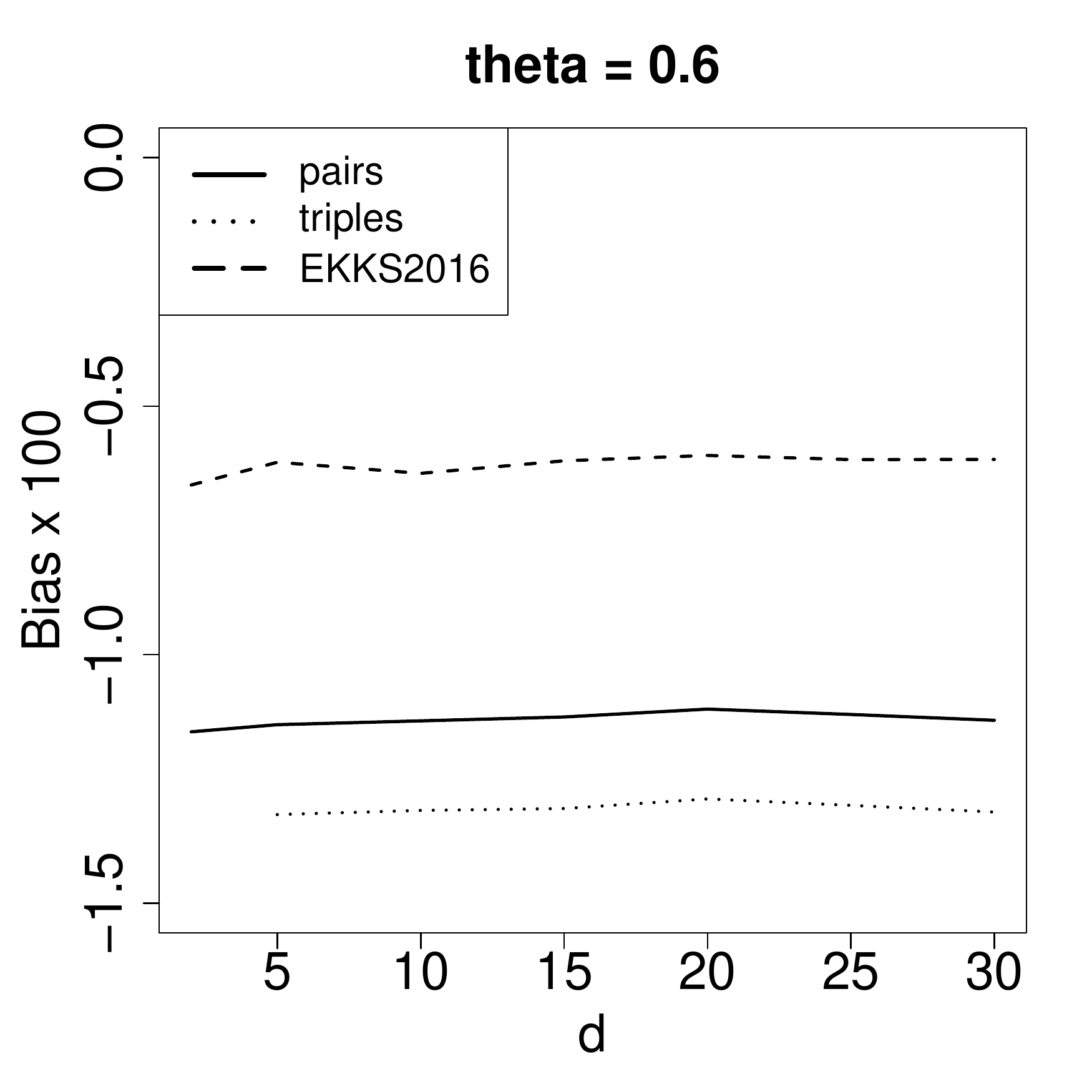}}
\subfloat{\includegraphics[width=0.3\textwidth]{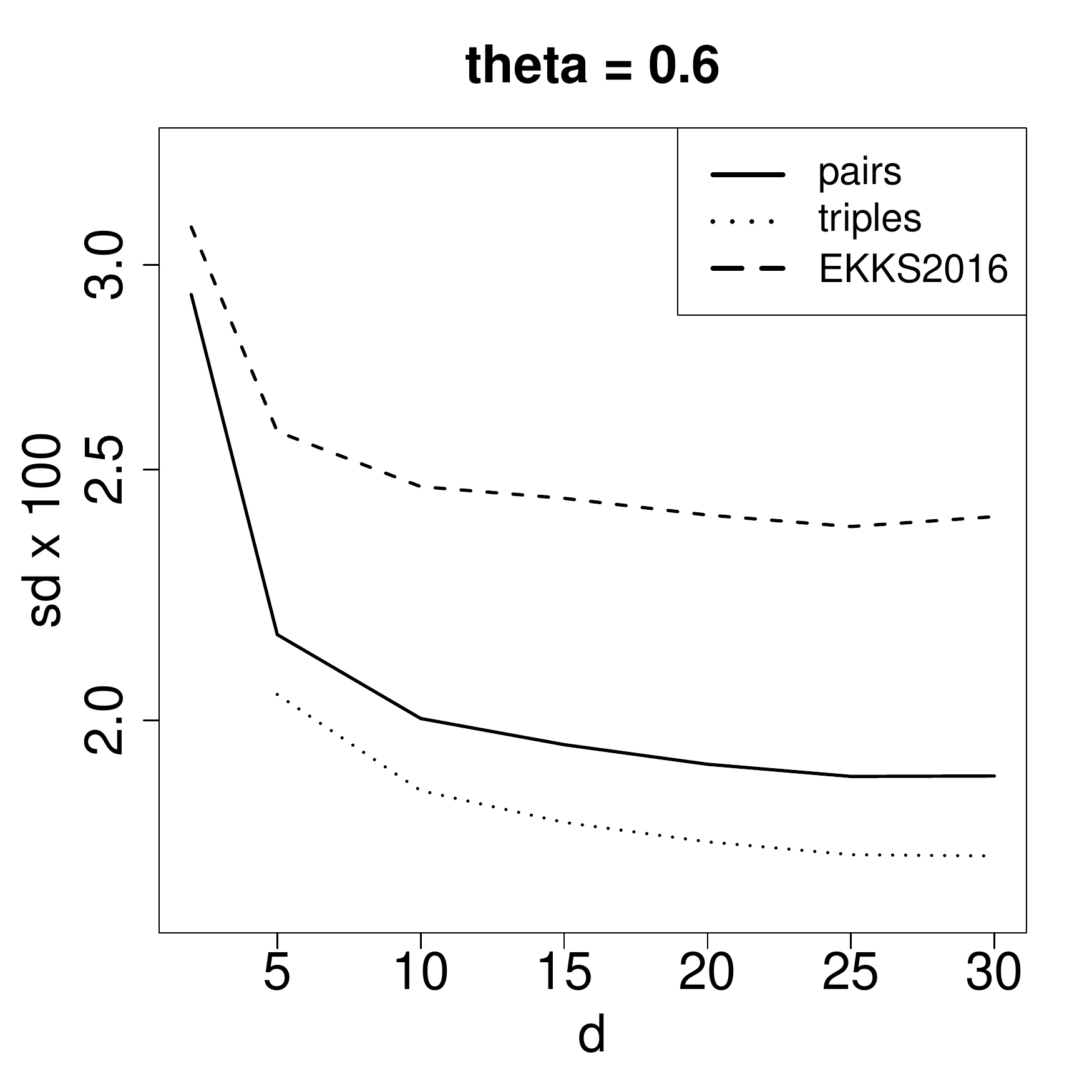}}
\subfloat{\includegraphics[width=0.3\textwidth]{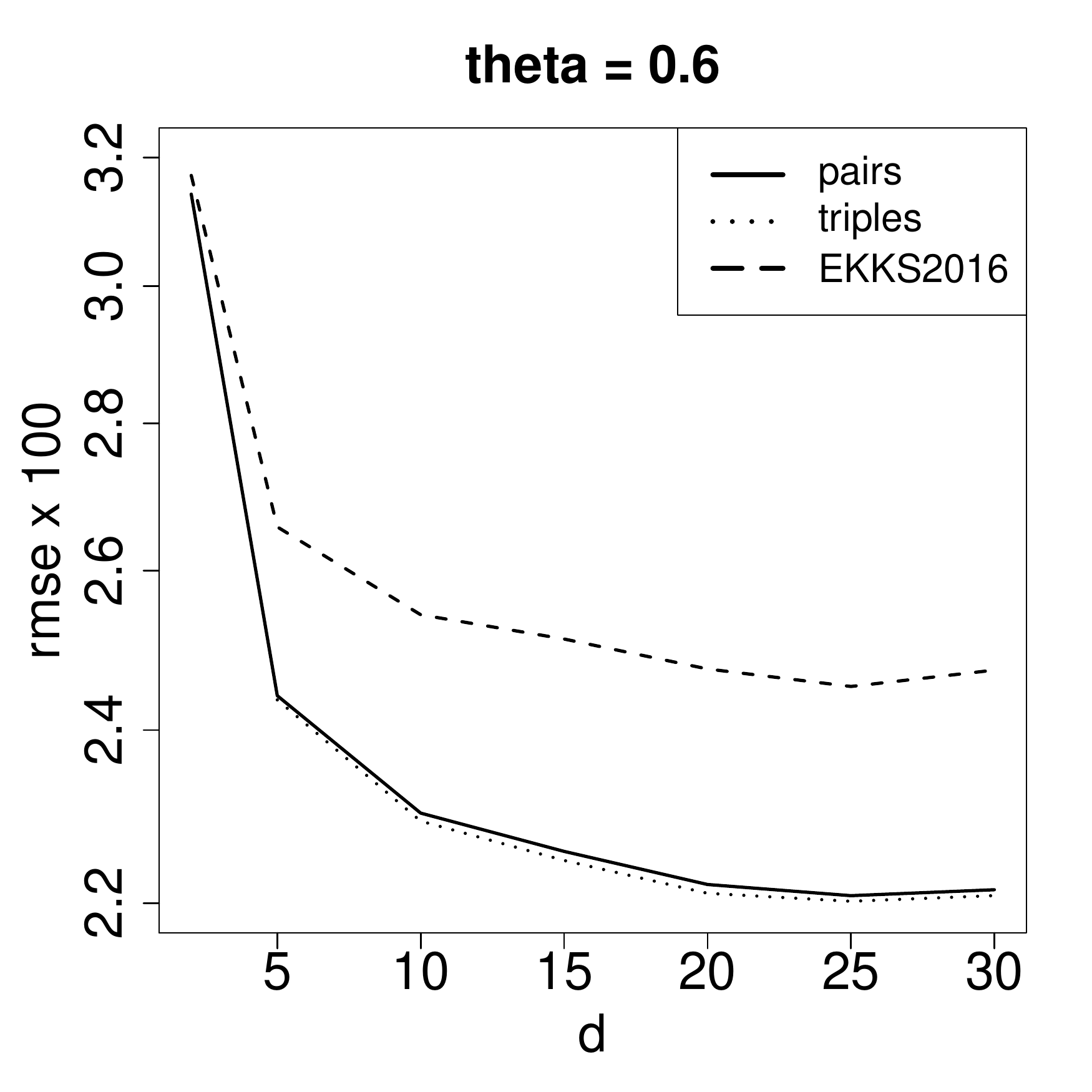}} \\
\subfloat{\includegraphics[width=0.3\textwidth]{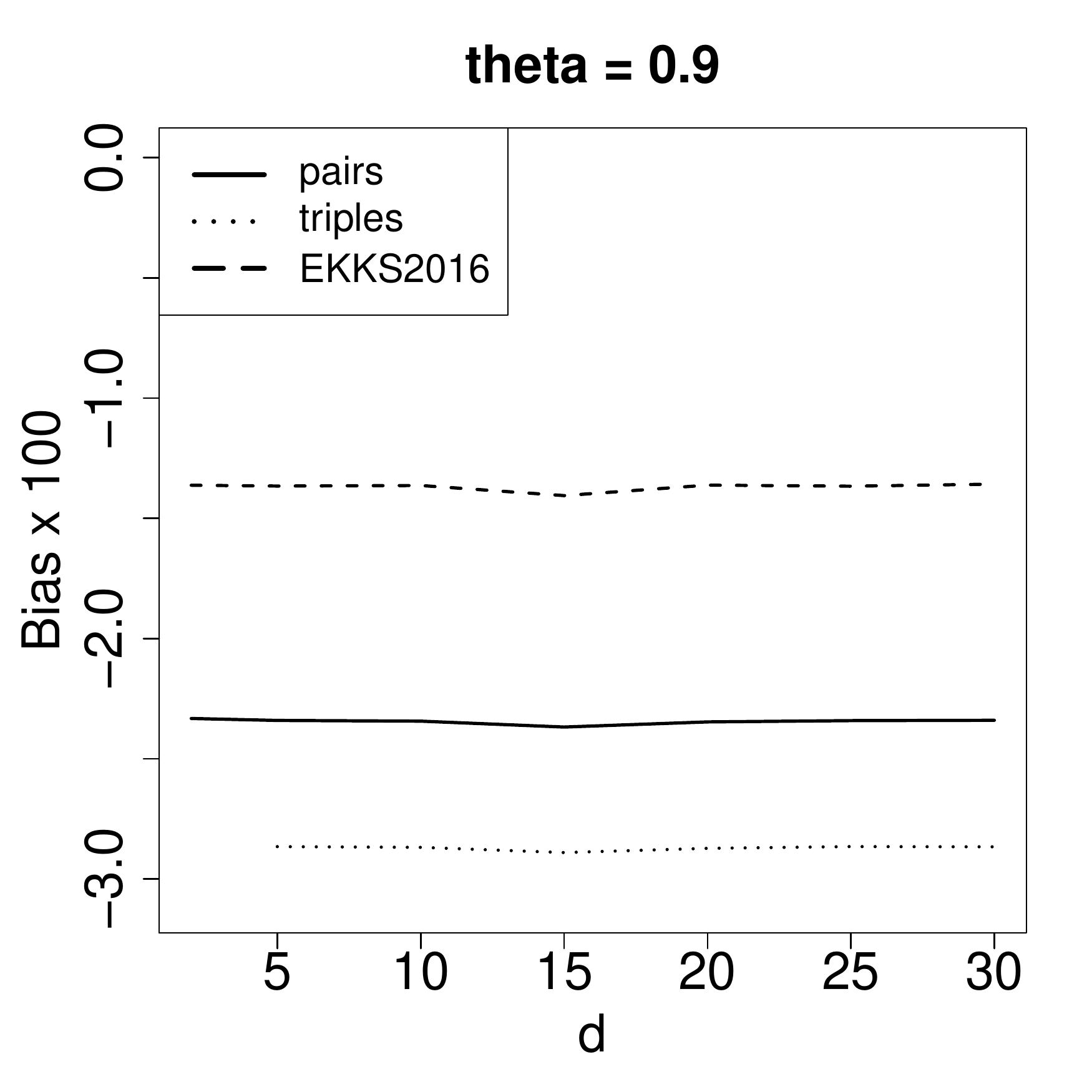}}
\subfloat{\includegraphics[width=0.3\textwidth]{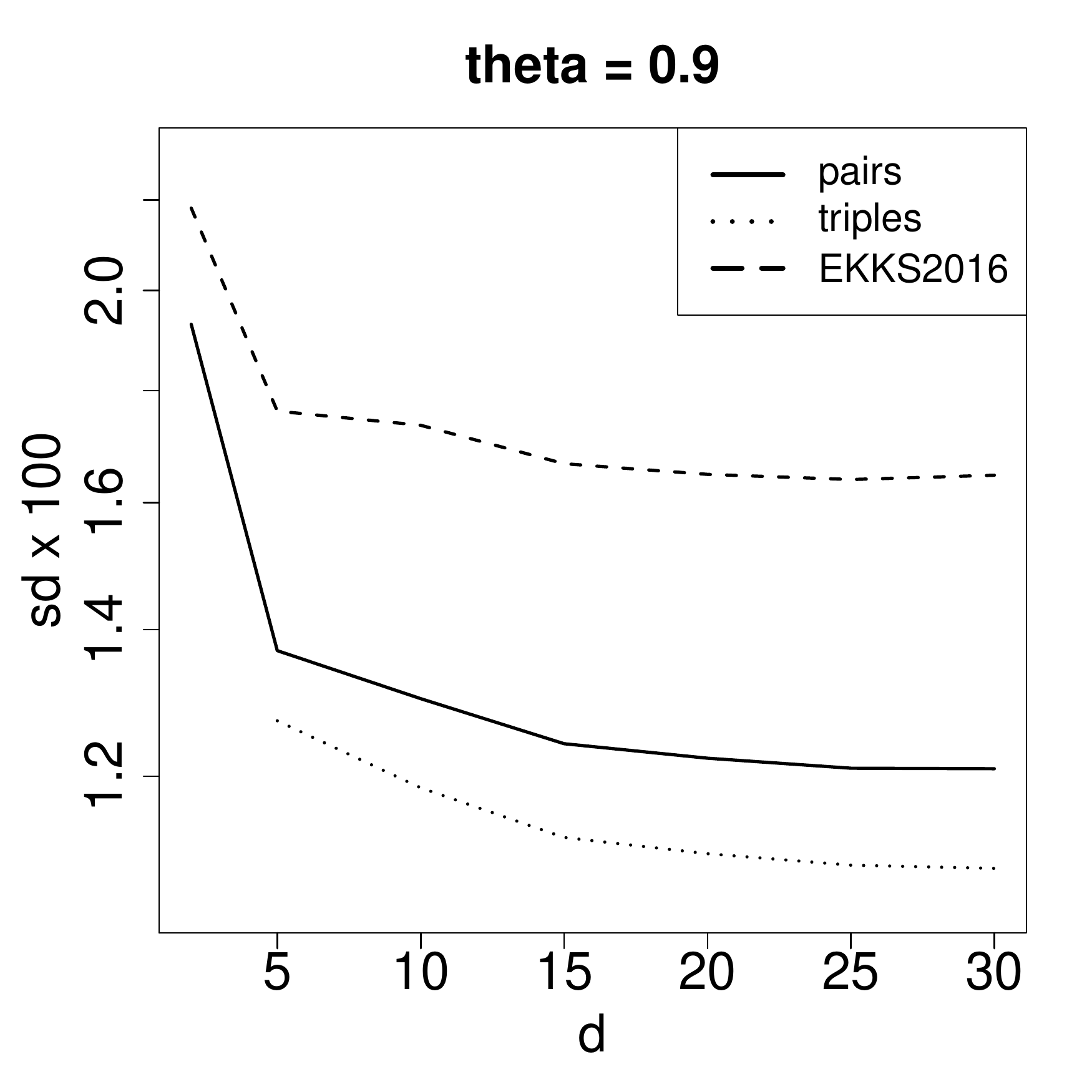}}
\subfloat{\includegraphics[width=0.3\textwidth]{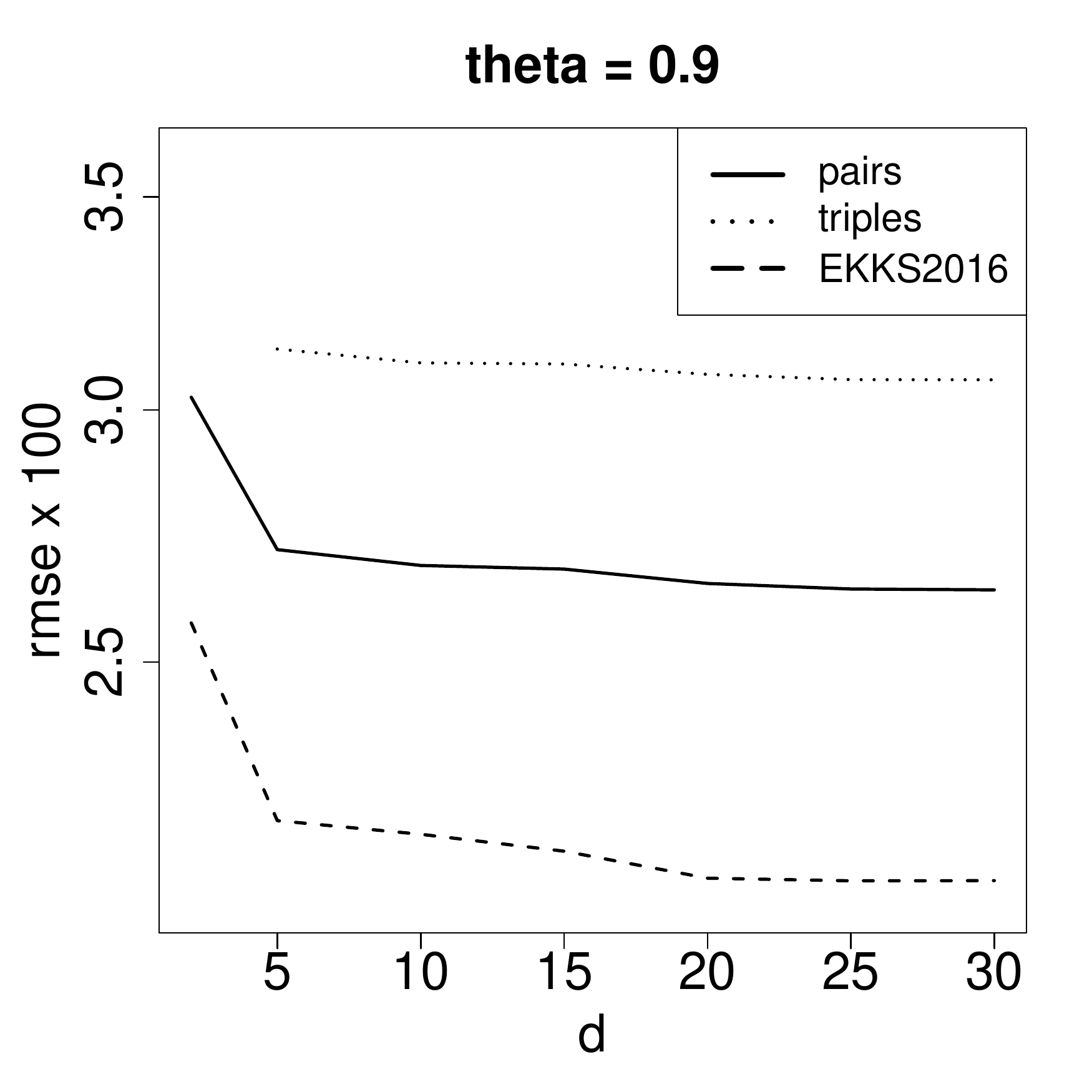}} \\
\subfloat{\includegraphics[width=0.3\textwidth]{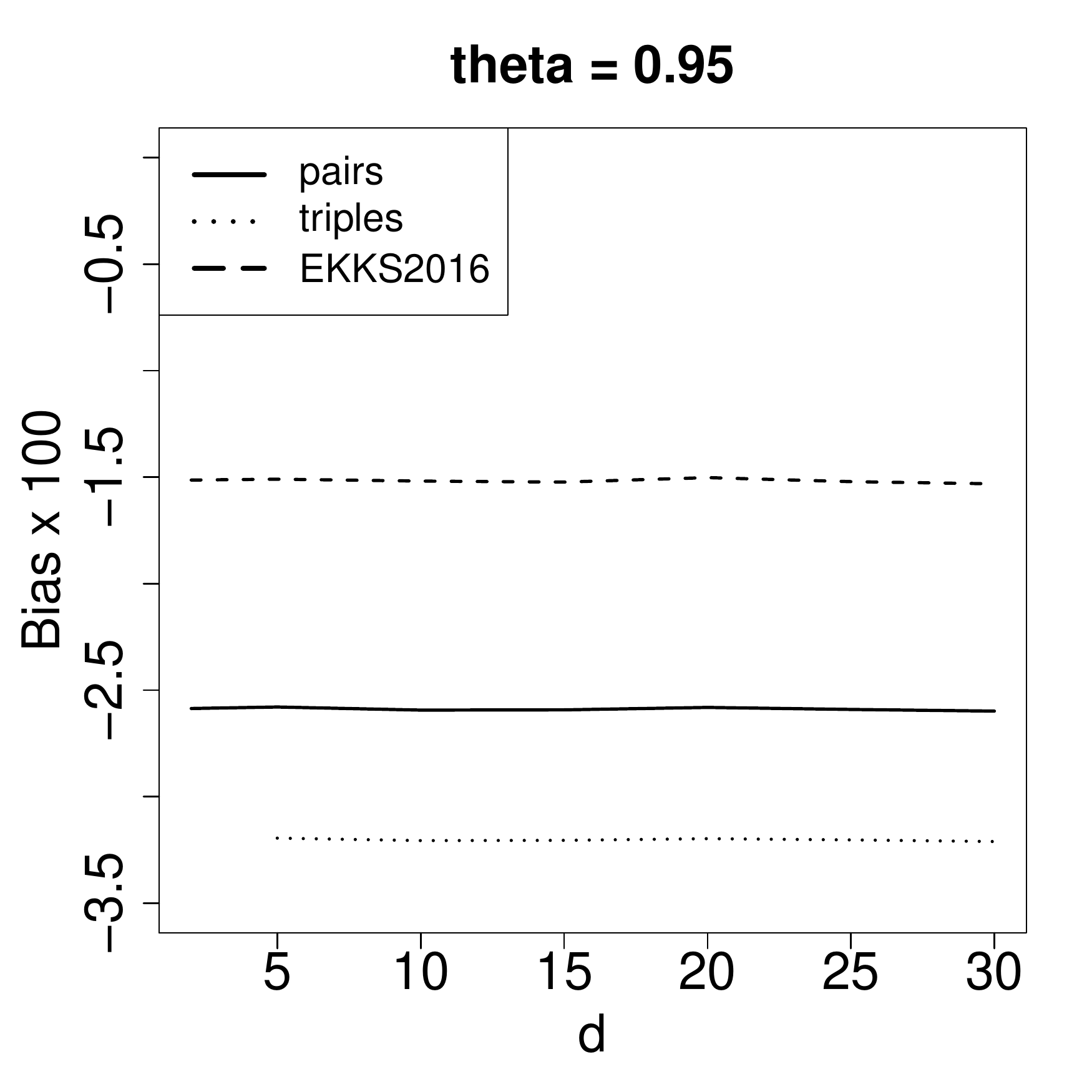}}
\subfloat{\includegraphics[width=0.3\textwidth]{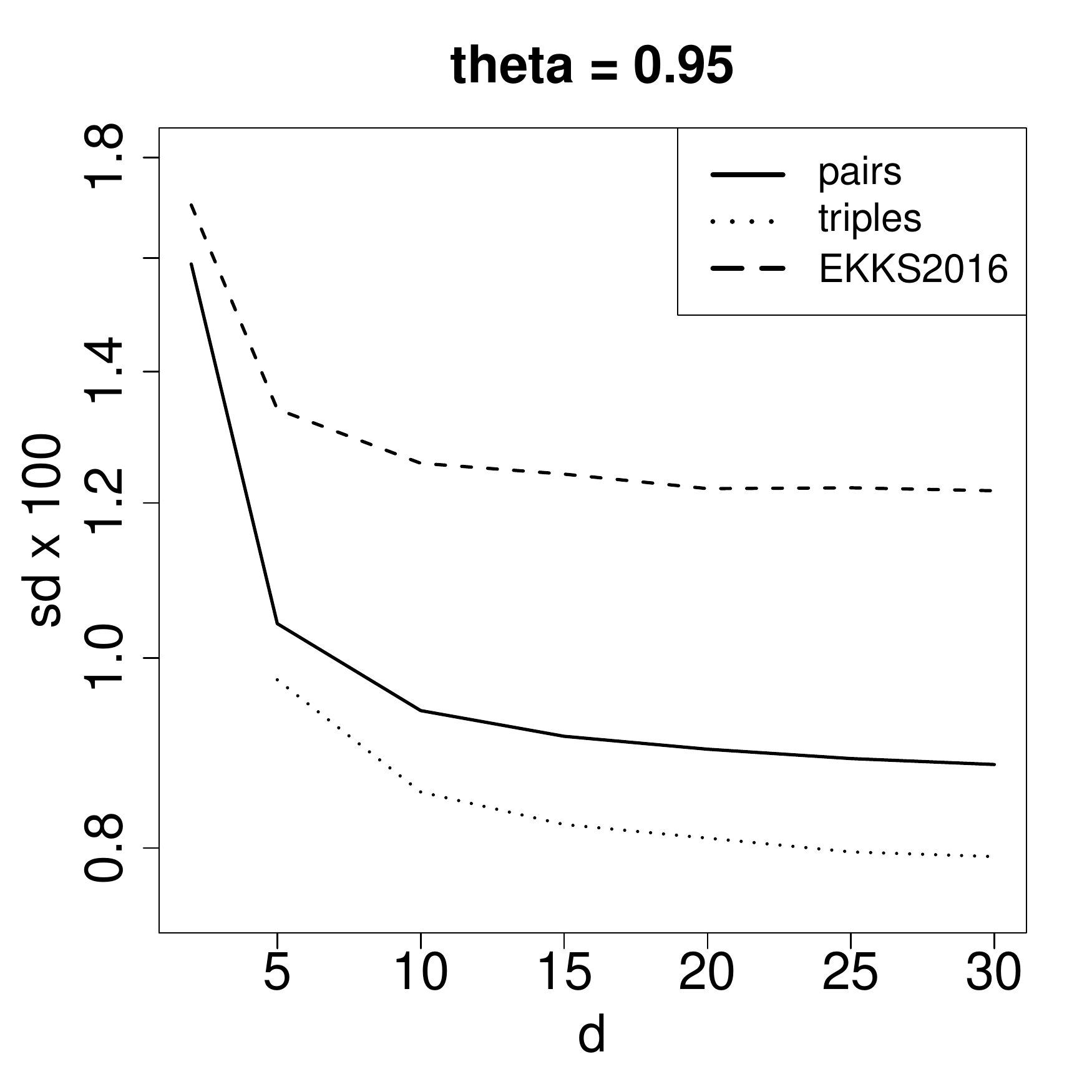}}
\subfloat{\includegraphics[width=0.3\textwidth]{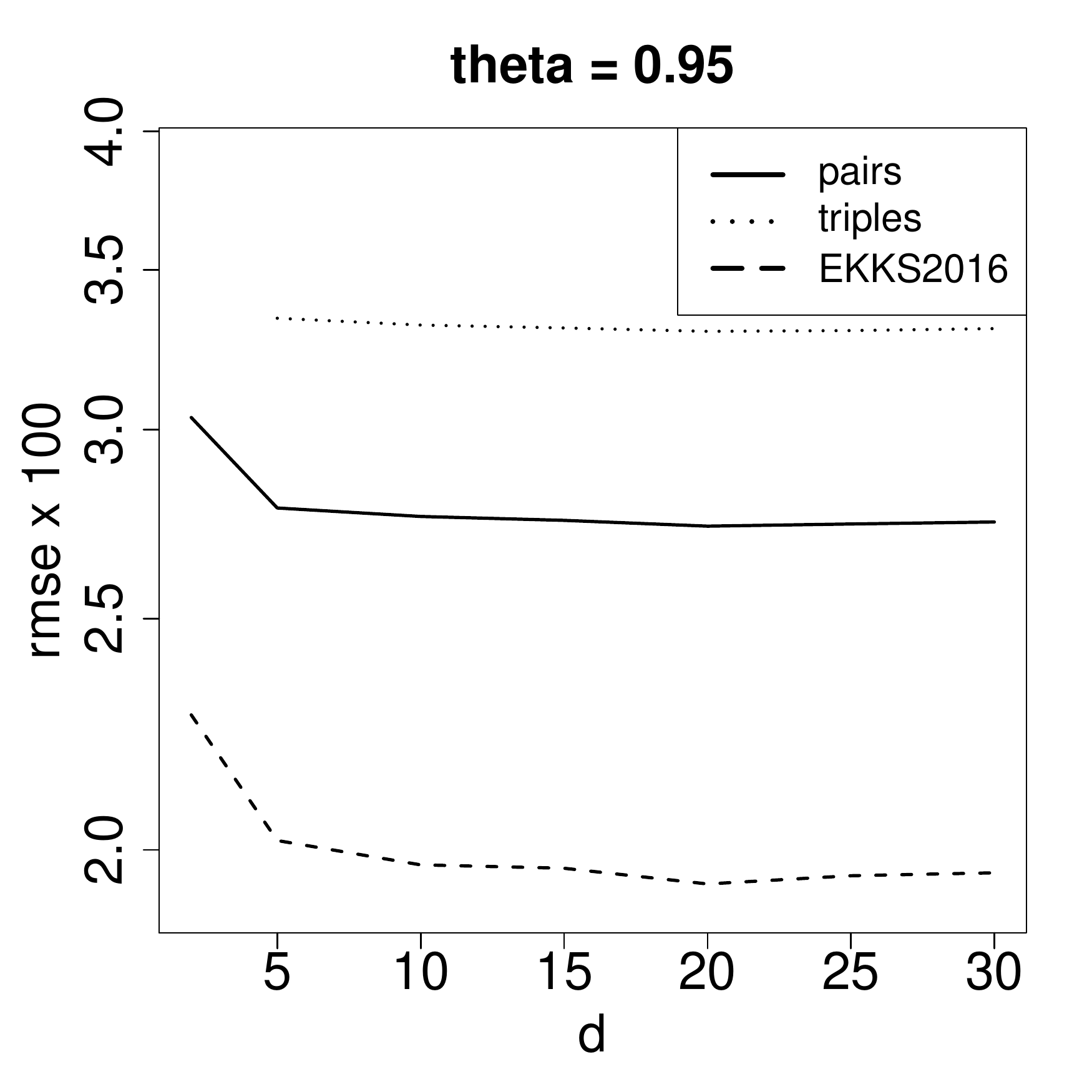}} \\
\caption{Logistic model: bias, standard deviation and RMSE for the estimators; $10 \, 000$ samples of size $n = 10 \, 000$. Standard errors and RMSEs are displayed on a logarithmic scale.}
\label{fig:simresults1}
\end{figure}

Figure~\ref{fig:simresults1} shows the bias, standard deviation and RMSE of three estimators based on the empirical tail dependence function: the two extremal coefficient estimators mentioned above and the pairwise M-estimator of \citet{einmahl2016} as implemented in the \textsf{R} package \textsf{spatialTailDep} \citep{kiriliouk+s:2014}. 
As the tuple size changes from pairs to triples, the absolute bias increases but the standard deviation decreases. When dependence is strong, $\theta = 0.3$, the gains in variance offset the losses in bias and the estimator based on $\QQ_3$ performs best. Note also that when the dependence is not too weak, the estimators based on extremal coefficients perform better than the pairwise M-estimator of \citet{einmahl2016}. Finally, our estimation procedures have almost constant RMSE as the dimension increases, in line with the pairwise composite likelihood methods studied in \citet{huser2015}.

Comparing these results to the ten likelihood-based estimators in \citet[Figure 4]{huser2015}, we see that our estimators are strong competitors in the sense that they rank highly when comparing RMSEs, and are not dominated by one of the likelihood-based estimators. More precisely, for $\theta = 0.3$, only the likelihood estimators based on the Poisson process representation \citep{coles1991} and the multivariate Generalized Pareto distribution outperform our estimators; for $\theta = 0.6$, the same two likelihood estimators outperform ours, but only for $d \geq 15$; finally, for $\theta = 0.9$ and $\theta = 0.95$ only the pairwise censored likelihood estimator \citep{huser2014} has a smaller RMSE than our estimators.

\subsection{Brown--Resnick process}
The Brown--Resnick process on a planar set $\mathcal{S} \subset \RR^2$ is given by
\begin{equation}\label{eq:BR}
Y(s) =  \max_{i \in \mathbb{N}} \xi_i \exp{ \left\{ \epsilon_i (s) - \gamma(s) \right\}}, \qquad s \in \mathcal{S},
\end{equation}
where $\{\xi_i\}_{i \geq 1}$ is a Poisson process on $(0,\infty)$ with intensity measure $\xi^{-2} \,\diff \xi$ and $\{\epsilon_i (\point) \}_{i \geq 1}$ are independent copies of a Gaussian process $\epsilon$ with stationary increments such that $\epsilon (0) = 0$ and with variance $2 \gamma (\point)$ and semi-variogram $\gamma (\point)$. In \citet{kabluchko2009} it is shown that the Brown--Resnick process with $\gamma (s) = (\norm{s} / \rho)^{\alpha}$ is the only possible limit of (rescaled) maxima of stationary and isotropic Gaussian random fields; here $\rho > 0$ and $0 < \alpha  \leq 2$.  

For $d$ locations $s_1, \ldots, s_d \in \mathcal{S}$, the distribution of the random vector $(Y(s_i))_{i = 1}^d$ is max-stable with tail dependence function $\ell$ depending on $\gamma(\point)$.  From \cite{huser2013}, we obtain the following representation for the extremal coefficients $\ell_J$ in \eqref{eq:extrCoeff}. Let $\Phi_a( \point; R)$ denote the cumulative distribution function of the $\Normal_a( 0, R )$ distribution. Then we have
\begin{equation*}
\ell_J  =  \sum_{j \in J}  \Phi_{d-1} ( \eta^{(j)} ; R^{(j)} ) , \qquad J \subset \{1,\ldots,d\}, J \neq \emptyset,
\end{equation*}
where $\eta^{(j)}  = (\eta_1^{(j)}, \ldots, \eta_{j-1}^{(j)} , \eta_{j+1}^{(j)} , \ldots , \eta_d^{(j)} )$ with $\eta_i^{(j)} = \sqrt{\gamma (s_j - s_i)/2}$, and where $R^{(j)}$ is a $(d-1) \times (d-1)$ correlation matrix with entries given by
\begin{equation*}
  R^{(j)}_{ik} 
  = 
  \frac%
  {\gamma (s_j - s_i) + \gamma (s_j - s_k) - \gamma (s_i - s_k)}%
  {2 \sqrt{\gamma (s_j-s_i) \, \gamma (s_j - s_k)}}, 
  \qquad i,k \in \{1,\ldots,d\} \setminus \{j\}.
\end{equation*}

We simulate 300 random samples of size $n = 1000$ from the Brown--Resnick process on a $3 \times 4$ unit distance grid using the \textsf{R} package \textsf{SpatialExtremes} \citep{ribatet:2015}. To arrive at a more realistic estimation problem, we perturb the samples thus obtained with additive noise, i.e., if $Y_{i} = (Y_{i1},\ldots,Y_{id})$ is an observation from the Brown--Resnick process, then we set $X_{ij} = Y_{ij} + |\epsilon_{ij}|$ for $i=1,\ldots,n$ and $j=1,\ldots,d$, where $\epsilon_{ij}$ are independent $\mathcal{N}(0,1/4)$ random variables.

We estimate the parameters $(\alpha,\rho) = (1,1)$ using the extremal coefficient estimator based on the subset of $\QQ_2$ in \eqref{eq:QQ:a} consisting of pairs of neighbouring locations, i.e., locations that are at most a distance $\sqrt{2}$ apart. This leads to $q = 29$ pairs. Including pairs of locations that are further away tends to drastically increase the bias \citep{einmahl2016}.

The upper panels of Figure \ref{fig:brplot} show the bias, standard deviation and RMSE for three estimators: the estimator based on the empirical tail dependence function with $\Omega (\theta) = \Sigma (\theta)^{-1}$ (solid lines), the estimator based on the bias-corrected tail dependence function with $\Omega (\theta) = \Sigma (\theta)^{-1}$ (dotted lines), and the pairwise M-estimator from \citet{einmahl2016} (dashed lines). 
We see that for the estimation of the shape parameter $\alpha = 1$ it is better to use one of the estimators based on the empirical stable tail dependence function, whereas for the scale parameter $\rho = 1$ the bias-corrected estimator performs better.

\begin{figure}[p]
\centering
\subfloat{\includegraphics[width=0.3\textwidth]{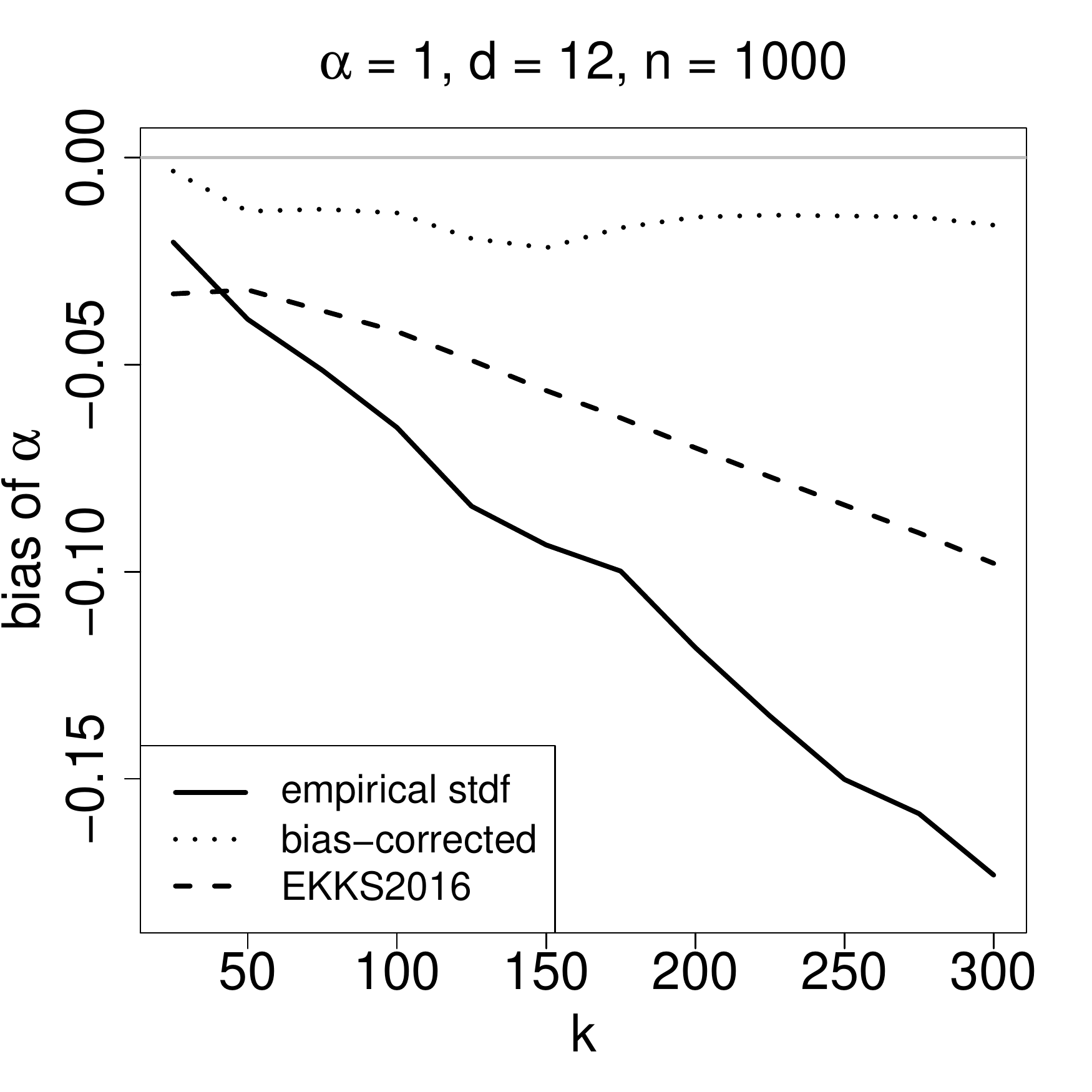}}
\subfloat{\includegraphics[width=0.3\textwidth]{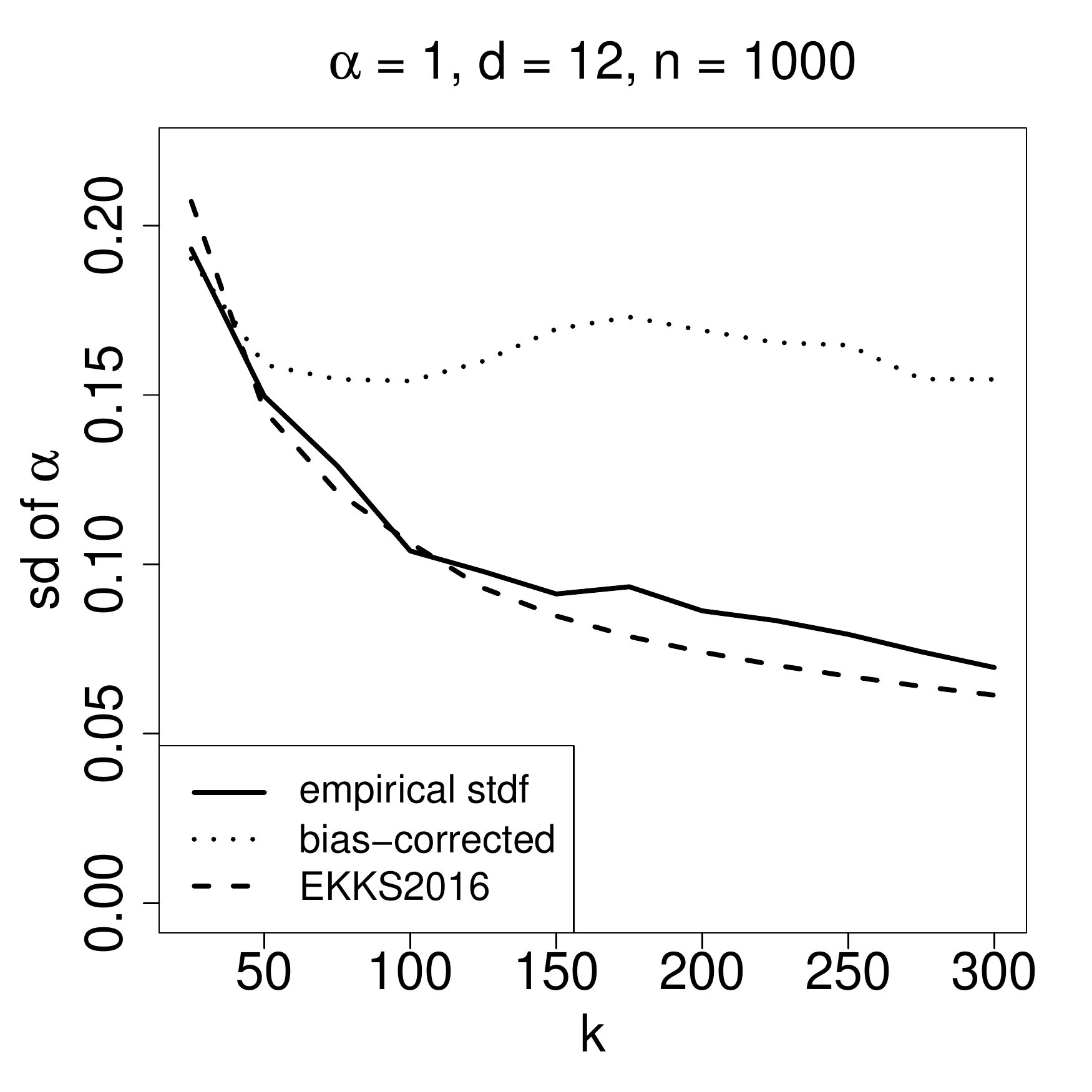}}
\subfloat{\includegraphics[width=0.3\textwidth]{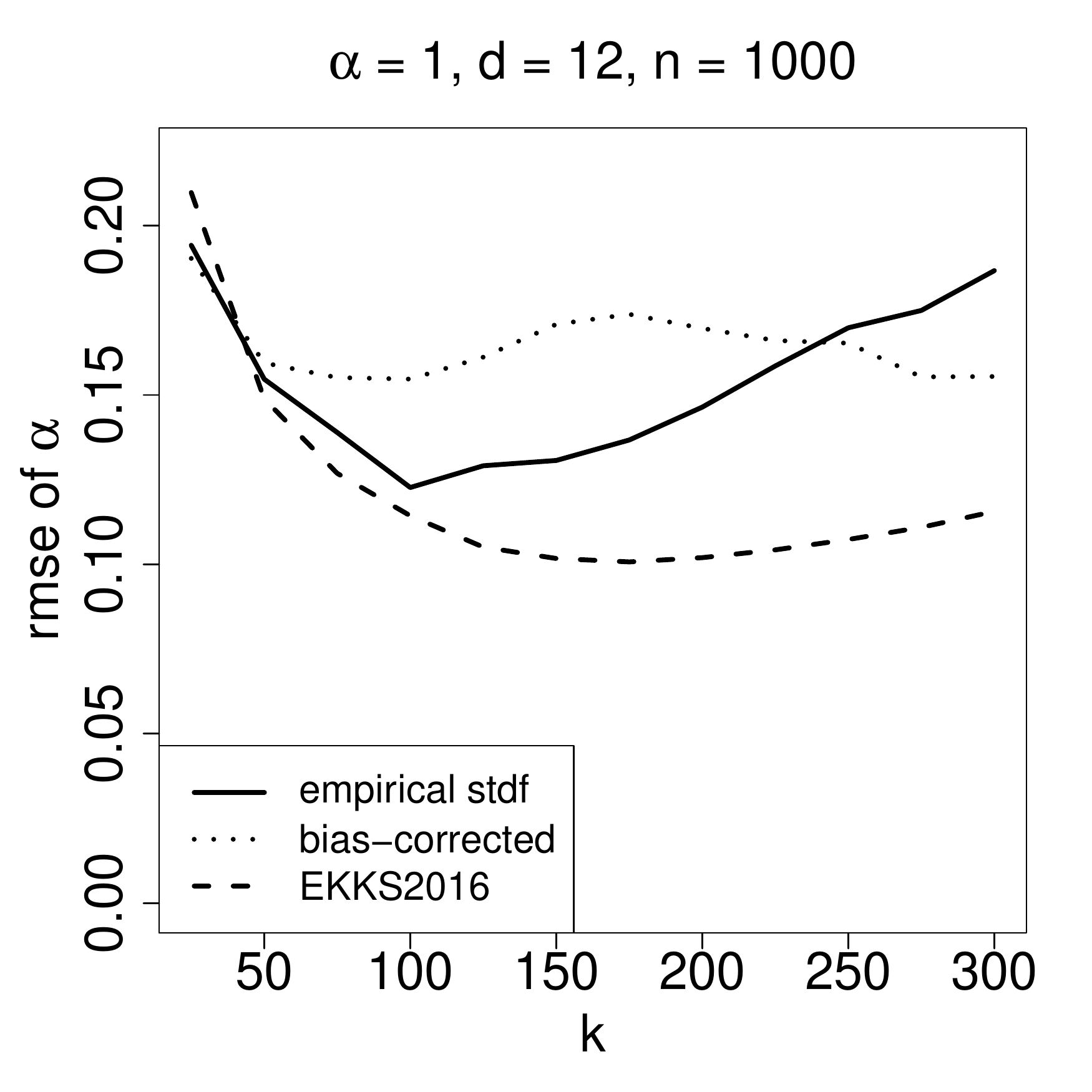}} \\
\subfloat{\includegraphics[width=0.3\textwidth]{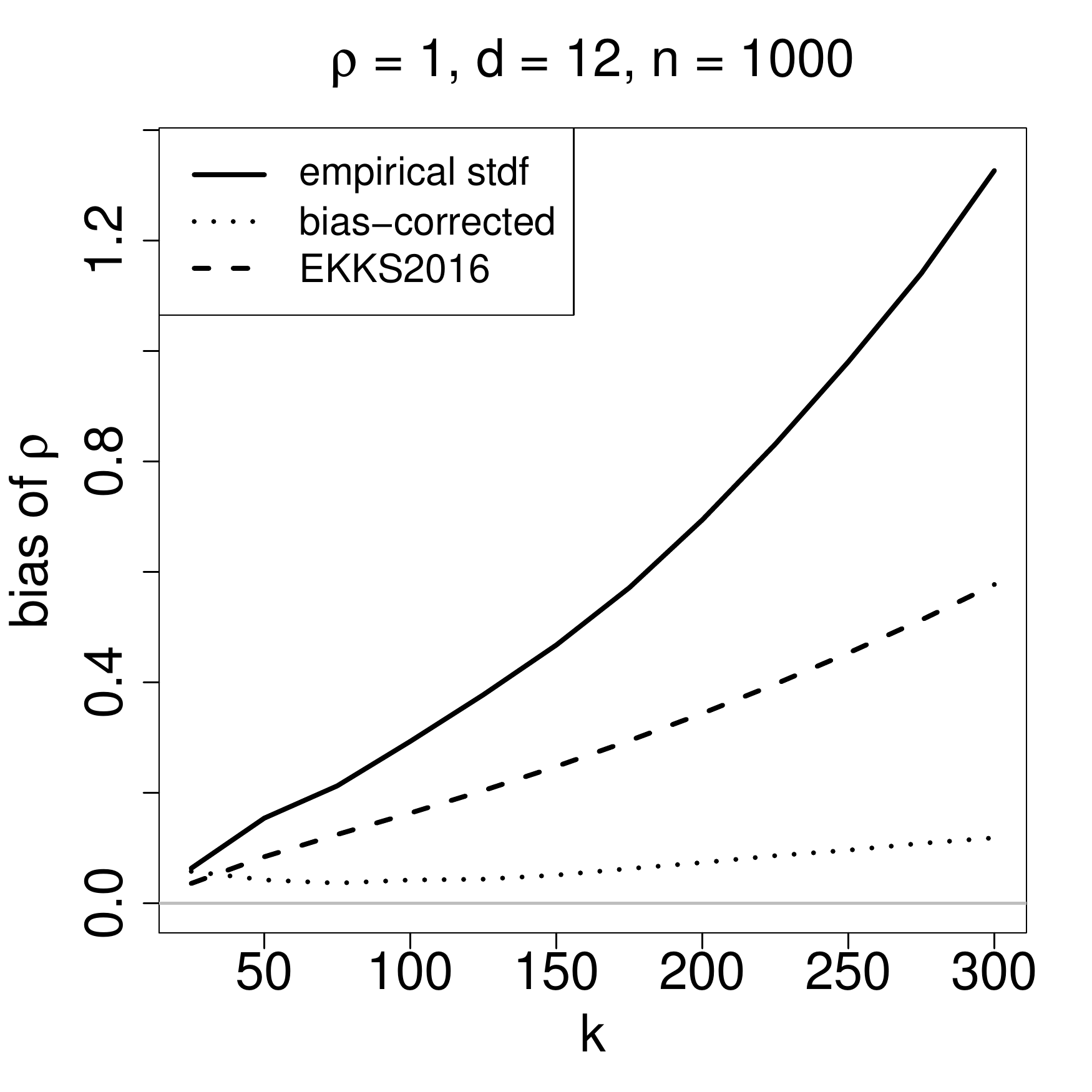}}
\subfloat{\includegraphics[width=0.3\textwidth]{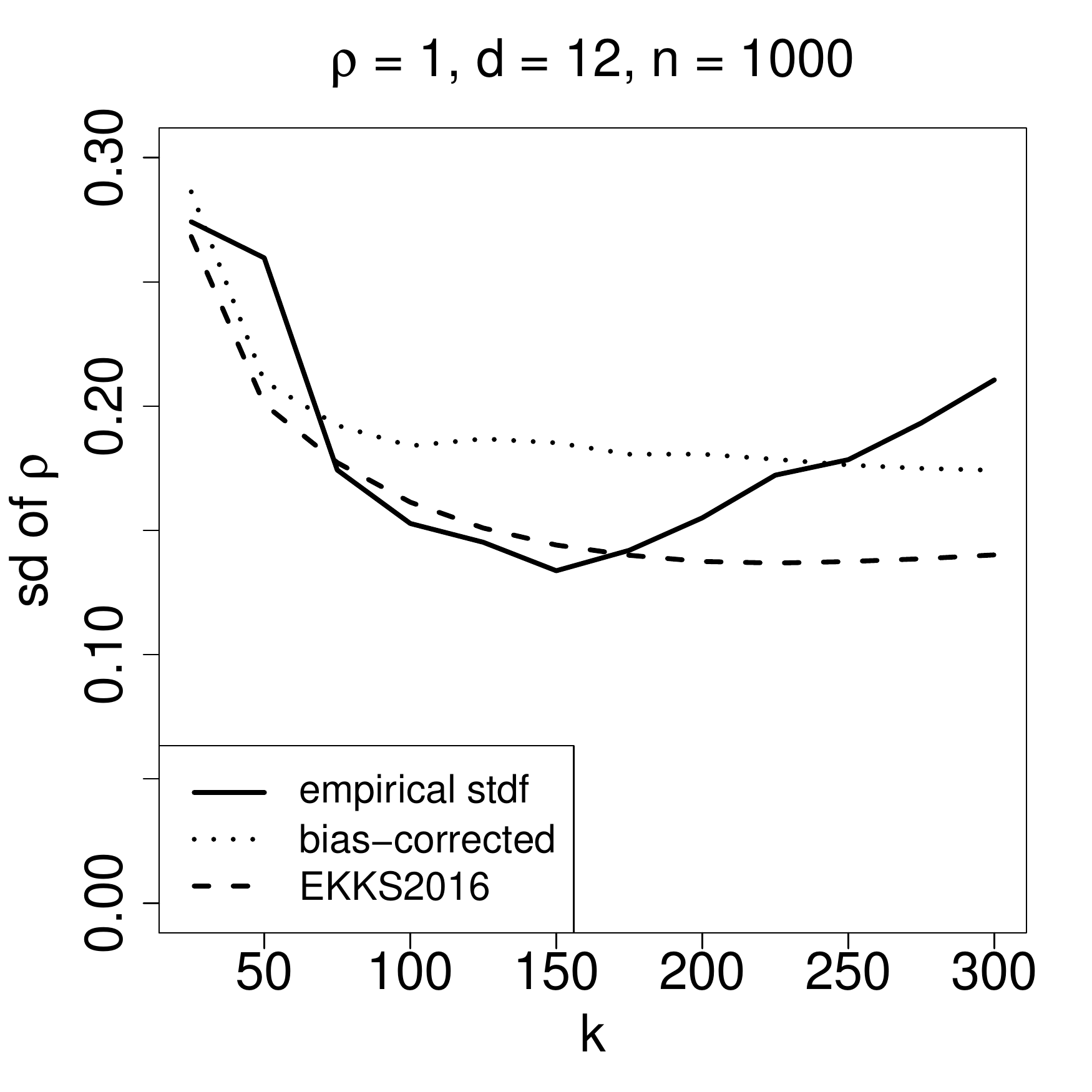}}
\subfloat{\includegraphics[width=0.3\textwidth]{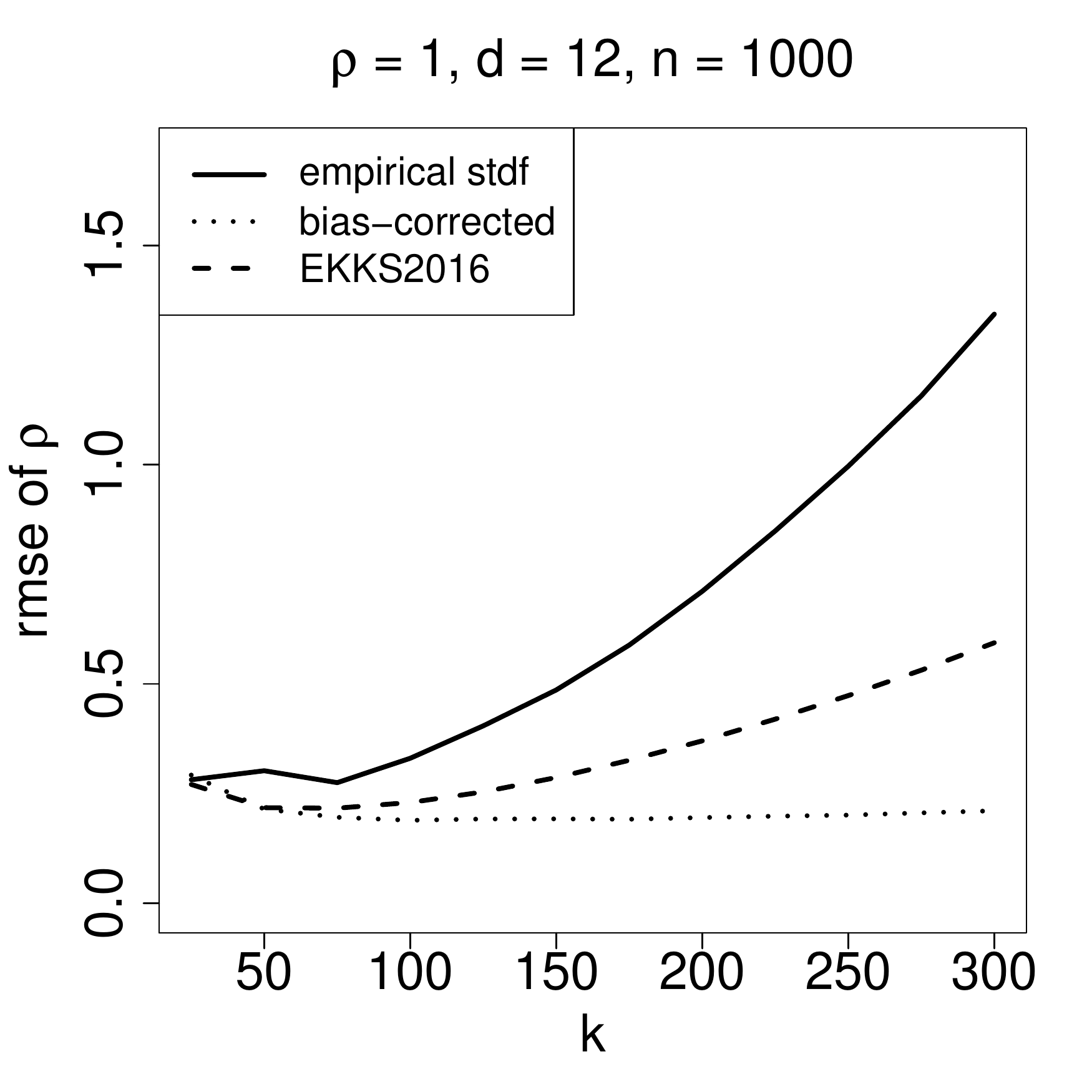}} \\
\subfloat{\includegraphics[width=0.3\textwidth]{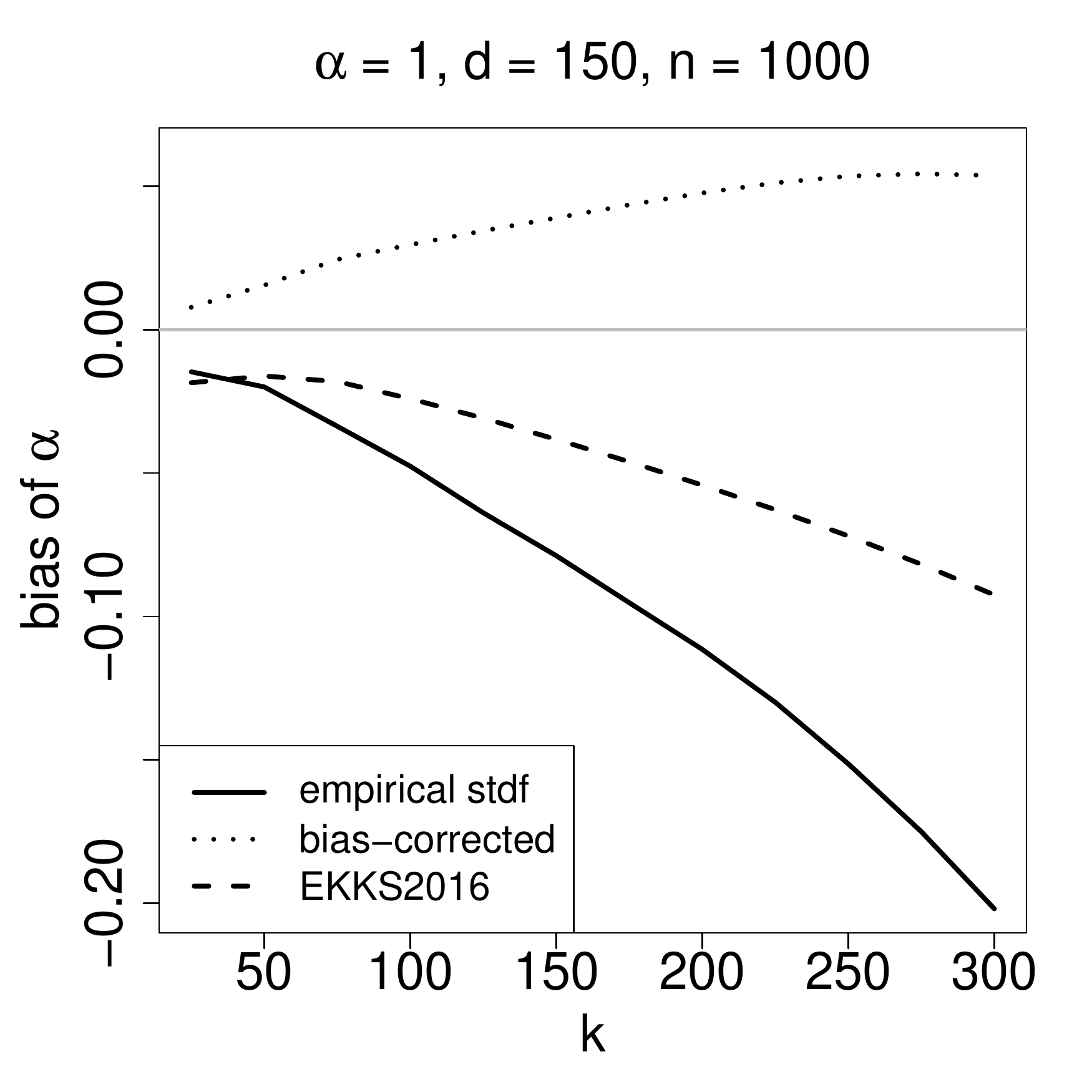}}
\subfloat{\includegraphics[width=0.3\textwidth]{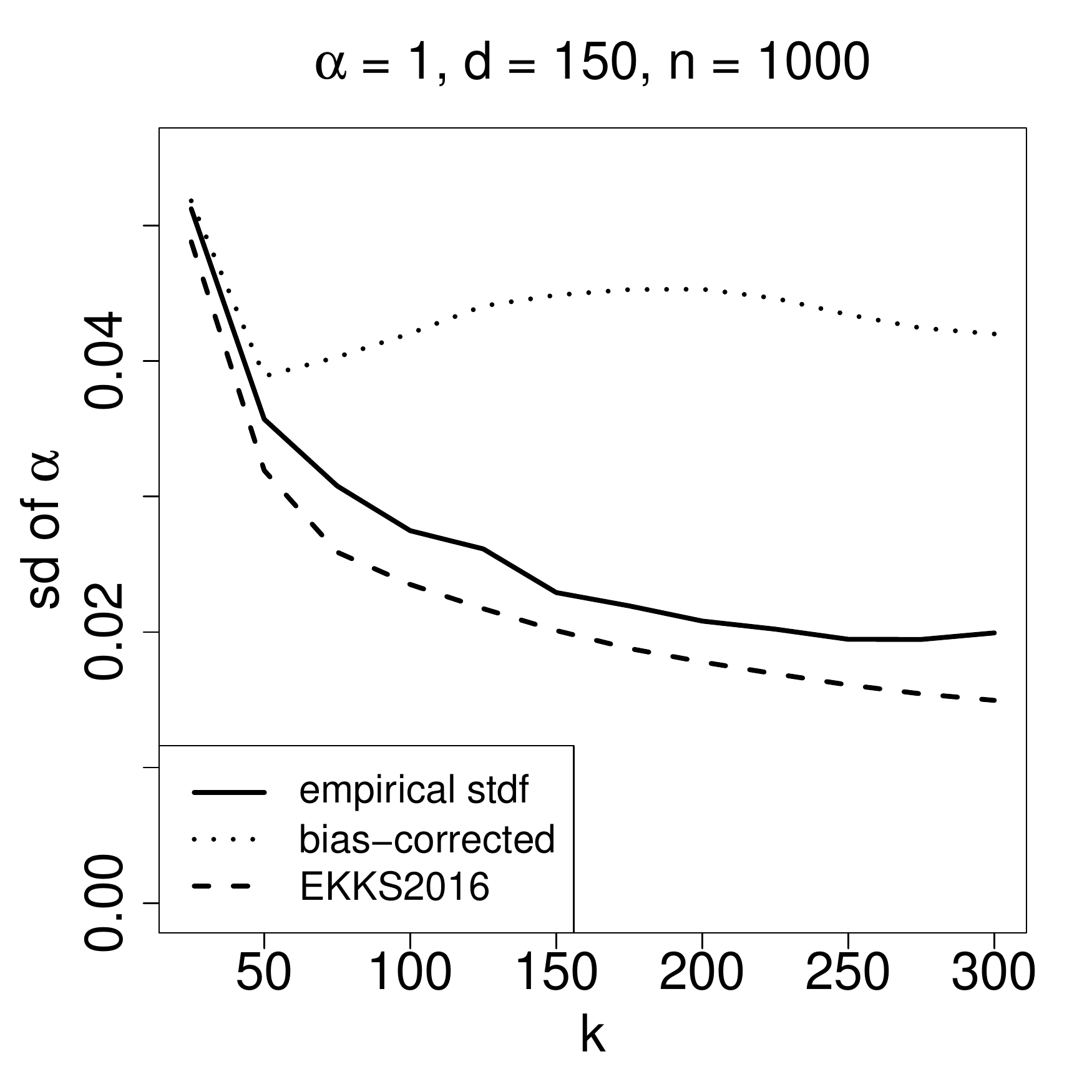}}
\subfloat{\includegraphics[width=0.3\textwidth]{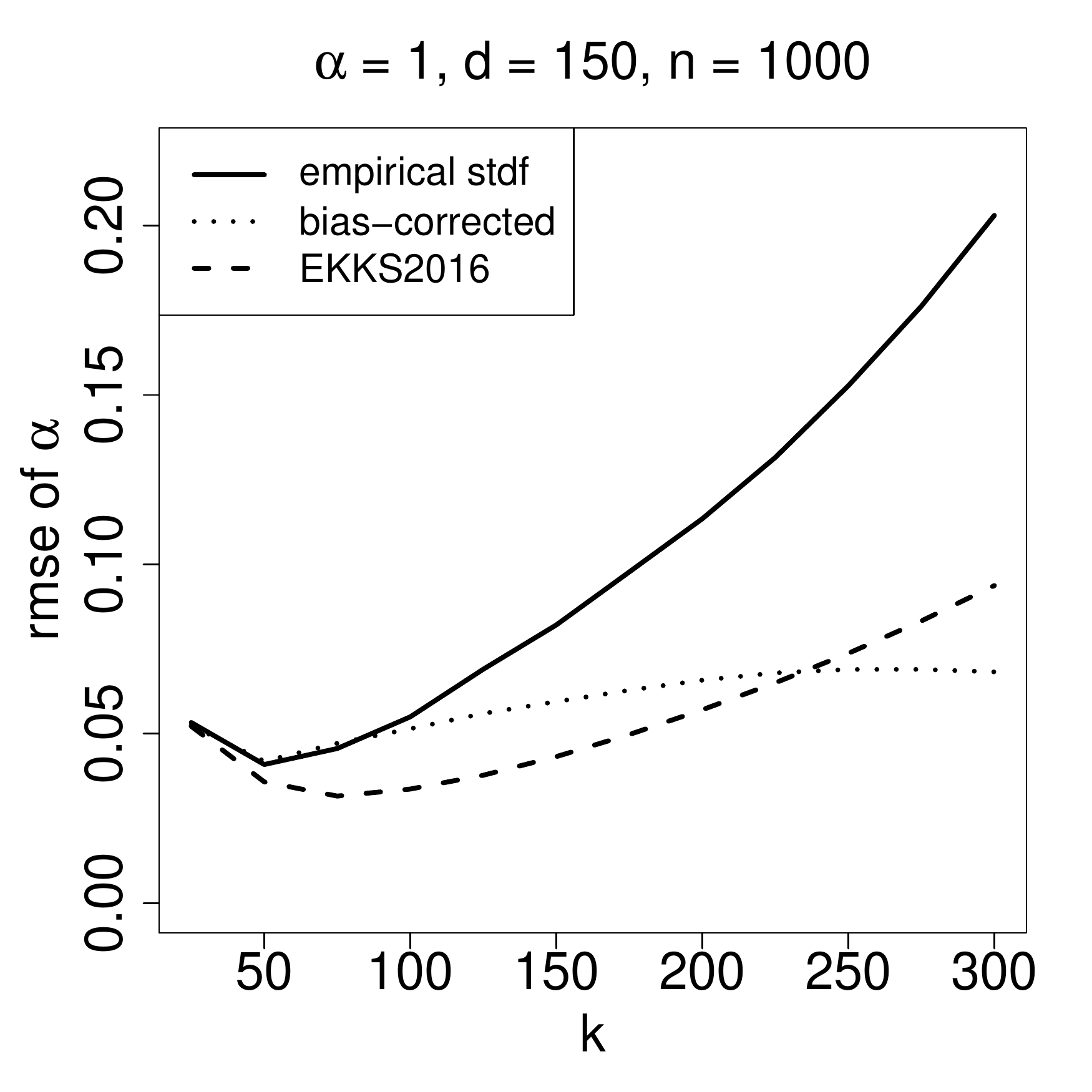}} \\
\subfloat{\includegraphics[width=0.3\textwidth]{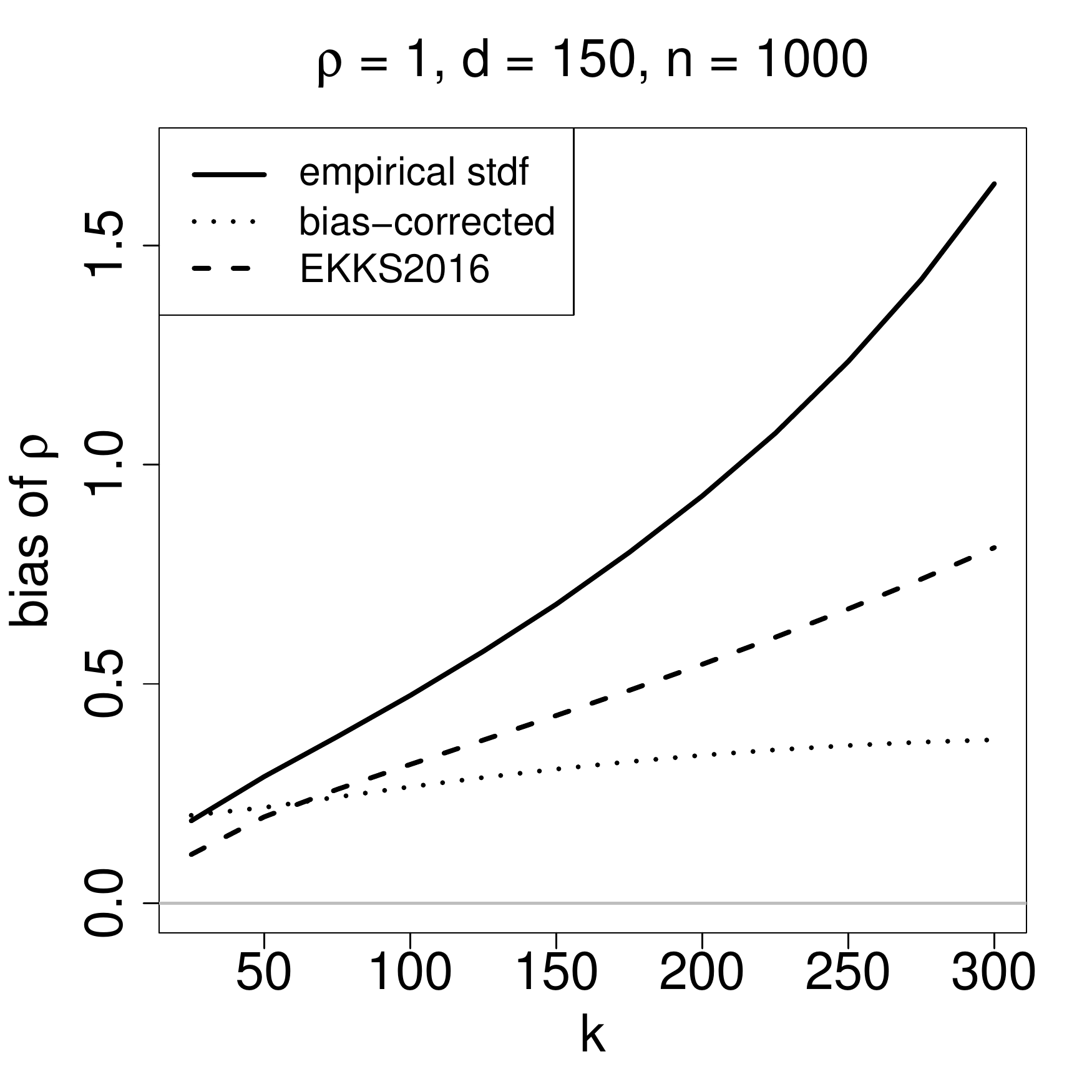}}
\subfloat{\includegraphics[width=0.3\textwidth]{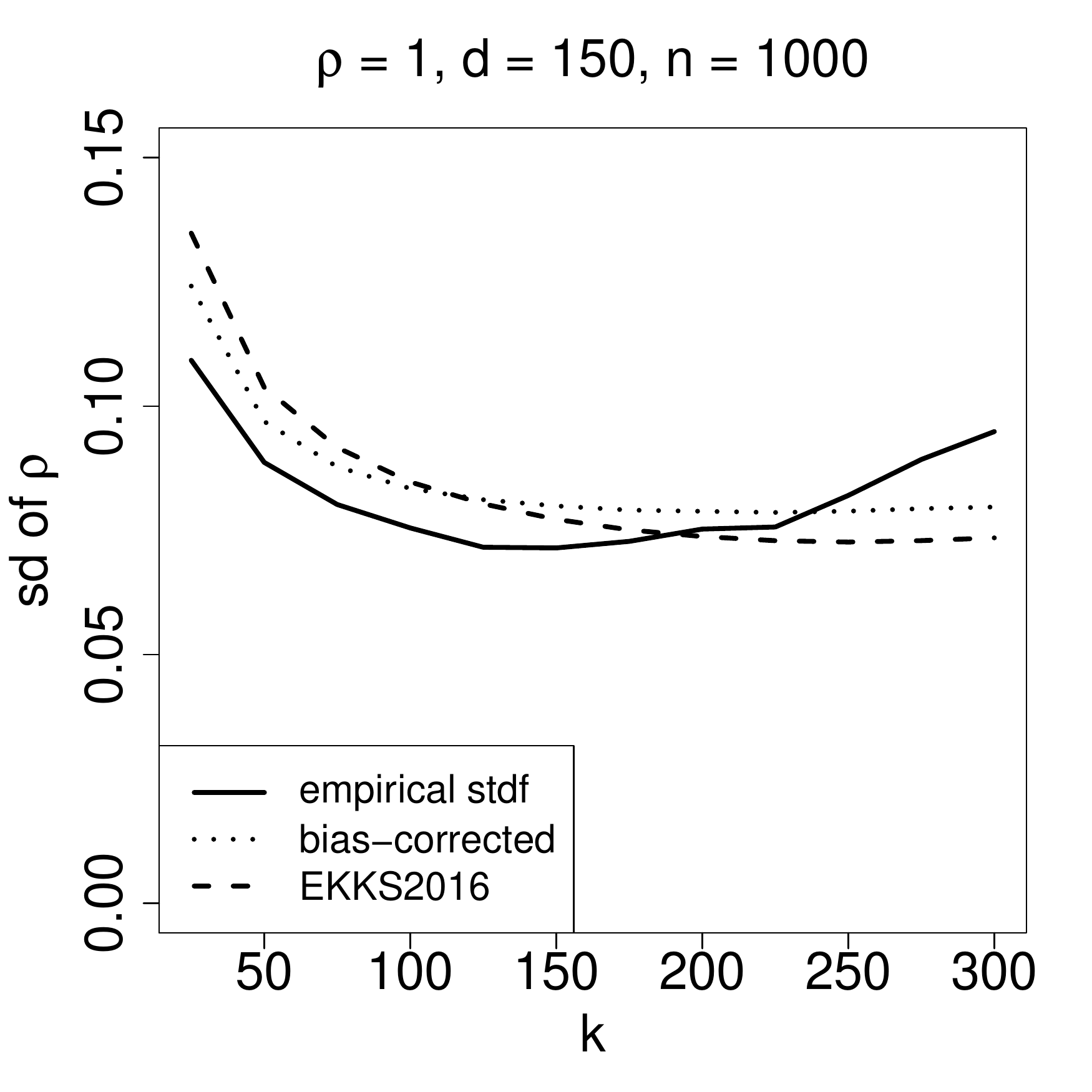}}
\subfloat{\includegraphics[width=0.3\textwidth]{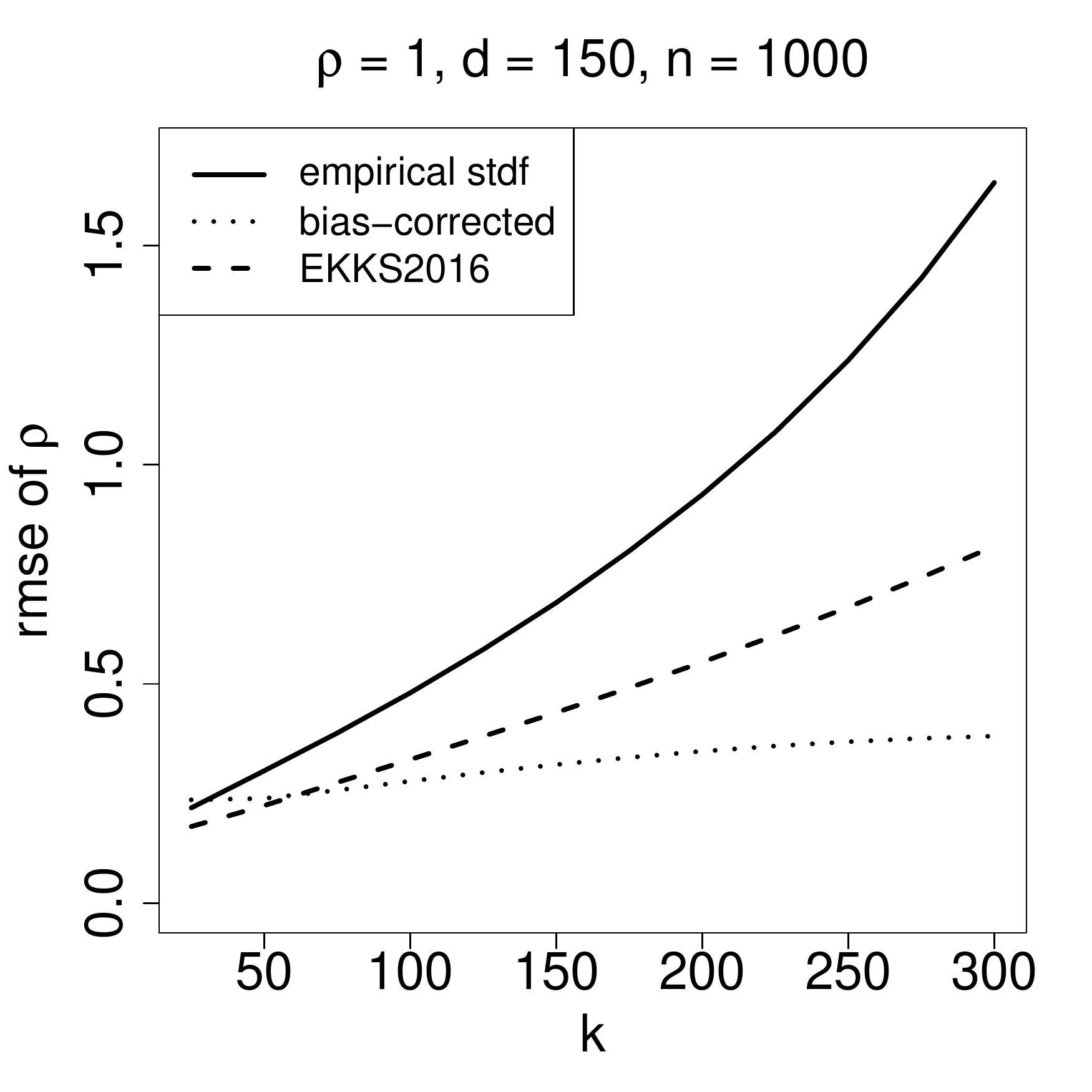}} 
\caption{Brown--Resnick process: bias, standard deviation and RMSE for the estimators in $d=12$ (upper panels) and $d = 150$ (lower panels); 300 samples of size $n = 1000$.}
\label{fig:brplot}
\end{figure}

To show the feasibility of the estimation procedure in high dimensions, we simulate $300$ samples of size $n = 1000$ from the perturbed Brown--Resnick process on a $10 \times 15$ unit-distance grid ($d = 150$), using again $(\alpha,\rho) = (1,1)$ and selecting pairs of neighbouring locations only, yielding $q = 527$ pairs in total. The bottom panels of Figure \ref{fig:brplot} show the bias, standard deviation and RMSE for the estimator based on the empirical tail dependence function with $\Omega (\theta) = I_q$ (solid lines), the estimator based on the bias-corrected tail dependence function with $\Omega (\theta) = I_q$ (dotted lines), and the pairwise M-estimator from \citet{einmahl2016} (dashed lines). Compared to $d = 12$ above, the estimation of $\alpha$ has improved whereas the estimation quality of $\rho$ stays roughly the same.

\subsection{Max-linear models on directed acyclic graphs}
\label{sec:maxlinear}

A max-linear or max-factor model has stable tail dependence function 
\begin{equation}\label{eq:mlstdf}
  \ell(x) =   \sum_{t=1}^r \max_{j=1,\ldots,d}{b_{jt} x_j},  \qquad x \in [0,\infty)^d,
\end{equation} 
where the factor loadings $b_{jt}$ are non-negative constants such that $\sum_{t=1}^r b_{jt} = 1$ for every $j \in \{1,\ldots,d\}$ and all column sums of the $d \times r$ matrix $B := (b_{jt})_{j,t}$ are positive \citep{einmahl2012}. 
An example of a random vector $Y = (Y_1, \ldots, Y_d)$ that has tail dependence function \eqref{eq:mlstdf} is $Y_j = \max_{t=1,\ldots,r}{b_{jt} Z_t}$ for $j \in \{1,\ldots,d\}$,
where $Z_1,\ldots,Z_r$ are independent unit Fr\'{e}chet variables. The random variables $Y_j$ are then unit Fr\'echet as well.

Since the rows of $B$ sum up to one, it has only $d \times (r-1)$ free elements. Further structure may be added to the coefficient matrix $B$, leading to parametric models whose parameter dimension is lower than 
$d \times (r-1)$; see below. Even then, the map $L$ in \eqref{eq:L} induced by restricting the points $c_m$ to be of the form $e_J$ in \eqref{ej} is typically not one-to-one. Therefore, we need more general choices of the points $c_m$ in the definition of the estimator.

In \citet{gissibl2015}, a link is established between max-linear models and structural equation models, from which graphical models based on directed acyclic graphs (DAGs) can be constructed. A max-linear structural equation model is defined via
\begin{equation*}
Y_j = \max_{k \in \mathrm{pa}(j)} u_{kj} Y_k \vee u_{j} Z_j, \qquad j = 1,\ldots,d,
\end{equation*}
where $\mathrm{pa}(j) \subset \{1,\ldots,d\}$ denotes the set of parents of node $j$ in the graph, $u_{kj} > 0$ for all $k \in \mathrm{pa}(j) \cup \{j\}$ and $u_{j} > 0$ for all $j \in \{1,\ldots,d\}$. We let $Z_1,\ldots,Z_d$ be independent unit Fr\'echet random variables. A max-linear structural equation model can then be written as a max-linear model with parameters determined by the paths of the corresponding graph.

We focus on the four-dimensional model corresponding to the following directed acyclic graph \citep[Example 2.1]{gissibl2015}:

\begin{minipage}{0.43\textwidth}
\begin{figure}[H]
\includegraphics[width=0.9\textwidth]{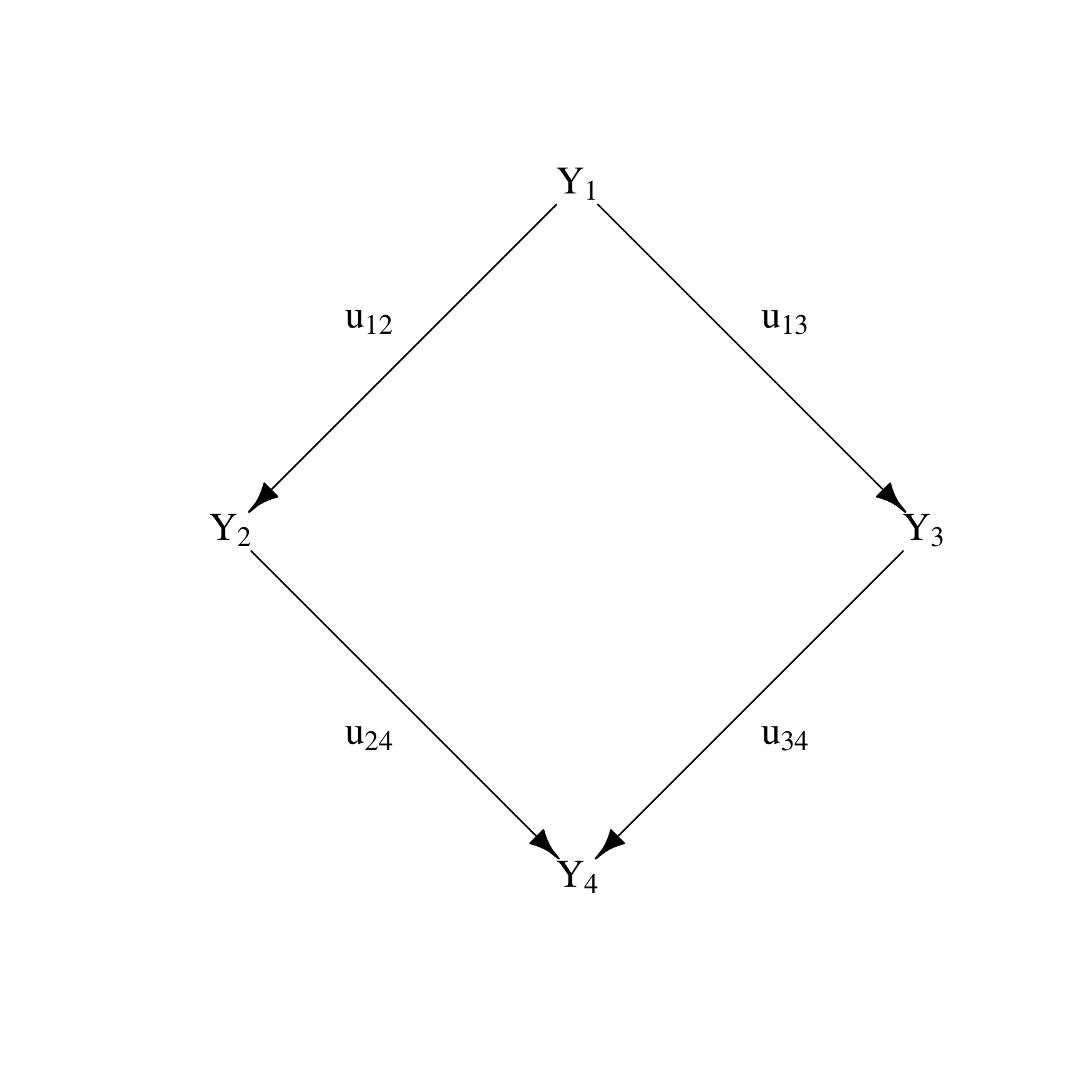}
\end{figure}
\end{minipage} \hspace{-1cm}
\begin{minipage}{0.55\textwidth}
\footnotesize
\begin{align*}
  Y_1 & = u_{1} Z_1, \\
  Y_2 & = u_{12} Y_1 \vee u_{2} Z_2 = u_{12} u_{1} Z_1 \vee u_{2} Z_2, \\
  Y_3 & = u_{13} Y_1 \vee u_{3} Z_3 = u_{13} u_{1} Z_1 \vee u_{3} Z_3, \\
  Y_4 & = u_{24} Y_2 \vee u_{34} Y_3 \vee u_{4} Z_4  \\
&  = (u_{24} u_{12} u_1 \vee u_{34} u_{13} u_1 ) Z_1 \vee u_{24} u_{2} Z_2 \vee u_{34} u_{3} Z_3 \vee u_4 Z_4.
\end{align*} 
\end{minipage} 

If we require $Y_1,\ldots,Y_4$ to be unit Fr\'echet, the matrix of factor loadings becomes
\begin{equation*}
B = \begin{pmatrix}
1 & 0 & 0 & 0 \\
u_{12} & u_2 & 0 & 0 \\
u_{13} & 0 & u_3 & 0 \\
u_{12} u_{24} \vee u_{13} u_{34} & u_2 u_{24} & u_3 u_{34} & u_4 \\ 
\end{pmatrix},
\end{equation*}
where the diagonal elements $u_j$ for $j \in \{ 2,3,4 \}$ are such that the row sums are equal to one. The parameter vector is then given by $\theta = (u_{12}, u_{13}, u_{24}, u_{34})$.

We conduct a simulation study based on $300$ samples of size $n = 1000$ from the four-dimensional model with tail dependence function \eqref{eq:mlstdf} and $B$ as above, with parameter vector $\theta=  (0.3,0.8,0.4,0.55)$. As before, we put $X_{ij} = Y_{ij} + |\epsilon_{ij}|$, with $(Y_{i1},\ldots,Y_{id})$ as above and $\epsilon_{ij}$ independent $\mathcal{N}(0,1/4)$ random variables.
The estimators are based on the $q = 72$ points $c_m$ on the grid $\{0,1/2,1\}^4$ having at least two positive coordinates. 

Figure~\ref{fig:maxlin2} shows the bias, standard deviation and RMSE for the estimator based on the empirical tail dependence function with $\Omega (\theta) = \Sigma (\theta)^{-1}$ (solid lines), the estimator based on the bias-corrected tail dependence function with $\Omega (\theta) = \Sigma (\theta)^{-1}$ (dotted lines) and the pairwise M-estimator from \citet{einmahl2016} (dashed lines).   
The difference between the pairwise M-estimator and our estimators based on the empirical tail dependence function is negligible. The estimators based on the empirical tail dependence function perform better than the ones based on the bias-corrected version, especially for the parameters $u_{13}$ and $u_{24}$.

\begin{figure}[p]
\centering
\subfloat{\includegraphics[width=0.3\textwidth]{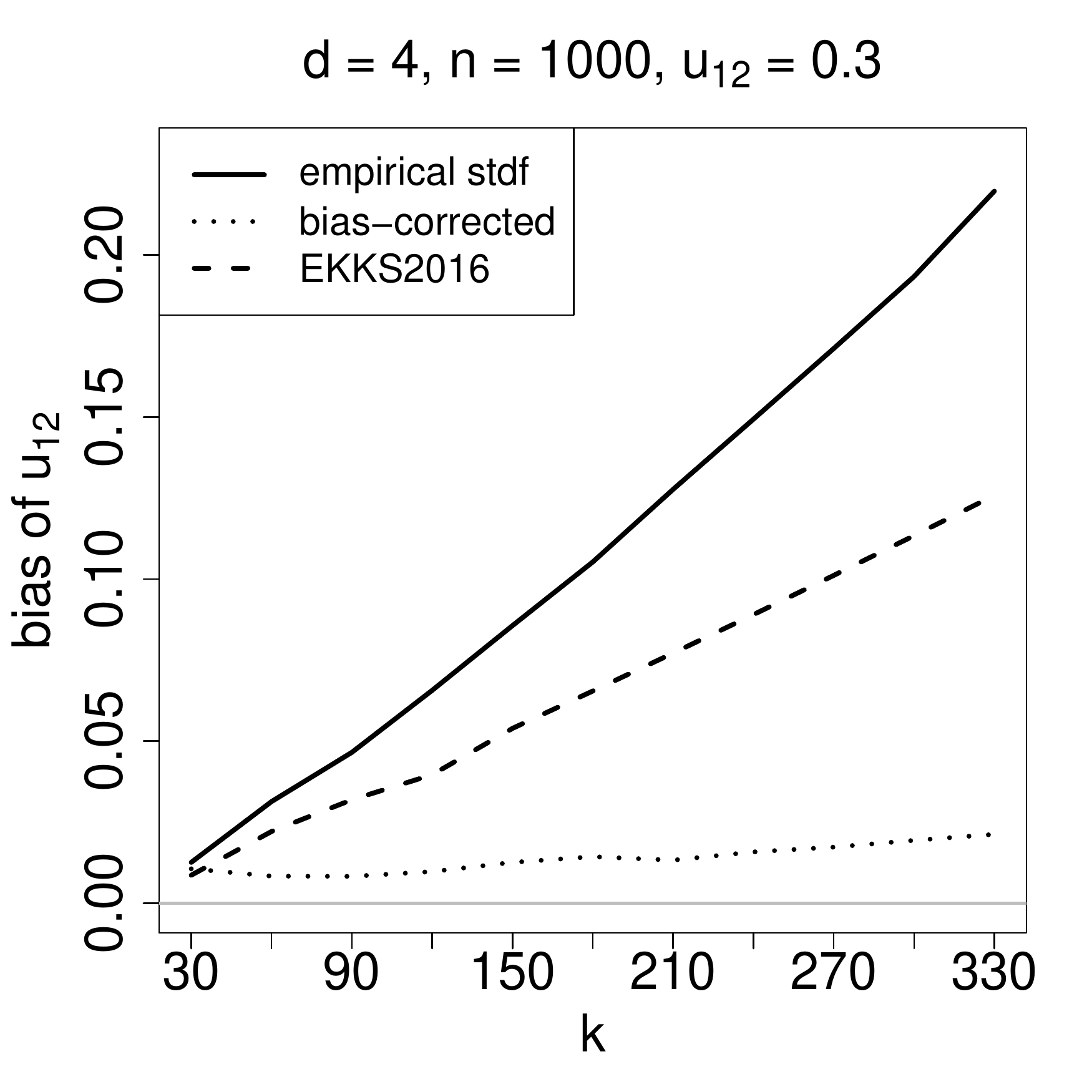}}
\subfloat{\includegraphics[width=0.3\textwidth]{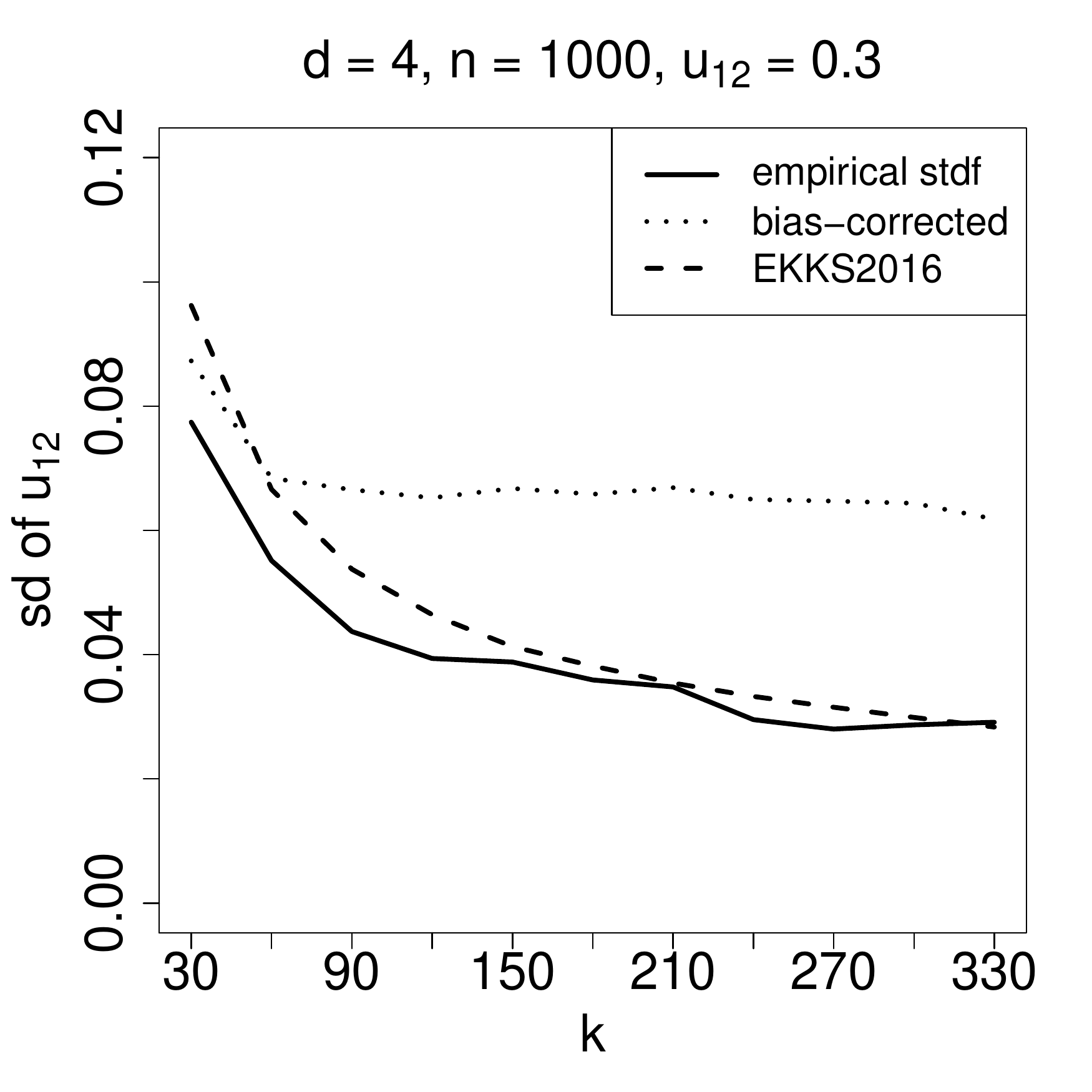}}
\subfloat{\includegraphics[width=0.3\textwidth]{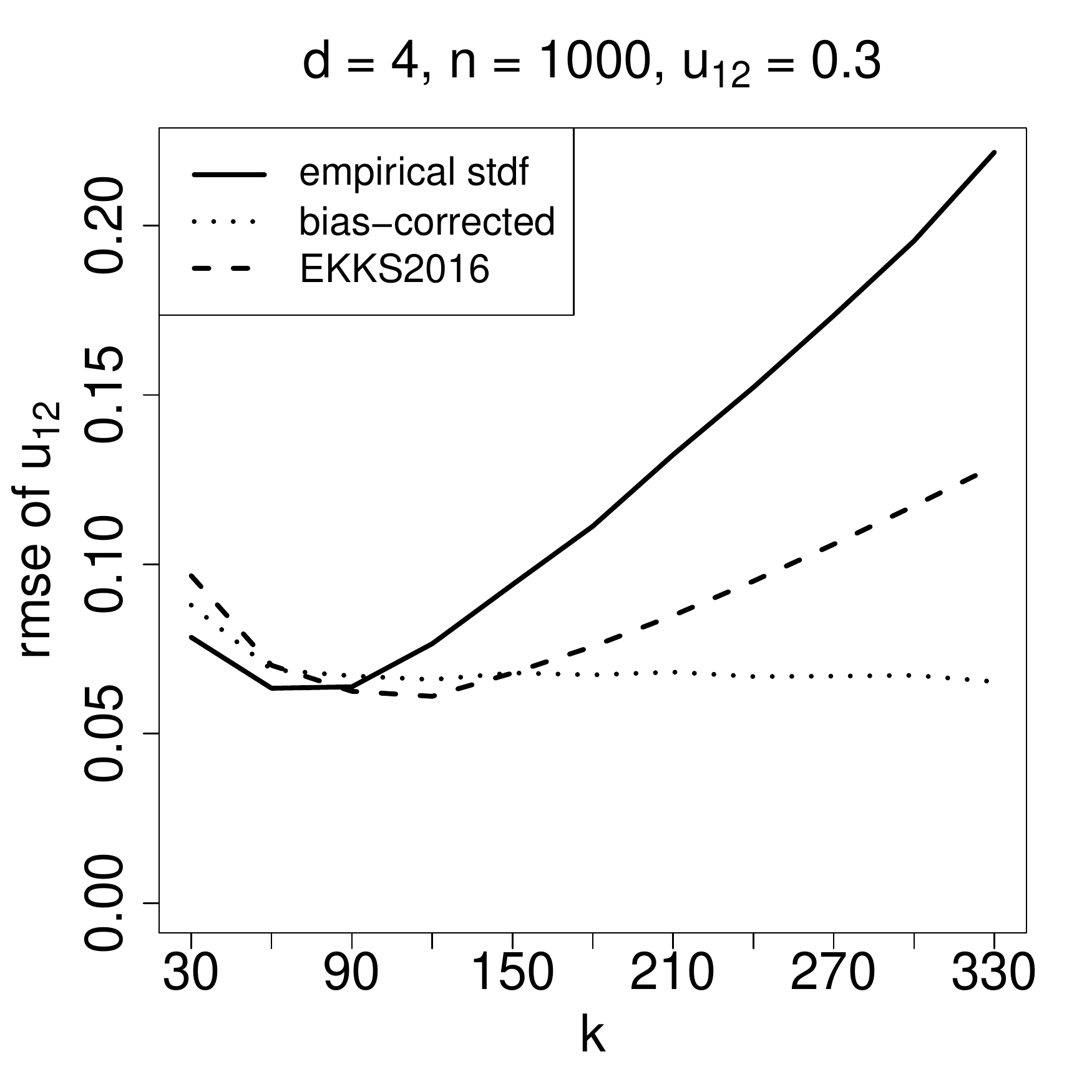}} \\
\subfloat{\includegraphics[width=0.3\textwidth]{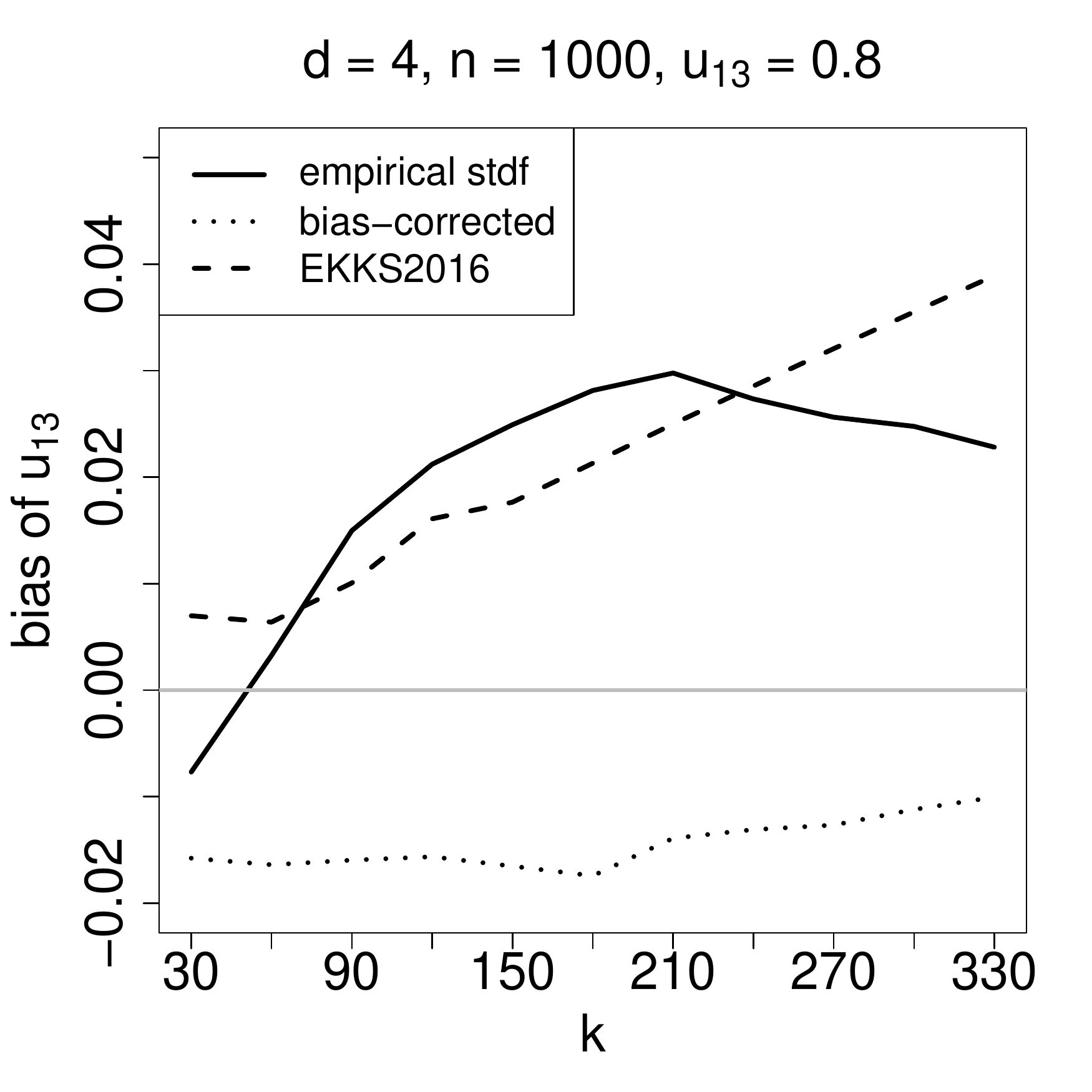}}
\subfloat{\includegraphics[width=0.3\textwidth]{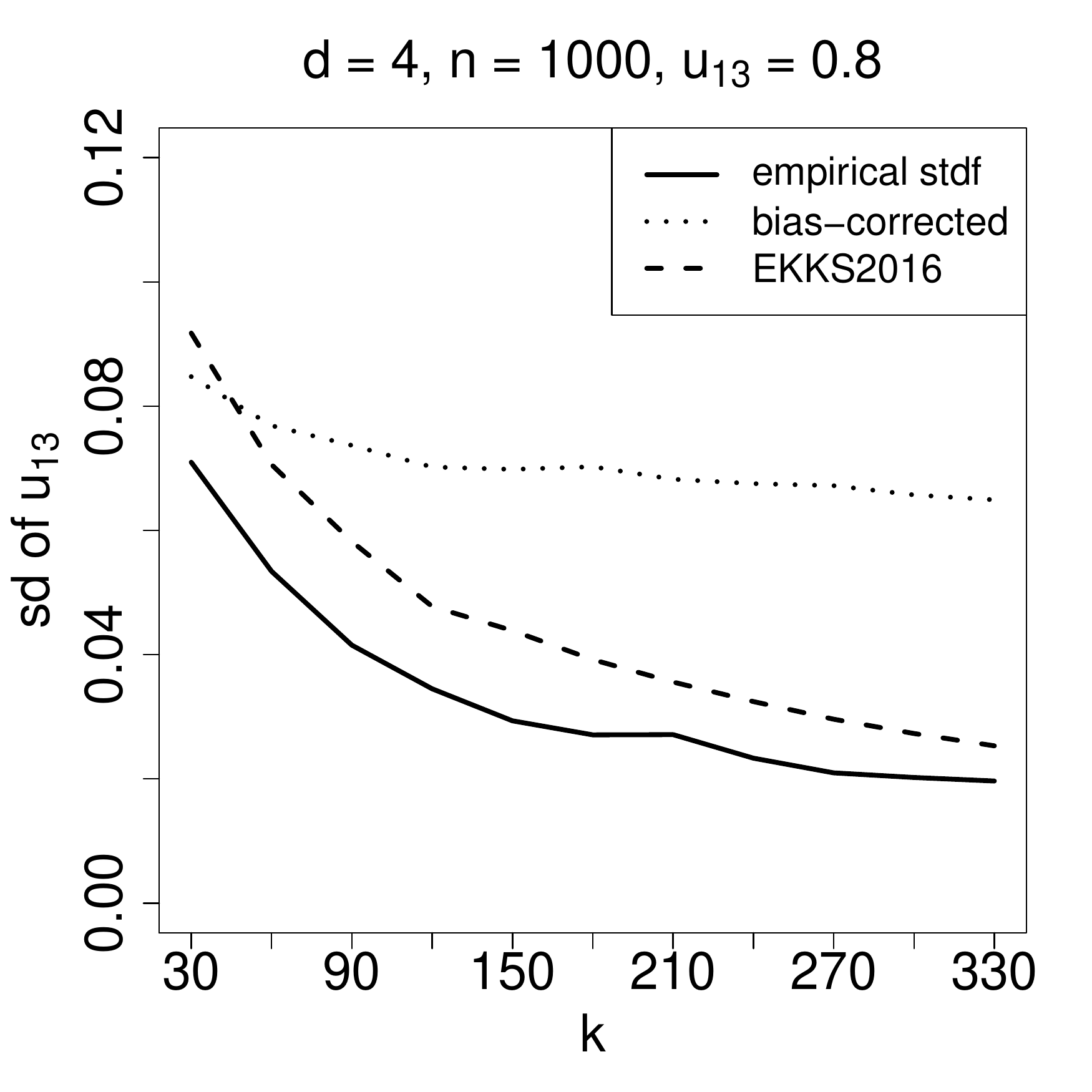}}
\subfloat{\includegraphics[width=0.3\textwidth]{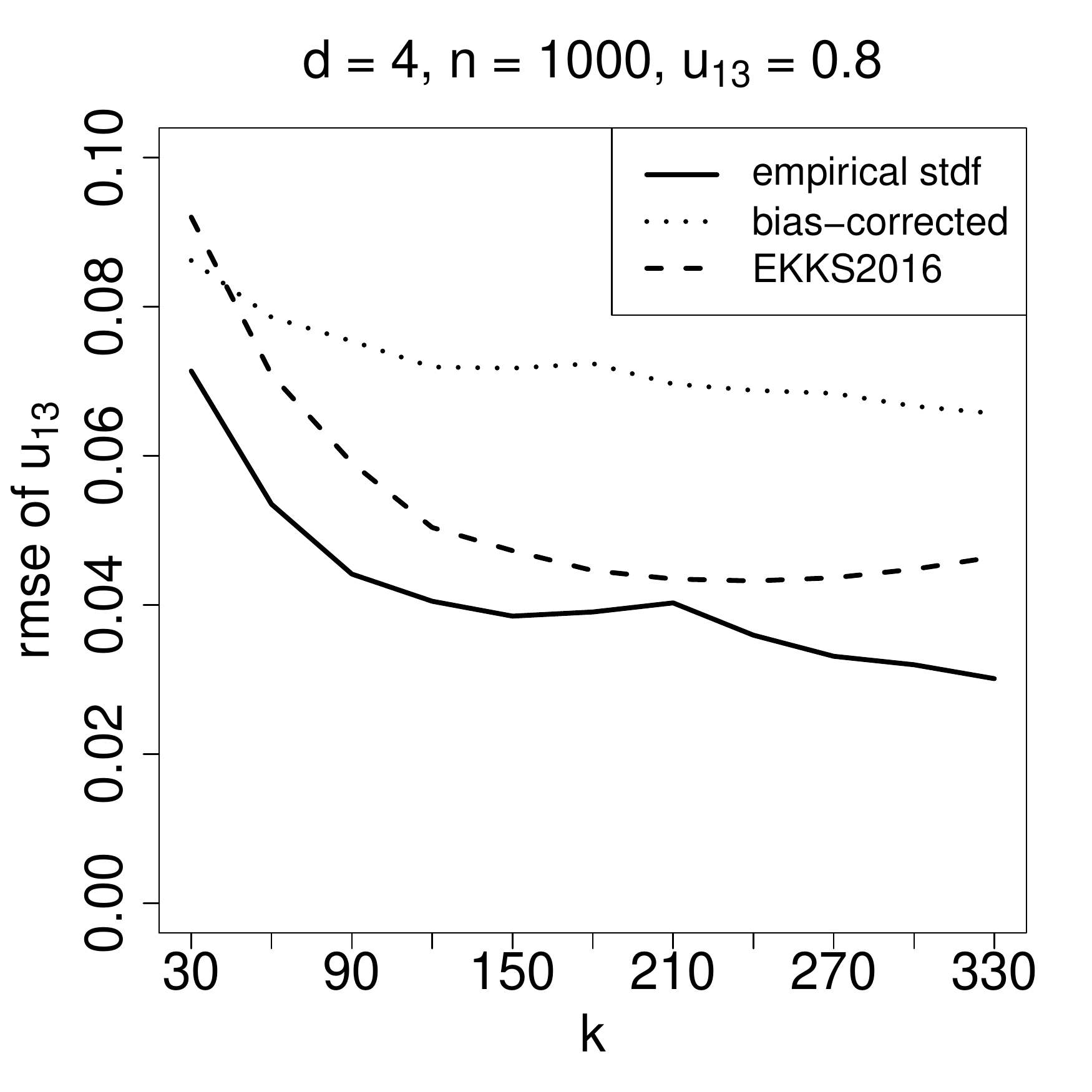}} \\
\subfloat{\includegraphics[width=0.3\textwidth]{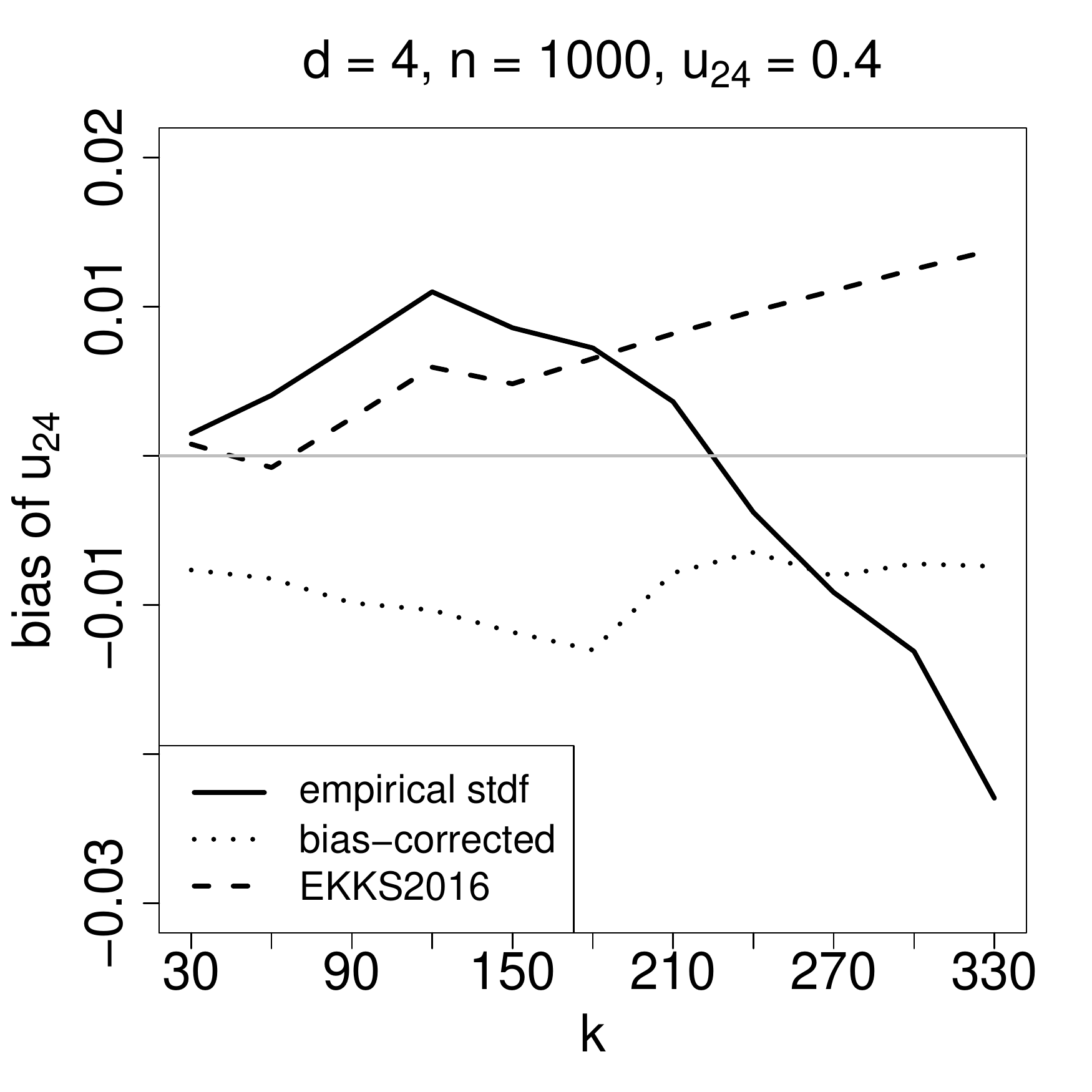}}
\subfloat{\includegraphics[width=0.3\textwidth]{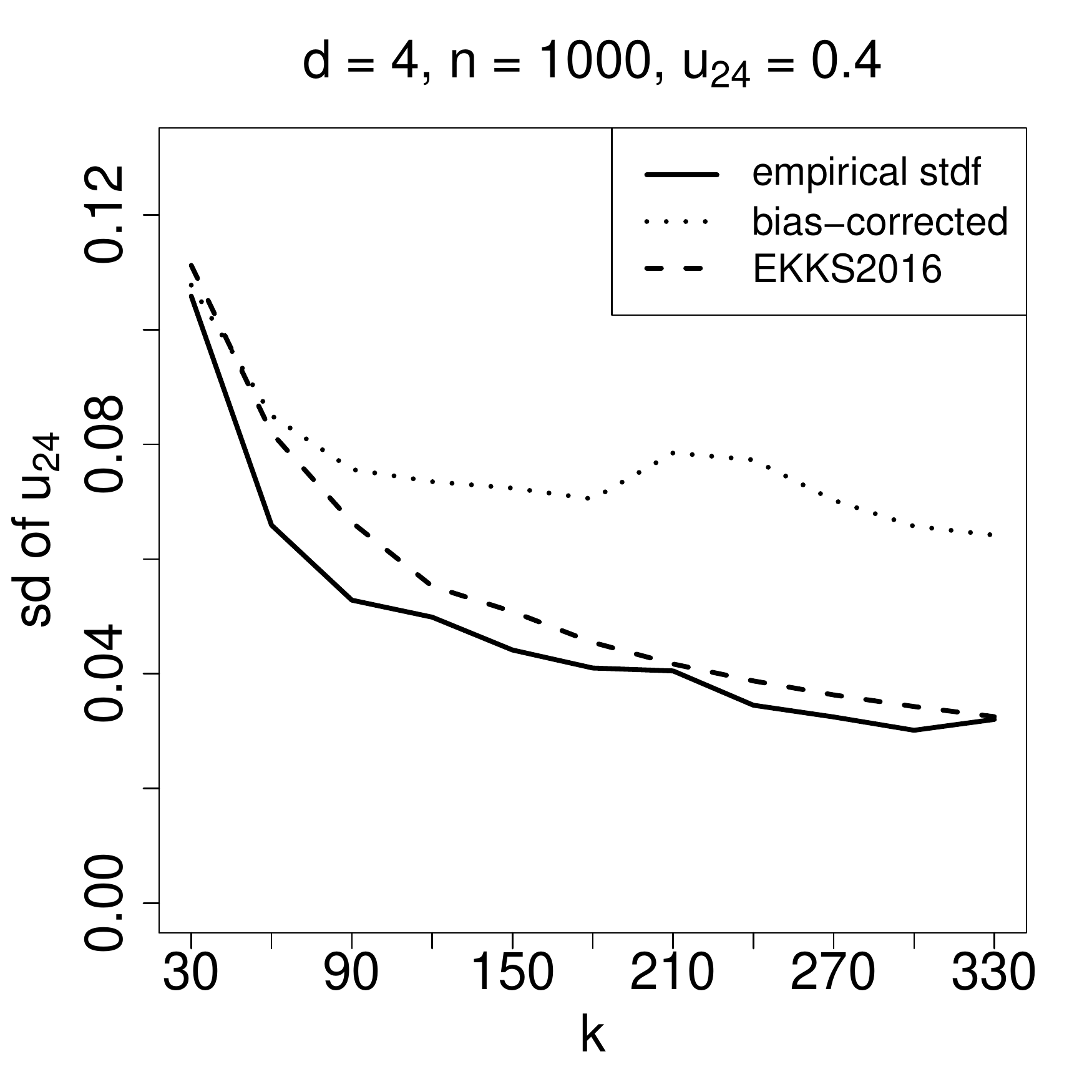}}
\subfloat{\includegraphics[width=0.3\textwidth]{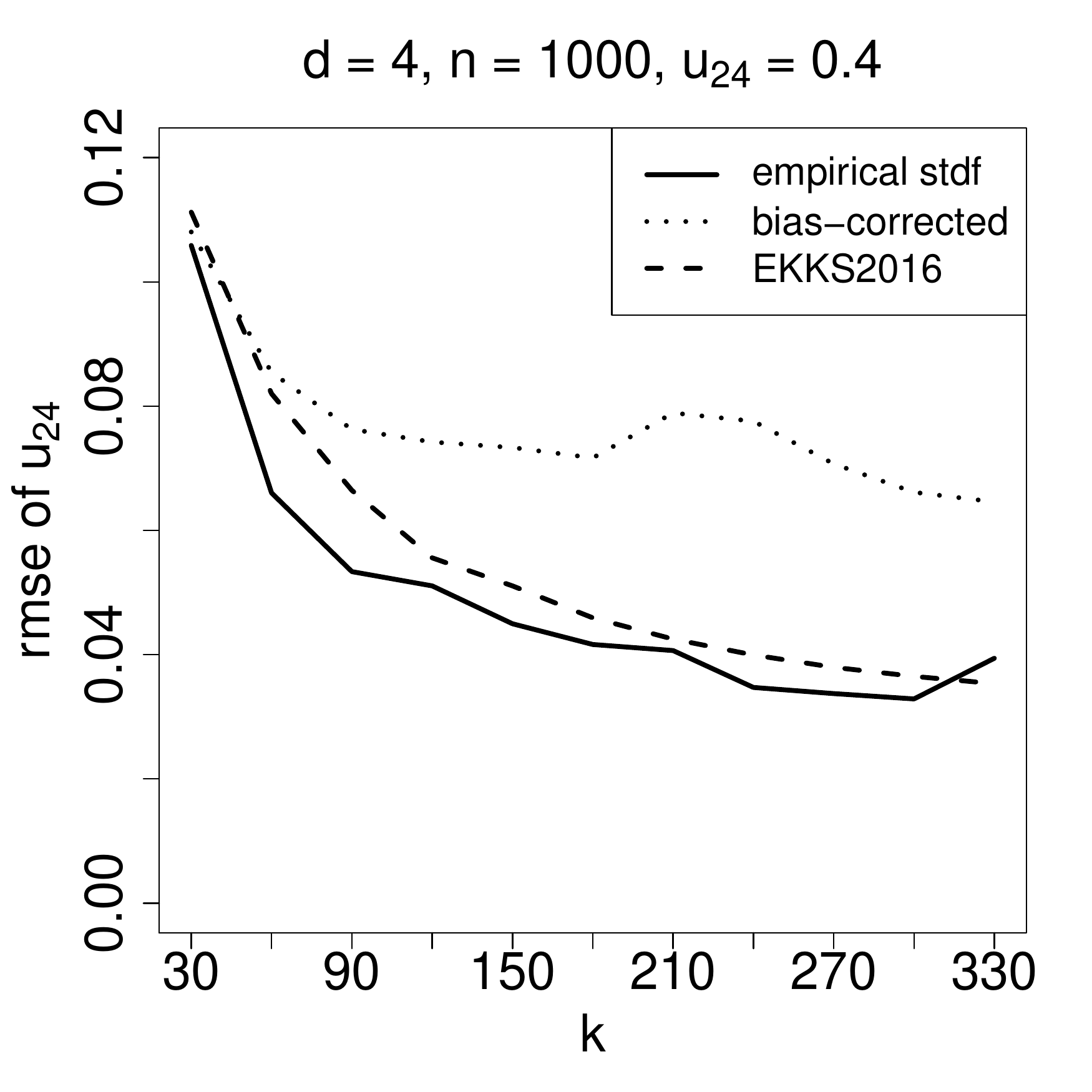}} \\
\subfloat{\includegraphics[width=0.3\textwidth]{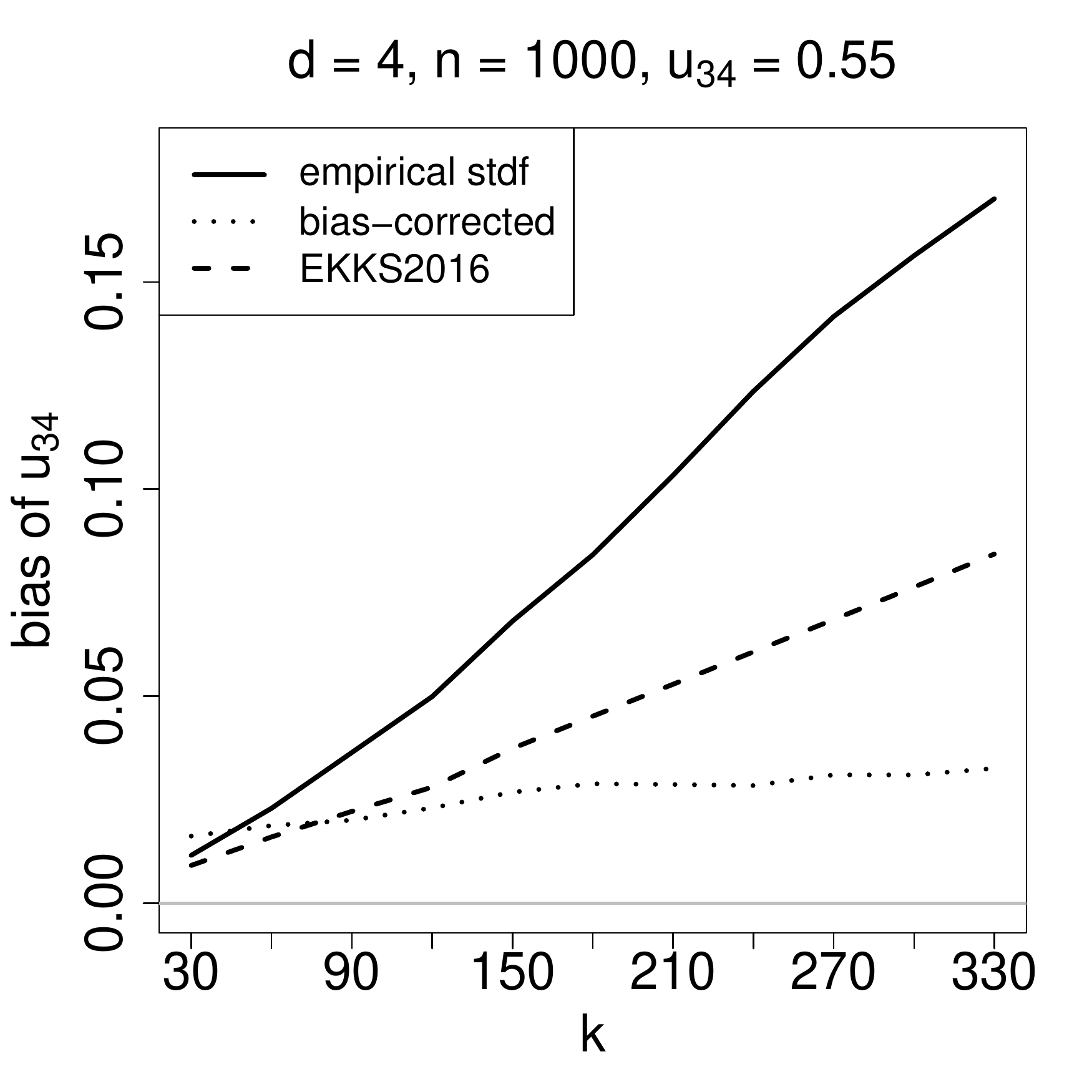}}
\subfloat{\includegraphics[width=0.3\textwidth]{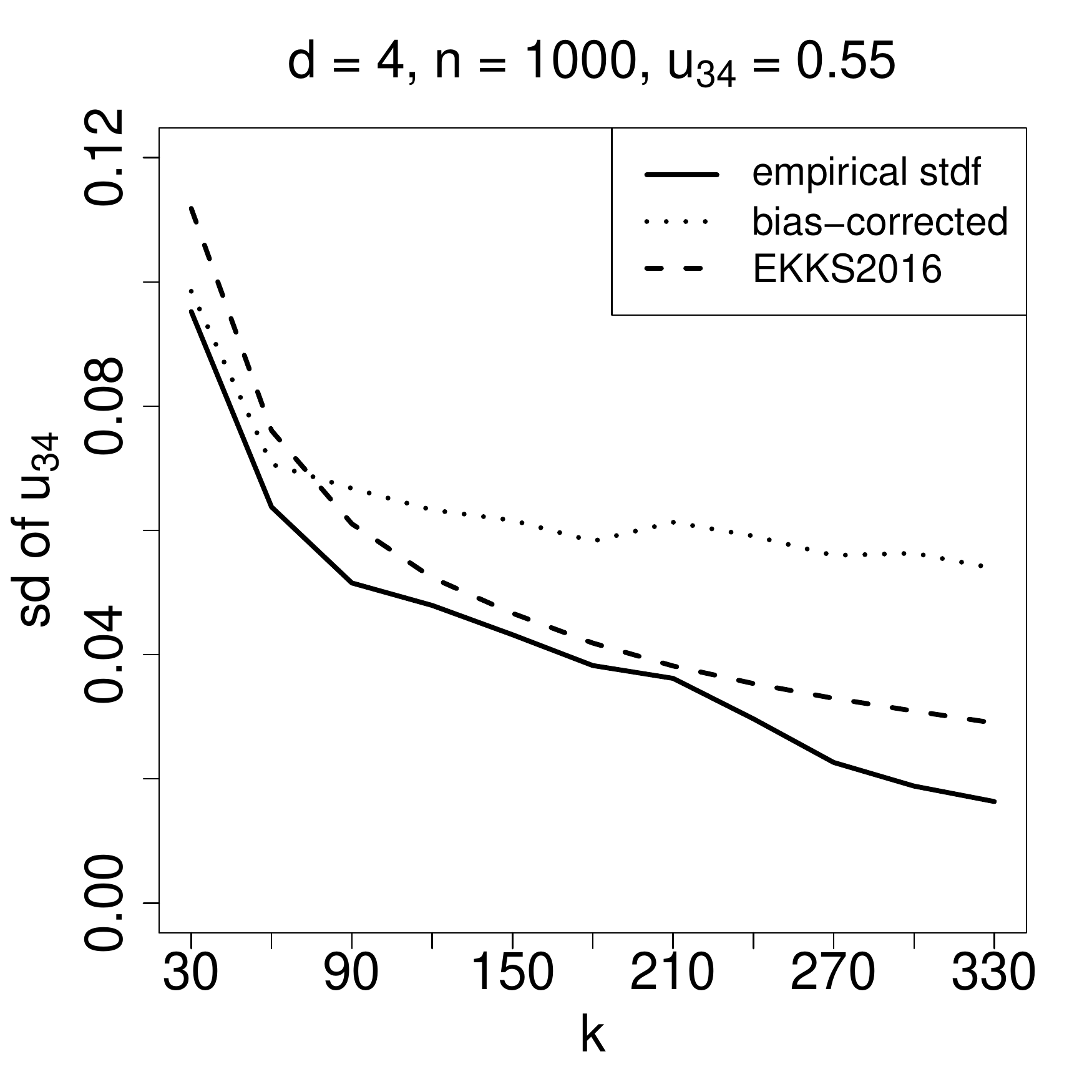}}
\subfloat{\includegraphics[width=0.3\textwidth]{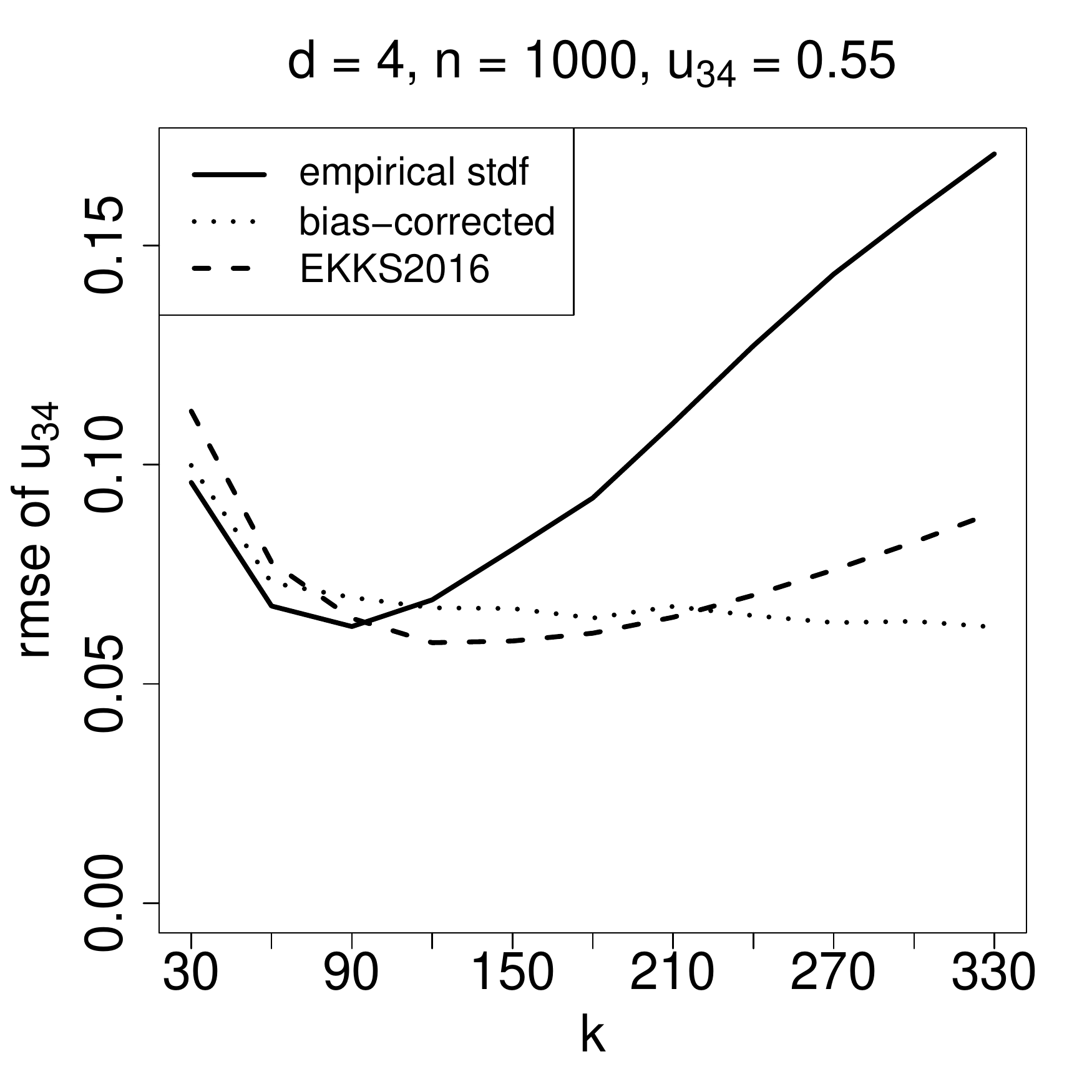}} 
\caption{Max-linear structural equation model based on a directed acylic graph: bias, standard deviation and RMSE for the estimators; $300$ samples of size $n = 1000$.}
\label{fig:maxlin2}
\end{figure}

\begin{rem}
For the weight matrix, we actually defined $\Omega(\theta)$ as $( \Sigma(\theta) + c I_q )^{-1}$ for some small $c > 0$. The reason for applying such a Tikhonov correction is that some eigenvalues of $\Sigma(\theta)$ are (near) zero, which can in turn be due to the fact that for max-linear models such as here, $\ell( c_m; \theta )$ may hit its lower bound $\max(c_{m,1}, \ldots, c_{m,d})$ for some $m \in \{1, \ldots, q\}$. 
\end{rem}

\subsection{Goodness-of-fit test}
We compare the performance of the goodness-of-fit test presented in Corollary~\ref{cor:GoF:Mopt} to the three goodness-of-fit test statistics $\kappa_n$, $\omega_n^2$, and $A_n^2$ proposed in \citet[page 18]{can2015}. In the simulation study there, the observed rejection frequencies are reported at the $5 \%$ significance level under null and alternative hypotheses for two bivariate models for $\ell$; a bivariate logistic model with $\theta \in (0,1)$ and 
\begin{align*}
\ell(x_1,x_2 ; \psi) = (1 - \psi) (x_1 + x_2) + \psi \sqrt{x_1^2 + x_2^2}, \qquad \psi \in (0,1),
\end{align*}  
i.e., a mixture between the logistic model and tail independence. For both models, they generate 300 samples of size $n = 1500$ from a ``null hypothesis" distribution function, for which the model is correct, and 100 samples of $n = 1500$ from an ``alternative hypothesis" distribution function, for which the model is incorrect. These distribution functions are described in equations (32), (33),  (35), and (36) of \citet{can2015}. We take $c_m \in \{(1/2,1/2), (1/2,1), (1,1/2), (1,1)\}$, $m =1,\ldots,4$, and $k = 200$.

Table \ref{tab:GoF} shows the observed fractions of Type I errors under the null hypotheses and the observed fraction of rejections under the alternative hypotheses. The results for $\kappa_n$, $\omega_n^2$, and $A_n^2$ are taken from \citet[Table 1]{can2015}. We see that our goodness-of-fit test performs comparably to the test statistics in \citet{can2015}.

\begin{table}[ht]
\centering
\begin{tabular}{lcccc}
\multicolumn{1}{c}{\text{}} & \multicolumn{2}{c}{Null} & \multicolumn{2}{c}{Alternative} \\
\cmidrule(r){2-3}
\cmidrule(r){4-5}
\text{ } & logistic  & mixture  & logistic & mixture \\
\midrule
$\kappa_n$ & 19/300 & 9/300 & 92/100 & 97/100 \\
$\omega_n^2$ & 11/300 & 13/300 & 90/100 & 97/100 \\
$A_n^2$ & 17/300 & 18/300 & 95/100 & 100/100 \\
$k f_{n,k} (\widehat{\theta}_{n,k})$ & 16/300 & 14/300 & 100/100 & 82/100 \\
\bottomrule
\end{tabular}
  \caption{Observed rejection frequencies at the $5 \%$ significance level under null and alternative hypotheses.}
  \label{tab:GoF}
\end{table}

It should be noted that the tests are of very different nature. The three test statistics in \citet{can2015} are functionals of a transformed empirical process and are therefore of omnibus-type. The results in there are based on the full max-domain of attraction condition on $F$ and the procedure is computationally complicated and therefore difficult to apply in dimensions (much) higher than two. The present test only performs comparisons at $q$ points and avoids integration. Therefore it is computationally much easier to apply in dimension $d > 2$.

\section{Tail dependence in European stock markets}
\label{sec:application}

We analyze data from the EURO STOXX 50 Index, which represents the performance of the largest 50 companies among 19 different ``supersectors" within the 12 main Eurozone countries. Since Germany (DE) and France (FR) together form $68 \%$ of the index, we will focus on these two countries only. Every company belongs to a supersector, of which there are 19 in total. We select two of them as an illustration: chemicals and insurance. We study the following five stocks: Bayer (DE, chemicals), BASF (DE, chemicals), Allianz (DE, insurance), Airliquide (FR, chemicals), and Axa (FR, insurance), and we take the weekly negative log-returns of the stock prices of these companies from Yahoo Finance\footnote{\url{http://finance.yahoo.com/}} for the period January 2002 to November 2015, leading to a sample of size $n=711$.
 
We fit a structural equation model based on the directed acyclic graph given in Figure~\ref{fig:app1}. The nodes \textsf{DE} and \textsf{FR} are represented by their national stock market indices, the DAX and the CAC40, respectively, and the nodes \textsf{chemicals} and \textsf{insurance} are represented by corresponding sub-indices of the EURO STOXX 50 Index. Note that this is a model for the tail dependence function only, i.e., we only assume that the joint distribution of the negative log-returns has tail dependence function $\ell$ as in \eqref{eq:mlstdf} with coefficient matrix $B$ given in Table~\ref{tab:Bmatrix}. We have $d = 10$ and the parameter vector is given by $\theta = (u_{12}, u_{13}, u_{14}, u_{15}, u_{26}, u_{46}, u_{27}, u_{47}, u_{38}, u_{48}, u_{39}, u_{59}, u_{2,10}, u_{5,10})$.

We perform the goodness-of-fit test described in Corollary~\ref{cor:GoF:eigen}, based on the $q = 1140$ points $c_m$ in the grid $\{0, 1/2, 1\}^8$ having either two or three non-zero coordinates. We take $\Omega (\theta) = I_q$, $k = 40$, and we choose $s$ such that $\nu_s > 0.1$, leading in this case to $s = 11$. The value of the test statistic is $5.28$; the $95 \% $ quantile of a $\chi^2_{11}$ distribution is $19.68$, so that the tail dependence model is not rejected.

The resulting parameter estimates are pictured at the edges of Figure~\ref{fig:app1}, where the relative width of each edge is proportional to its parameter value. The standard errors are given in parentheses. We note that, except for Allianz, the influence of the stock market indices DAX and CAC40 is (much) stronger than the influence of the sector indices chemicals and insurance.

\begin{figure}[ht]
\centering
\includegraphics[width=0.7\textwidth]{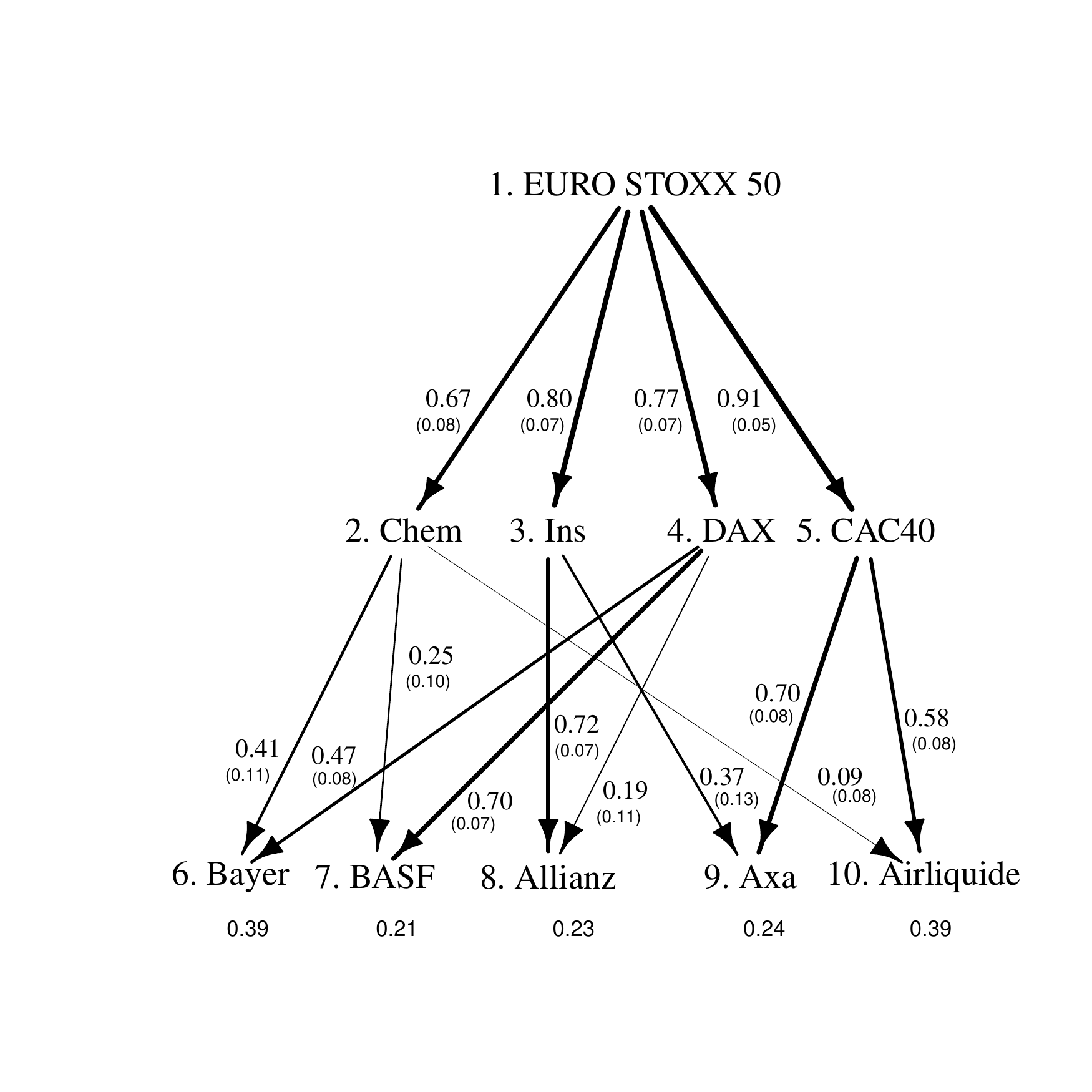}
\caption{European stock market data: directed acyclic graph with 14 parameters, whose estimates are shown near the corresponding edges. The relative width of each edge is proportional to its parameter value. The bottom row shows the estimated diagonal elements $u_6,\ldots,u_{10}$ of the matrix $B$ in Table \ref{tab:Bmatrix}.}
\label{fig:app1}
\end{figure}

\begin{table}[ht]
\begin{equation*}
B = \begin{pmatrix}
1 & 0 & 0 & 0 & 0 & 0 & 0 & 0 & 0 & 0 \\
u_{12} & u_2 & 0 & 0 & 0 & 0 & 0 & 0 & 0 & 0 \\
u_{13} & 0 & u_3 & 0 & 0 & 0 & 0 & 0 & 0 & 0 \\
u_{14} & 0 & 0 & u_4 & 0 & 0 & 0 & 0 & 0 & 0 \\
u_{15} & 0 & 0 & 0 & u_5 & 0 & 0 & 0 & 0 & 0 \\
u_{12} u_{26} \vee u_{14} u_{46} & u_2 u_{26} & 0 & u_4 u_{46} & 0 & u_6 & 0 & 0 & 0 & 0 \\
u_{12} u_{27} \vee u_{14} u_{47} & u_2 u_{27} & 0 & u_4 u_{47} & 0 & 0 & u_7 & 0 & 0 & 0 \\
u_{13} u_{38} \vee u_{14} u_{48} & 0 & u_3 u_{38} & u_4 u_{48} & 0 & 0 & 0 & u_8 & 0 & 0 \\
u_{13} u_{39} \vee u_{15} u_{59} & 0 & u_3 u_{39} & 0 & u_5 u_{59} & 0 & 0 & 0 & u_9 & 0 \\
u_{12} u_{2,10} \vee u_{15} u_{5,10} & u_2 u_{2,10} & 0 & 0 & u_5 u_{5,10} & 0 & 0 & 0 & 0 & u_{10}
\end{pmatrix}
\end{equation*}
\caption{European stock market data: coefficient matrix of the max-linear model stemming from the directed acyclic graph in Figure~\ref{fig:app1}. The diagonal elements $u_i$, for $i = 2,\ldots,10$ are such that the rows sum up to one.}
\label{tab:Bmatrix}
\end{table}

\appendix
\section{Proofs}\label{sec:proofs}
\begin{proof}[Proof of Theorem~\ref{thm:consistency}]
This proof follows the same lines as the one of \citet[Proof of Theorem 1]{einmahl2016}.
Let $\varepsilon_0 > 0$ be such that the closed ball $B_{\varepsilon_0} (\theta_0) = \{\theta : \left\Vert \theta - \theta_0 \right\Vert \leq \varepsilon_0 \}$ is a subset of $\Theta$; such an $\varepsilon_0$ exists since $\theta_0$ is an interior point of $\Theta$. Fix $\varepsilon > 0$ such that $0 < \varepsilon \leq \varepsilon_0$. Let, more precisely than in \eqref{eq:estimator}, $\Htheta$ be the \textit{set} of minimizers of the right-hand side of \eqref{eq:estimator}.
We show first that
\begin{equation}
\label{eq:consistency}
  \PP[
    \Htheta \neq \emptyset
    \text{ and }
    \Htheta \subset B_{\varepsilon} (\theta_0)
  ]
  \rightarrow 1, \qquad n \to \infty.
\end{equation}
Because $L$ is a homeomorphism, there exists $\delta >0$ such that for $\theta \in \Theta$, $\left\Vert L(\theta) - L(\theta_0) \right\Vert \leq \delta$ implies
$ \left\Vert \theta - \theta_0 \right\Vert \leq \varepsilon$.
Equivalently, for every $\theta \in \Theta$ such that $\left\Vert \theta - \theta_0 \right\Vert > \varepsilon$ we have $\left\Vert L(\theta) - L(\theta_0) \right\Vert > \delta$. Define the event
\begin{equation*}
A_n = \left\{ \Vert  L(\theta_0) - \widehat{L}_{n.k} \Vert < \frac{\delta\sqrt{\lambda_1}}{(1+\sqrt{\lambda_1})\max(1,\sqrt{\lambda_q(\theta_0)})}\right\}.
\end{equation*}
If $\theta \in \Theta$ is such that $\left\Vert \theta - \theta_0 \right\Vert > \varepsilon$, then on the event $A_n$, we have
\begin{align*}
\Vert D_{n,k} (\theta) \Vert_{\Omega(\theta)} & \geq
 \sqrt{\lambda_1(\theta)}
\Vert D_{n,k} (\theta) \Vert \\
& \geq  \sqrt{\lambda_1}  \Vert L(\theta_0) - L(\theta) - \left(  L (\theta_0) - \widehat{L}_{n,k} \right)   \Vert \\
& \geq \sqrt{\lambda_1} \left( \Vert L(\theta_0) - L (\theta) \Vert  - \Vert L(\theta_0) - \widehat{L}_{n,k} 	\Vert\right) \\
& >   \sqrt{\lambda_{1}}\left(\delta-\frac{\delta\sqrt{\lambda_1}}{1+\sqrt{\lambda_1}}\right)=\frac{\delta\sqrt{\lambda_1}}{1+\sqrt{\lambda_1}}.
\end{align*}
It follows that on $A_n$,
\begin{align*}
  \inf_{\theta: \Vert \theta - \theta_0 \Vert > \varepsilon}
  \Vert D_{n,k} (\theta) \Vert_{\Omega(\theta)}
  &\geq \frac{\delta\sqrt{\lambda_1}}{1+\sqrt{\lambda_1}} > \sqrt{\lambda_q(\theta_0)}\Vert L(\theta_0) - \widehat{L}_{n,k} \Vert \\ &\geq \Vert L(\theta_0) - \widehat{L}_{n,k} \Vert_{\Omega(\theta_0)}
  \geq \inf_{\theta: \Vert \theta - \theta_0 \Vert \leq \varepsilon}
  \Vert L(\theta) - \widehat{L}_{n,k} \Vert_{\Omega(\theta)}.
\end{align*}
The infimum on the right-hand side is actually a minimum since $L$ is continuous and $B_{\varepsilon} (\theta_0) $ is compact.
Hence on $A_n$ the set $\Htheta$ is non-empty and $\Htheta \subset B_{\varepsilon} (\theta_0)$.
To show \eqref{eq:consistency}, it remains to prove that $\mathbb{P}[A_n] \to 1$ as $n \to \infty$, but this follows from \eqref{ltol}.

Next we will prove that, with probability tending to one, $\Htheta$ has exactly one element, i.e., the function $f_{n,k}$ has a unique minimizer. To do so, we will show that there exists $\varepsilon_1 \in (0, \varepsilon_0]$ such that, with probability tending to one, the Hessian of $f_{n,k}$ is positive definite on $B_{\varepsilon_1}(\theta_0)$ and thus $f_{n,k}$ is strictly convex on $B_{\varepsilon_1}(\theta_0)$. In combination with \eqref{eq:consistency} for  $\varepsilon \in (0,\varepsilon_1]$, this will yield the desired conclusion.

For $\theta \in \Theta$, define the symmetric $p \times p$ matrix $\Hessian(\theta; \theta_0)$ by
\begin{align*}
  \left( \Hessian(\theta; \theta_0) \right)_{i,j} :=
&  \, 2  \left( \frac{\partial L (\theta)}{\partial \theta_j} \right)^T \Omega(\theta) \left( \frac{\partial L (\theta)}{\partial \theta_i}  \right) - 2 \left( \frac{\partial^2 L (\theta)}{\partial \theta_j \partial \theta_i}  \right)^T \, \Omega(\theta) \,
  \bigl( L(\theta_0) - L (\theta) \bigr) \\
  &  - 2 \left( \frac{\partial L (\theta)}{ \partial \theta_i}  \right)^T \, \frac{\partial\Omega(\theta)}{\partial \theta_j} \,
  \bigl( L(\theta_0) - L (\theta) \bigr)
  - 2 \left( \frac{\partial L (\theta)}{ \partial \theta_j}  \right)^T \, \frac{\partial\Omega(\theta)}{\partial \theta_i} \,
  \bigl( L(\theta_0) - L (\theta) \bigr) \\
  & +   \bigl( L(\theta_0) - L (\theta) \bigr)^T \,   \frac{\partial^2\Omega(\theta)}{\partial \theta_j\partial \theta_i}   \,
  \bigl( L(\theta_0) - L (\theta) \bigr),
\end{align*}
for $i, j \in \{1, \ldots, p\}$. The map $\theta \mapsto \Hessian(\theta; \theta_0)$ is continuous and
\begin{equation}\label{H00}
    \Hessian(\theta_0; \theta_0)
  = 2 \, \dot{L}(\theta_0)^T \, \Omega(\theta_0) \, \dot{L}(\theta_0),
\end{equation}
is a positive definite matrix. This $p \times p$ matrix is non-singular, since the $q \times q$ matrix $\Omega(\theta_0)$ is non-singular and the $q \times p$ matrix $\dot{L}(\theta_0)$ has rank $p$ (recall $q \ge p$). Let $\lVert \, \cdot \, \rVert$ denote the spectral  norm of a matrix. From Weyl's perturbation theorem
\citep[page~145]{jiang2010}, there exists an $\eta > 0$ such that every symmetric matrix $A \in \RR^{p \times p}$ with $\norm{A - \Hessian(\theta_0; \theta_0)} \leq \eta$ has positive eigenvalues and is therefore positive definite. Let $\varepsilon_1 \in (0, \varepsilon_0]$ be sufficiently small such that the second-order partial derivatives of $L$ and $\Omega$ are continuous on $B_{\varepsilon_1}(\theta_0)$ and such that $\norm{ \Hessian(\theta; \theta_0) - \Hessian(\theta_0; \theta_0)} \le \eta/2$ for all $\theta \in B_{\varepsilon_1}(\theta_0)$.

Let $\Hessian_{n,k,\Omega} (\theta) \in \mathbb{R}^{p \times p}$ denote the Hessian matrix of $f_{n,k}$. Its $(i,j)$-th element is
\begin{align*}
\bigl( \mathcal{H}_{n,k,\Omega} (\theta) \bigr)_{ij}
& = \frac{\partial^2}{\partial \theta_j \partial \theta_i}
\left[ D_{n,k} (\theta)^T \, \Omega(\theta) \, D_{n,k} (\theta) \right]\\
& = \frac{\partial}{\partial \theta_j}
\left[ - 2 D_{n,k} (\theta)^T \, \Omega(\theta) \frac{\partial L (\theta)}{\partial \theta_i} + D_{n,k} (\theta)^T \frac{\partial\Omega(\theta)}{\partial \theta_i} D_{n,k} (\theta) \right] \\
& = 2  \left( \frac{\partial L (\theta)}{\partial \theta_j} \right)^T \Omega(\theta) \left( \frac{\partial L (\theta)}{\partial \theta_i}  \right) - 2 \left( \frac{\partial^2 L (\theta)}{\partial \theta_j \partial \theta_i}  \right)^T \Omega(\theta) \,  D_{n,k} (\theta)\\
&\quad - 2 \left( \frac{\partial L (\theta)}{ \partial \theta_i}  \right)^T \, \frac{\partial\Omega(\theta)}{\partial \theta_j} \,
 D_{n,k} (\theta)
  - 2 \left( \frac{\partial L (\theta)}{ \partial \theta_j}  \right)^T \, \frac{\partial\Omega(\theta)}{\partial \theta_i} \,
 D_{n,k} (\theta)
  \\
  &\quad +   D_{n,k} (\theta)^T \,   \frac{\partial^2\Omega(\theta)}{\partial \theta_j\partial \theta_i}   \,
  D_{n,k} (\theta).
\end{align*}
Since $D_{n,k}(\theta) = \widehat{L}_{n,k} - L(\theta)$ and since $\widehat{L}_{n,k}$ converges in probability to $L(\theta_0)$, we obtain
\begin{equation}\label{eq:unicon}
  \sup_{\theta \in B_{\varepsilon_1}(\theta_0)}
  \norm{ \Hessian_{n,k,\Omega}(\theta) - \Hessian(\theta; \theta_0) }
  \pto 0,
  \qquad n \to \infty.
\end{equation}
By the triangle inequality, it follows that
\begin{equation}\label{conH}
  \Pr \biggl[ \sup_{\theta \in B_{\varepsilon_1}(\theta_0)}
  \norm{ \Hessian_{n,k,\Omega}(\theta) - \Hessian(\theta_0; \theta_0) } \leq \eta \biggr]
  \to 1,
  \qquad n \to \infty.
\end{equation}
In view of our choice for $\eta$, this implies that, with probability tending to one, $\Hessian_{n,k}(\theta)$ is positive definite for all $\theta \in B_{\varepsilon_1}(\theta_0)$, as required.
\end{proof}

\begin{proof}[Proof of Theorem~\ref{thm:an}]
Let $\nabla f_{n,k}(\theta)$, a $1 \times q$ vector, be the gradient of $f_{n,k}$ at $\theta$. By \eqref{eq:nonparcond}, we have
\begin{align}
\nonumber
  \sqrt{k} \, \nabla f_{n,k} (\theta_0)
  &= 
  - 2 \sqrt{k} \, D_{n,k} (\theta_0)^T \, \Omega(\theta_0) \, \dot{L} (\theta_0) 
  +
  \sqrt{k} D_{n,k} (\theta_0)^T \bigl(\nabla \Omega(\theta)|_{\theta=\theta_0} \bigr) D_{n,k} (\theta_0) \\
  &=
  - 2 \sqrt{k} \, D_{n,k} (\theta_0)^T \, \Omega(\theta_0) \, \dot{L} (\theta_0) + o_P(1),
  \qquad \text{as $n \to \infty$.}
\label{eq:nablafnk:expansion}
\end{align}
Since $\htheta$ is a minimizer of $f_{n,k}$, we have $\nabla f_{n,k} (\htheta) = 0$. An application of the mean value theorem to the function $t \mapsto \nabla f_{n,k} \bigl( \theta_0 + t (\htheta - \theta_0) \bigr)$ at $t = 0$ and $t = 1$ yields
\begin{equation}
\label{eq:nablafnk:zero}
  0
  = 
  \nabla f_{n,k}  (\htheta)^T 
  = 
  \nabla f_{n,k}  (\theta_0)^T 
  + 
  \Hessian_{n,k,\Omega}(\widetilde{\theta}_{n,k}) \, (\htheta - \theta_0),
\end{equation}
where $\widetilde{\theta}_{n,k}$ is a random vector on the segment connecting $\theta_0$ and $\htheta$ and $\Hessian_{n,k,\Omega}$ is the Hessian matrix of $f_{n,k}$ as in the proof of Theorem~\ref{thm:consistency}. Since $\htheta \pto \theta_0$, we have $\widetilde{\theta}_{n,k} \pto \theta_0$ as $n \to \infty$ too. By \eqref{eq:unicon} and \eqref{H00} and continuity of $\theta \mapsto \Hessian(\theta; \theta_0)$, it then follows that 
\begin{equation}
\label{eq:Hessiannk:limit}
  \mathcal{H}_{n,k,\Omega} (\widetilde{\theta}_{n,k}) 
  \pto \mathcal{H} (\theta_0; \theta_0) 
  = 2 \dot{L}(\theta_0)^T \, \Omega(\theta_0) \, \dot{L}(\theta_0),
  \qquad \text{as $n \to \infty$}.
\end{equation}
Since $\mathcal{H} (\theta_0 ; \theta_0)$ is non-singular, the matrix $\mathcal{H}_{n,k,\Omega} (\widetilde{\theta}_{n,k})$ is non-singular with probability tending to one as well.  Combine equations~\eqref{eq:nablafnk:expansion}, \eqref{eq:nablafnk:zero} and \eqref{eq:Hessiannk:limit} to see that
\begin{align*}
  \sqrt{n} \bigl( \htheta - \theta_0 \bigr)
  &= 
  - \Hessian_{n,k,\Omega}(\widetilde{\theta}_{n,k})^{-1}
  \, \sqrt{k} \, \nabla f_{n,k}(\theta_0)^T + o_p(1) \\
  &=
  \bigl( \dot{L}(\theta_0)^T \Omega(\theta_0) \dot{L}(\theta_0) \bigr)^{-1} \, 
  \dot{L}(\theta_0)^T \Omega(\theta_0) \, 
  \sqrt{k} \, D_{n,k}(\theta_0) + o_p(1),
  \qquad \text{as $n \to \infty$.}
\end{align*}
Convergence in distribution to the stated normal distribution follows from \eqref{eq:nonparcond} and Slutsky's lemma.
\end{proof}

\begin{proof}[Proof of Corollary~\ref{cor:GoF}]
Since $D_{n,k}( \theta ) = \widehat{L}_{n,k} - L(\theta)$, we have
\[
  \sqrt{k} \, D_{n,k}( \htheta )
  =
  \sqrt{k} \, D_{n,k}( \theta_0 )
  -
  \sqrt{k} \bigl( L( \htheta ) - L( \theta_0 ) \bigr).
\]
By \eqref{eq:theta:an} and the delta method, we have
\begin{align*}
  \sqrt{k} \bigl( L( \htheta ) - L( \theta_0 ) \bigr)
  &=
  \dot{L} \, \sqrt{k} ( \htheta - \theta_0 ) + o_p(1) \\
  &=
  \dot{L} \, ( \dot{L}^T \Omega \dot{L} )^{-1} \dot{L}^T \Omega \, \sqrt{k} \, D_{n,k}(\theta_0)
  + o_p(1) \\
  &=
  P(\theta_0) \, \sqrt{k} \, D_{n,k}( \theta_0 ) + o_p(1),
  \qquad \text{as $n \to \infty$,}
\end{align*}
where $\dot{L}$ and $\Omega$ are evaluated at $\theta_0$. Combination of the two previous displays yields
\[
  \sqrt{k} \, D_{n,k}( \htheta )
  =
  (I_q - P(\theta_0)) \, \sqrt{k} \, D_{n,k}(\theta_0) + o_p(1),
  \qquad \text{as $n \to \infty$.}
\]
By \eqref{eq:nonparcond} and Slutsky's lemma, we arrive at \eqref{eq:Dnk:GoF}, as required. 

The $q \times q$ matrix $P$ has rank $p$ since the $q \times p$ matrix $\dot{L}$ has rank $p$ and the $q \times q$ matrix $\Omega$ is non-singular. 
Since $P^2 = P$, it follows that rank$(I_q - P) = $ rank$(I_q) - $rank$(P) = q-p.$ 
\end{proof}

\begin{proof}[Proof of Corollary~\ref{cor:GoF:Mopt}]
Equation \eqref{eq:nonparcond} can be written as
\[
  Z_{n,k} 
  := 
  \sqrt{k} \, D_{n,k}(\theta_0) \dto Z 
  \sim 
  \Normal_q( 0, \Sigma(\theta_0) \bigr), 
  \qquad \text{as $n \to \infty$.}
\]
In view of \eqref{eq:Dnk:GoF} and $\Omega(\theta) = \Sigma(\theta)^{-1}$, we find, by Slutsky's lemma and the continuous mapping theorem,
\begin{align*}
  k \, f_{n,k}( \htheta )
  &=
  k \, D_{n,k}( \htheta )^T \, \Sigma( \htheta )^{-1} \, D_{n,k}( \htheta ) \\
  &=
  Z_{n,k}^T \, 
  (I_q - P(\theta_0))^T \, \Sigma( \htheta )^{-1} \, (I_q - P(\theta_0)) \,
  Z_{n,k}
  + o_p(1) \\
  &\dto
  Z^T \, (I_q - P(\theta_0))^T \, \Sigma( \theta_0 )^{-1} \, (I_q - P(\theta_0)) \, Z,
  \qquad \text{as $n \to \infty$;}
\end{align*}
here $P = \dot{L} \, (\dot{L}^T \Sigma^{-1} \dot{L})^{-1} \, \dot{L}^T \Sigma^{-1}$, with $\dot{L}$ and $\Sigma$ evaluated at $\theta_0$.

It remains to identify the distribution of the limit random variable. The random vector $Z$ is equal in distribution to $\Sigma^{1/2} Y$, where $Y \sim \Normal_q(0, I_q)$ and where $\Sigma^{1/2}$ is a symmetric square root of $\Sigma$. Straightforward calculation yields
\begin{equation*}
  Z^T (I_q - P)^T \, \Sigma^{-1} \, (I_q - P) \, Z \eqd Y^T (I_q - B) Y
\end{equation*}
where
$
  B = \Sigma^{-1/2} \dot{L} \, (\dot{L}^T \Sigma^{-1} \dot{L})^{-1} \, \dot{L}^T \Sigma^{-1/2}.
$
It is easily checked that $B$ is a projection matrix ($B = B^T = B^2$). Moreover, $B$ has rank $p$. It follows that $I_q - B$ is a projection matrix too and that it has rank $q-p$. The  distribution of the limit random variable now follows by standard properties of quadratic forms of normal random vectors. 
\end{proof}

\begin{proof}[Proof of Corollary~\ref{cor:GoF:eigen}]
Let $Z \sim \mathcal{N}_q( 0, \Sigma( \theta_0 ) )$, which by \eqref{eq:nonparcond} is the limit in distribution of $\sqrt{k} \, D_{n,k}( \theta_0 )$. By \eqref{eq:Dnk:GoF} and the continuous mapping theorem, we have, as $n \to \infty$,
\begin{equation}
\label{eq:Tnk}
  k \, D_{n,k}( \htheta )^T \, A( \htheta ) \, D_{n,k}( \htheta )
  \dto
  Z^T \, (I_q - P(\theta_0))^T \, A(\theta_0) \, (I_q - P(\theta_0)) \, Z.
\end{equation}
We can represent $(I_q - P) Z$ as $V D^{1/2} Y$, with $Y \sim \Normal_q( 0, I_q )$. The limiting random variable in \eqref{eq:Tnk} is then given by
\begin{equation*}
  Y^T D^{1/2} V^T \, V_s D_s^{-1} V_s^T V D^{1/2} Y.
\end{equation*}
Since $V$ is an orthogonal matrix, this expression simplifies to $\sum_{j=1}^s Y_j^2$, which has the stated $\chi_s^2$ distribution.
\end{proof}

\begin{proof}[Proof of Remark~\ref{rem:spq}]
Inspection of the proofs of Corollaries~\ref{cor:GoF:Mopt} and~\ref{cor:GoF:eigen} shows that the difference between the two test statistics converges in distribution to the random variable $Z^T \, R(\theta_0) \, Z$, where $Z$ is a certain $q$-variate normal random vector and where
\[
  R(\theta_0) 
  = 
  \bigl( I_q - P(\theta_0) \bigr)^T \, (\Sigma(\theta_0)^{-1} - A(\theta_0) \bigr) \, \bigl( I_q - P(\theta_0) \bigr).
\]
The matrix $R(\theta_0)$ can be shown to be equal to zero, proving the claim of the remark. To see why $R(\theta_0)$ is zero, note first that, suppressing $\theta_0$ and writing $Q = I_q - P$, we have $Q^2 = Q$ and $\Sigma Q^T = Q \Sigma = Q \Sigma Q^T$. Recall the eigenvalue equation $Q \Sigma Q^T v_j = \nu_j v_j$ for $j = 1, \ldots, q$. Note that $\nu_j > 0$ if $j \le s$ and $\nu_j = 0$ if $j \ge s+1$. The eigenvalue equation implies that $Qv_j = v_j$ for $j \le s$ while $Q \Sigma v_j = 0$ for $j \ge s+1$. Since the vectors $v_1, \ldots, v_q$ are orthogonal, we find that the vectors $v_1, \ldots, v_s, \Sigma v_{s+1}, \ldots, \Sigma v_q$ are linearly independent. It then suffices to show that $R v_j = 0$ for all $j \le s$ and $R \Sigma v_j = 0$ for all $j \ge s+1$. The first property follows from the fact that $\Sigma^{-1} v_j = \nu_j^{-1} Q^T v_j$ and $A v_j = \nu_j^{-1} v_j$ for $j \le s$ (use the eigenvalue equation again), while the second property follows from $Q \Sigma v_j = 0$ for $j \ge s+1$.
\end{proof}

\bibliographystyle{chicago}
\bibliography{lib-CUETD}

\end{document}